\PassOptionsToPackage{capitalise}{cleveref}
\PassOptionsToPackage{draft}{todonotes}
\PassOptionsToPackage{svgnames,dvipsnames}{xcolor}
\documentclass[copyright,creativecommons]{eptcs}
\providecommand{\event}{ICE 2022} 

\newif\iftr\trtrue    %%% Comment next line for the TR
\iftr\trfalse         %%% Uncomment this line for the camera ready
 
\usepackage[utf8]{inputenc}
\usepackage{xargs}
\usepackage{xspace}
\usepackage{url}
\usepackage[nomargin,multiuser,inline,draft]{fixme} 
\fxusetheme{color}
\usepackage{environ}                % Provides NewEnviron w/ eager expansion
\usepackage{etex,etoolbox}          % Needed by theorem-appendix.tex
\usepackage{graphicx}               % For \scalebox etc.

\usepackage{nicefrac}
\usepackage{xifthen}                % for conditional commands
\usepackage{appendix}               % for correct numeration of statements

\usepackage{stmaryrd}
\usepackage{pdfsync}

\PassOptionsToPackage{usenames,dvipsnames,svgnames,table}{xcolor}
\usepackage{xcolor} % must be loaded first
\usepackage{tikz}                   % Graphs
\usetikzlibrary{shadows,arrows,shapes,automata,positioning,decorations.markings}

\usepackage{amsthm}
\usepackage{amsmath,amssymb,amsbsy}
\usepackage{mathtools} % \xRightArrow and stuff - load late to avoid conflicts
\usepackage{thmtools, thm-restate}  % Advanced theorem handling
\newtheorem{definition}{Definition}[section]
\newtheorem{lemma}[definition]{Lemma}

\newtheorem{proposition}[definition]{Proposition}
\newtheorem{example}[definition]{Example}
\newtheorem{fact}[definition]{Fact}
\usepackage{apxproof}

\usepackage{hyperref}
\hypersetup{
  colorlinks = true,
  linkcolor = purple,
  urlcolor  = purple,
  citecolor = teal,
  anchorcolor = teal
}
\usepackage[capitalise]{cleveref}
\crefformat{enumi}{condition~#2#1#3}
\crefname{fact}{Fact}{Facts}
\Crefname{fact}{Fact}{Facts}
\crefformat{fact}{Fact~#2#1#3}
\DeclareMathSizes{10}{10}{6}{6}

\tikzset{
  cnode/.style={
    shape=circle,
    minimum size = 0mm,
    inner sep = 1pt,
    font=\tiny,
    draw
  },
  carrow/.style={
    ->,
    shorten >=1pt,
    >=stealth',
    auto,
    draw,
    sloped
  }
}

\usepackage[utf8]{inputenc}
\usepackage{dashbox}
\newcommand\dboxed[1]{
  \dbox{\ensuremath{#1}}
}

\newcommand{\csys}[2][c]{
	 \dbox{\ensuremath{
		\begin{array}{#1}
		  #2
		\end{array}
	 }
  }
}

\newif\ifemi
\DeclareGraphicsExtensions{%
    .png,.PNG,%
    .pdf,.PDF,%
    .jpg,.mps,.jpeg,.jbig2,.jb2,.JPG,.JPEG,.JBIG2,.JB2}

\newcommand{\tnxbehapi}[1][partly]{
  Research {#1} supported by the EU
  H2020 RISE programme under the Marie Skłodowska-Curie grant
  agreement No 778233,
}

\newcommand{\tnxitmatters}{
  by the MIUR project PRIN
  2017FTXR7S \quo{IT-MaTTerS} (Methods and Tools for Trustworthy Smart
  Systems)
}
\newcommand{\tnxunict}{
  and by the Progetto di Ateneo \quo{Piaceri}.
}
\newcommand{\tnxgncs}{
  The authors have also been partially supported by 
INdAM as members of GNCS (Gruppo Nazionale per il Calcolo Scientifico).
}
  
\usepackage{bm}
\usepackage{xspace}
\fxusetheme{color}
\FXRegisterAuthor{eM}{aeM}{\color{orange} {\underline{eM}}}

\usepackage[normalem]{ulem} % underline command breaks over line ends

\usepackage{xifthen}        % for conditional commands
\newcommand{\ifempty}[3]{%
  \ifthenelse{\isempty{#1}}{#2}{#3}%
}

\newcommand{\mkfun}[4][\colorFun]{
  \newcommand{#2}[1][#4]{
    {#1\textsf{#3}}
    \ifempty{##1}{}{
      ({##1})}
  }
}

\newcommand{\mkuop}[4][\colorFun]{
  \newcommand{#2}[1][#4]{
    {#1\textsf{#3}}
    \ifempty{##1}{}{
      \, {##1}}
  }
}

\newcommand{\hidden}[1]{}

\newcommand{\cf}[2]{
  \fontsize{#1}{#1}{\selectfont{#2}}
}
\ifemi
\usepackage{showlabels}

\newcommand{\emi}[2]{
  \marginpar{\fcolorbox{red}{shadecolor}{\cf{#1}{{#2}}}}
}
\newcommand{\emic}[2]{\par
  \fcolorbox{red}{shadecolor}{\parbox{\linewidth}{ 
      \color{gray}
      \begin{description}
      \item[{\color{blue} #2}]{\sf #1}
      \end{description}}}
}
\else
\newcommand{\emi}[2]{}
\newcommand{\emic}[2]{}{}
\fi

\newcommand{\sst}{\;\big|\;}
\newcommand{\dom}[1]{\operatorname{dom}({#1})}

\newcommand{\conf}[1]{\ensuremath{\langle {#1} \rangle}}

\newcommand{\qqand}[1][and]{\qquad\text{#1}\qquad}
\newcommand{\qand}[1][and]{\quad\text{#1}\quad}

\newcommand{\upd}[3]{{#1}[{#2} \mapsto {#3}]}

\newcommand{\quo}[1]{\lq\lq {#1}\rq\rq}
\def\finex{{\unskip\nobreak\hfil
\penalty50\hskip1em\null\nobreak\hfil{\Large $\diamond$}
\parfillskip=0pt\finalhyphendemerits=0\endgraf}}

\definecolor{shadecolor}{rgb}{1,0.99,0.9}
\definecolor{bg}{rgb}{0.95,0.95,0.95}

\def \bfr {\begin{color}{blue}} 
\def \efr {\end{color}}

\def \biv {\begin{color}{purple}} 
\def \eiv {\end{color}}

\def \beM {\begin{color}{orange}} 
\def \eeM {\end{color}}

\newcommand{\prestrict}[2]{{#1}\mid_{#2}}

\newcommand{\cat}{{\cdot}}
\newcommand{\tstates}[1]{\mathsf{T}(#1)}

\mkfun{\ptpof}{ptp}{{\gint[]}}

\newcommand{\cproj}[2]{%
    \ifstrequal{#2}{\HH}%
        {_{#2}\rfloor{#1}}%
        {\ifstrequal{#2}{\KK}{{#1}\lfloor_{#2}}{\PackageError{cproj}{Undefined option: #2}{}}}%
		}

\newcommandx{\nof}[3][1={\HH,\KK}, 2=\aCM, usedefault=@]{\mathsf{nof}_{#1}\ifempty{#3}{}{(#2, #3)}}

\newcommand{\Dual}[1]{\overline{#1}}

\newcommand{\comp}{{\asymp}}

\newcommand{\ssem}[1]{{\llbracket #1 \rrbracket}}

\newcommand{\gateway}[2]{\mathsf{gw}(#1,#2)}

\newcommand{\connect}[2]{\stackrel{\hspace{-3pt}#1\!\leftrightarrow\!#2\!}{}}

\usepackage{mfirstuc}

\newcommandx{\gspider}[2][1=T_1,2=T_2,usedefault=@]{\text{spider}(\ifempty{#1}{\_}{#1}, \ifempty{#2}{\_}{#2})}
\newcommandx{\glink}[3][1=\aG,2=T_1,3=T_2,usedefault=@]{\text{link}(\ifempty{#1}{\_}{#1}, \ifempty{#2}{\_}{#2},  \ifempty{#3}{\_}{#3})}
\newcommandx{\gintcomp}[2][1=I,2=J,usedefault=@]{\stackrel{\p[#1],\p[#2]}{\bowtie}}

\newcommand{\fillcolor}{orange!5}
\def\colorPtp{\color{blue}}
\def\colorFun{\color{NavyBlue}}

\def\colorOp{\color{OliveGreen}}
\def\colorNode{\color{cyan}}
\def\colorR{\color{OliveGreen}}
\def\colorE{\color{orange}}

\def\colorMsg{\color{BrickRed}}

\newcommand{\msg}[1][m]{\mathsf{\colorMsg{#1}}}
\newcommand{\msgset}{\mathcal{\colorMsg M}}

\newcommand{\lset}{\mathcal{L}}
\newcommand{\lint}{\lset_{\text{int}}}
\newcommand{\lact}{\lset_{\text{act}}}
\newcommand{\lio}{\lset_{\text{i/o}}}

\newcommand{\sset}{\mathcal{S}}

\newcommand{\ptp}[1][A]{\ensuremath{\mathsf{\colorPtp{\capitalisewords{#1}}}}}

\newcommand{\p}{\ptp}
\newcommand{\q}{{\ptp[B]}}

\newcommandx{\ggcommon}[3][1=\ptp,2={\aH},3={\aH'},usedefault=@]{f_{#1}}
\newcommandx{\opair}[2][1={\ae},2={\ae'},usedefault=@]{\conf{{#1},{#2}}}
\newcommandx{\hopair}[2][1={\aE},2={\aE'},usedefault=@]{\llparenthesis\, {#1},{#2}\, \rrparenthesis}
\newcommandx{\wf}[2][1={\aG},2={\aG'},usedefault=@]{wf({#1}, {#2})}
\newcommandx{\wb}[2][1={\aG},2={\aG'},usedefault=@]{wb({#1}, {#2})}
\newcommandx{\ws}[2][1={\aG},2={\aG'},usedefault=@]{ws({#1}, {#2})}
\newcommandx{\widx}[2][1={\aW},2={i},usedefault=@]{{#1}[{#2}]}
\newcommandx{\outop}[2][1=\gname,2={}]{{\colorOp{!}}^{{#1}{#2}}}
\newcommandx{\inop}[2][1=\gname,2={}]{{\colorOp{?}}^{{#1}{#2}}}
\newcommandx{\aout}[5][1={\p},2={\q},3={},4=m,5={},usedefault=@]{
  \achan[#1][#2] \outop[{#3}] {\msg[#4]}{#5}
}
\newcommandx{\ain}[5][1={\p},2={\q},3={},4=m,5={},usedefault=@]{
  \achan[#1][#2] \inop[{#3}] {\msg[#4]}{#5}
}
\newcommandx{\adep}[1][1={}]{
  \conf{ \aout[@][@][@][@][{#1}], \ain[@][@][@][@][{#1}]}
}

\newcommandx{\hproj}[2][1=\aH, 2=\ptp, usedefault=@]{
  \ifempty{#1}{}{{#1}}\ifempty{#2}{}{{^{\scriptscriptstyle @{#2}}}}
}
\newcommandx{\eproj}[2][1=\aE,2=\ptp, usedefault=@]{
  {{#1}}\ifempty{#2}{}{{^{\scriptscriptstyle @{#2}}}}
}

\newcommandx{\mklooptwo}[4][1=.5,2=1.5]{ %does not insert the a gate below
  \node[ogate,above = #1 of {#3}] (entry#3) {};
  \pgfgetlastxy \xentry \yentry;
  \pgfmathtruncatemacro{\xentryrounded}{\xentry};
  \path (#4);
   \pgfgetlastxy \xexit \yexit;
  \pgfmathtruncatemacro{\xexitrounded}{\xexit};
  \path[line] (entry#3) -- (#3);
  \pgfmathsetmacro\tmpdiff{abs(\xentryrounded - \xexitrounded)}
  \path[line] (#4) -|  ($(#4)+(\tmpdiff,0)+(#2,0)$) |- (entry#3);
}

\newcommand{\apom}{r}

\newcommand{\aR}[1][R]{{\colorR{#1}}}

\newcommand{\aConf}{s}

\newcommandx{\detM}[1][1=\aCM,usedefault=@]{\Delta({#1})}
\newcommand{\gsubs}[2]{^{#1} / _{#2}}
\newcommandx{\gsubst}[3][1=\aM,2=q,3=q',usedefault=@]{
  \left \{\gsubs{#3}{#2} \right \}#1
}

\newcommand{\HH}{\mathsf{\color{blue}H}}
\newcommand{\KK}{\mathsf{\color{blue}K}}

\newcommand{\Set}[1]{\{\,#1\,\}}

\newcommandx{\cm}[2][1=\ptp, 2=\aM]{{#2}_{#1}}
\newcommandx{\achan}[2][1=A,2=B,usedefault=@]{{\ptp[#1]\,\ptp[#2]}}
\newcommand{\ptpset}{\mathcal{\colorPtp{P}}}

\newcommand{\pset}{\ptpset}

\newcommand{\oact}{\outop[]}
\newcommand{\iact}{\inop[]}
\newcommand{\tset}{\to}

\newcommandx{\cauttr}[5][1=p,2={\gint[][A][@][B]}, 3=q, 4=H, 5=K, usedefault=@]{
  {#1} \xrightarrow[{\ptp[#4]\ \ptp[#5]}]{#2} {#3}
}
\newcommand{\RS}[1][]{\mathcal{R}({#1})}

\newcommand{\trans}[2][{}]{\,\xrightarrow{#2}_{#1}\,}

\newcommandx{\acfsmout}[3][1=A,2=B,3=m,usedefault=@]{\achan[{#1}][{#2}] \oact {\msg[{#3}]}}
\newcommandx{\acfsmin}[3][1=A,2=B,3=m,usedefault=@]{\achan[{#1}][{#2}] \iact {\msg[{#3}]}}
\newcommandx{\fsaout}[2][1={\p},2={},usedefault=@]{
  \ptp[#1] \ \outop[]\ \msg[{#2}]
}
\newcommandx{\fsain}[2][1={\p},2={},usedefault=@]{
  \ptp[#1] \ \inop[]\ \msg[{#2}]
}

\makeatletter
\newcommand{\linenumfontsize}{\@setfontsize{\linenumfontsize}{3pt}{3pt}}
\makeatother
\newcommand{\aG}{\mathsf{G}}

\newcommand{\gseqop}{{\colorOp ;}\,}
\newcommand{\gparop}{{\colorOp \ |\ }}
\newcommand{\gchoop}{{\colorOp \ +\ }}
\newcommand{\grecop}{{\colorOp *}}
\newcommand{\grecopp}{{\colorOp{@}}}
\newcommand{\gname}[1][i]{{\colorNode{\scriptstyle\textsf{#1}}}}
\newcommandx{\nmerge}[2][1={i},2={},usedefault=@]{
  \ifempty{#2}{
    \ifempty{#1}{\mu}{\gname[-{#1}]}
  }{-{#2}}
}

\mkfun{\esbj}{sbj}{\ae}

\makeatletter%
\@ifclassloaded{exam-paper}%
  {}%
  {\makeatletter%
    \@ifclassloaded{test}%
    {}%
    \makeatother%
  }
\makeatother%

\newcommandx{\gnode}[2][1=i,2=\gint,usedefault=@]{
  {\ifempty{#1}{}{\colorNode{\gname[{#1}].}}} {#2}
}

\newcommandx{\gint}[4][1={i},2=A,3=\msg,4=B,usedefault=@]{
  \ptp[#2] {\colorOp \xrightarrow{\scriptscriptstyle\gname[#1]}} \ptp[#4] \colon {\msg[{#3}]}
}
\newcommandx{\gout}[4][1=\gname,2=\ptp,3=\msg,4={\ptp[C]},usedefault=@]{
  \achan[{#2}][{#4}] {\colorOp {\colorOp{!}}} {\msg[{#3}]}
}
\newcommandx{\gin}[4][1={},2=\ptp,3=\msg,4={\ptp[C]},usedefault=@]{
  \achan[{#2}][{#4}] {\colorOp {\colorOp{?}}} {\msg[{#3}]}
}
\newcommandx{\gseq}[3][1=i,2={\aG},3={\aG'},usedefault=@]{
  \gnode[{#1}][{#2} \gseqop {#3}]
}
\newcommandx{\gpar}[3][1={},2={\aG},3={\aG'},usedefault=@]{
  \gnode[{#1}][\ifempty{#1}{{#2} \gparop {#3}}{({#2} \gparop {#3})}]
}
\newcommandx{\gcho}[3][1={},2={\aG},3={\aG'},usedefault=@]{
  \gnode[{#1}][\ifempty{#1}{{#2} \gchoop {#3}}{\big({#2} \gchoop {#3}\big)}]
}
\newcommandx{\gchov}[3][1={},2={\aG},3={\aG'},usedefault=@]{
  \gnode[{#1}][\left(
  \begin{array}l
    \ifempty{#1}{{#2} \\ \gchoop \\ {#3}}{\!\!{#2} \\ \gchoop \\ {#3}}
  \end{array}\right)
  ]
}
\newcommandx{\grec}[3][1={},2={\aG},3={\p},usedefault=@]{
  \gnode[{#1}][\ifempty{#1}{\grecop {#2} \grecopp {#3}}{\big(\grecop {#2} \grecopp {#3}\big)}]
}

\newcommand{\getcentroid}[2]{
    \coordinate (tmpgatecoord) at (0,0);
    \foreach \n [count=\i] in {#1}{
      \path (\n);
      \coordinate (tmpgatecoord) at ($(tmpgatecoord) + (\n)$);
      \coordinate (#2) at ($1/\i*(tmpgatecoord)$);
    }
}

\tikzset{
  src/.style={draw,circle,fill=white,
    minimum size=2mm,
    inner sep=0pt},
  sink/.style={draw,circle,double,fill=white,
    minimum size=1.5mm,
    inner sep=0pt},
  node/.style={draw,circle,fill=black,
    minimum size=2mm,
    inner sep=0pt},
  block/.style = {rectangle, draw=gray, align=center, fill=orange!25, rounded corners=0.1cm,
    minimum size=5mm, inner sep=2pt},
  prenode/.style = {minimum size=9pt,inner sep=2pt, font=\Large},
  bblock/.style = {rectangle, draw=blue!50, opacity=.5, line width=1pt, align=center, fill=white, rounded corners=0.1cm,
    minimum size=7mm, inner sep=2pt},
  prenode/.style = {minimum size=9pt,inner sep=2pt, font=\Large},
  agate/.style={draw, rectangle,
    minimum size=3mm,
    inner sep=0pt,
    fill=orange!25,
    postaction={path picture={% 
        \draw[red]
        ([yshift=\gatedistanceinand]path picture bounding box.south) --
        ([yshift=-\gatedistanceinand]path picture bounding box.north) ;}}},
  ogate/.style = {
    diamond, draw, fill=orange!25,
    minimum size=4mm,
    inner sep=0pt,
    postaction={path picture={% 
        \draw[red]
        ([yshift=\gatedistancein]path picture bounding box.south) -- ([yshift=-\gatedistancein]path picture bounding box.north)
        ([xshift=-\gatedistancein]path picture bounding box.east) -- ([xshift=\gatedistancein]path picture bounding box.west)
        ;}}},
  anygate/.style = {circle, draw, fill=white,
    minimum size=4mm,
    inner sep=0pt,
    postaction={path picture={% 
        \draw[black]
        ([xshift=-\gatedistancein,yshift=\gatedistancein]path picture bounding box.south east) --
        ([xshift=\gatedistancein,yshift=-\gatedistancein]path picture bounding box.north west)
        ([xshift=-\gatedistancein,yshift=-\gatedistancein]path picture bounding box.north east) --
        ([xshift=\gatedistancein,yshift=\gatedistancein]path picture bounding box.south west)
        ;}}},
  elli/.style = {draw,densely dotted,-},
  line/.style = {draw,->, rounded corners=0.07cm,>=latex},
  nline/.style = {draw,semithick, ->},
  pline/.style = {draw,->,>=latex},
  node distance=1cm and 0.7cm,
  baseline=(current  bounding  box.center),
}

\tikzset{
  graphaxiom/.style={
    node distance=0.7 and 0.7cm,
    scale=.5,
    transform shape
  }
}

\tikzset{
  hgsem/.style={
    draw,
    node distance=2cm and 1cm,
    transform shape,
    smooth,
    every node/.style = {font=\sffamily\bfseries}
  }
}

\tikzset{
  hgstyle/.style={
    src color={#1},
    tgt color={#1},
    centroid color={#1},
    centroid label={#1},
    centroid name={#1},
    centroid radius={#1},
    centroid ratio={#1},
    xoffset={#1},
    yoffset={#1},
    xsrcoffset={#1},
    ysrcoffset={#1},
    xtgtoffset={#1},
    ytgtoffset={#1},
    font={#1},
    centroid angle={#1},
    centroid tolerance={#1}
  },
  src color/.store in = \hgsrccol,
  tgt color/.store in = \hgtgtcol,
  centroid color/.store in =\hgfillcolor,
  centroid label/.store in =\hglabel,
  centroid name/.store in =\hgname,
  centroid radius/.store in = \hgradius,
  centroid ratio/.store in = \hgratio,
  xoffset/.store in =\hgxoffset,
  yoffset/.store in =\hgyoffset,
  xsrcoffset/.store in =\hgxsrcoffset,
  ysrcoffset/.store in =\hgysrcoffset,
  xtgtoffset/.store in =\hgxtgtoffset,
  ytgtoffset/.store in =\hgytgtoffset,
  centroid angle/.store in =\hgangle,
  centroid tolerance/.store in =\hgtolerance,
  src color = black,
  tgt color = black,
  centroid color = orange!40,
  centroid label={},
  centroid name={dummycentroid},
  centroid radius = .7pt,
  centroid ratio = .35,
  xoffset = 0,
  yoffset = 0,
  xsrcoffset = 0,
  ysrcoffset = 0,
  xtgtoffset = 0,
  ytgtoffset = 0,
  font=\sffamily\scriptsize,
  centroid angle=0,
  centroid tolerance=10pt
}

\newcommandx{\mkhg}[5][1={},4={},5={},usedefault=@]{
  \begingroup
  \tikzset{#1}
  \StrCount{#2,}{,
  }[\l] % from package xstring
  \StrCount{#3,}{,}[\m] % from package xstring
  \ifthenelse{\l = 1 \AND \m = 1}{
    \ifempty{#4}{
      \ifempty{#5}{
        \path[hgsem, ->, >=stealth', shorten >=1pt] (#2) -- (#3);
      }{
        \path[hgsem, ->, >=stealth', shorten >=1pt] (#2) #5 (#3);
      }
    }{
      \ifempty{#5}{
        \path[hgsem, ->, >=stealth', shorten >=1pt, #4] (#2) -- (#3);
      }{
        \path[hgsem, ->, >=stealth', shorten >=1pt, #4] (#2) #5 (#3);
      }
    }
  }{
    \coordinate (srcoffset) at (\hgxsrcoffset,\hgysrcoffset);
    \coordinate (tgtoffset) at (\hgxtgtoffset,\hgytgtoffset);
    \getcentroid{#2}{srccentroid};
    \getcentroid{#3}{tgtcentroid};
    \node[label={left:\hglabel}] (\hgname) at ($(srccentroid)!{1-\hgratio}!\hgangle:(tgtcentroid) + (\hgxoffset,\hgyoffset)$) {};
    \pgfgetlastxy \xc \yc;
    \pgfmathtruncatemacro{\xcontrol}{\xc};
    \pgfmathtruncatemacro{\ycontrol}{\yc};
    \foreach \n in {#2}{
      \path (\n);
      \pgfgetlastxy \xntmp \yntmp;
      \pgfmathtruncatemacro{\xn}{\xntmp};
      \pgfmathtruncatemacro{\yn}{\yntmp};
      \pgfmathsetmacro\xtmpdiff{abs(\xn - \xcontrol + \hgxsrcoffset)};
      \pgfmathsetmacro\ytmpdiff{abs(\yn - \ycontrol + \hgytgtoffset)};
      \ifdim \xtmpdiff pt > \hgtolerance
      \ifempty{#4}{
        \path[hgsem, \hgsrccol] (\n) .. controls ($(srccentroid.center) + (srcoffset)$) .. (\hgname.center);
      }{
        \path[hgsem, \hgsrccol] (\n) .. controls ($(srccentroid.center) + (srcoffset)$) .. (\hgname.center);
      }
      \else
      \ifempty{#4}{
        \path[hgsem, \hgsrccol] (\n) -- (\hgname.center);
      }{
        \path[hgsem, \hgsrccol, #4] (\n) -- (\hgname.center);
      }
      \fi
    }
    \foreach \n in {#3}{
      \path (\n);
      \pgfgetlastxy \xntmp \yntmp;
      \pgfmathtruncatemacro{\xn}{\xntmp};
      \pgfmathtruncatemacro{\yn}{\yntmp};
      \pgfmathsetmacro\xtmpdiff{abs(\xn - \xcontrol)};
      \pgfmathsetmacro\ytmpdiff{abs(\yn - \ycontrol)};
      \ifdim \xtmpdiff pt > \hgtolerance
      \ifempty{#4}{
        \path[hgsem, ->, >=stealth', shorten >=1pt, \hgtgtcol] (\hgname.center) .. controls (tgtcentroid.center) and ($(tgtcentroid.center) + (tgtoffset)$) .. (\n);
      }{
        \path[hgsem, ->, >=stealth', shorten >=1pt, \hgtgtcol,#4] (\hgname.center) .. controls (tgtcentroid.center) and ($(tgtcentroid.center) + (tgtoffset)$) .. (\n);
      }
      \else
      \ifempty{#4}{
        \path[hgsem, ->, >=stealth', shorten >=1pt, \hgtgtcol] (\hgname.center) --  (\n);
      }{
        \path[hgsem, ->, >=stealth', shorten >=1pt, \hgtgtcol] (\hgname.center) --  (\n);
      }
      \fi
    }
    \fill[\hgfillcolor] (\hgname) circle [radius=\hgradius];
  }
  \endgroup
}

\newcommandx{\hgordeq}[1][1={\aH},usedefault=@]{\sqsubseteq_{#1}}
\newcommandx{\gintsem}[4][4=.5]{
  \tikz[hgsem,scale=#4,every node/.style={font=\scriptsize}]{
    \node (out) {$\aout[{#1}][{#2}][][{\msg[{#3}]}]$};
    \node[below = 20pt of out] (in) {$\ain[{#1}][{#2}][][{\msg[{#3}]}]$};
    \mkhg{out}{in};
  }
}

\newcommandx{\gsem}[2][1={\aG},2={},usedefault=@]{[\![ {#1} ]\!]_{#2}}
\newcommandx{\rbot}{\text{undef}}
\newcommandx{\rtrs}[1][1={\aH},usedefault=@]{{#1}^{\star}}
\newcommandx{\gord}[1][1={\aG},usedefault=@]{\leq_{#1}}
\newcommandx{\gordeq}[1][1={\aG},usedefault=@]{\leq_{#1}}
\mkfun{\cause}{cs}{}
\mkfun{\effect}{ef}{}
\mkfun{\participants}{ptps}{}

\newcommandx{\aW}{w}
\newcommandx{\rlang}{\mathcal{L}}
\newcommand{\gfun}[1]{\ensuremath{\mathsf{\colorFun #1}}}
\mkfun{\eact}{\gfun{act}}{}
\mkfun{\enode}{\gfun{cp}}{}
\mkfun{\cp}{\gfun{cp}}{}

\mkuop{\rmax}{\gfun{max}}{\aH}
\mkuop{\rmin}{\gfun{min}}{\aH}
\mkuop{\rMAX}{\gfun{lst}}{\aH}
\mkuop{\rMIN}{\gfun{fst}}{\aH}

\newcommandx{\rseq}[2][1=\aG,2={\aG'},usedefault=@]{\gfun{seq}({#1},{#2})}
\newcommandx{\rpar}[2][1=\aG,2={\aG'},usedefault=@]{\gfun{par}({#1},{#2})}

\newcommandx{\gproj}[2][1=\aG,2=\ptp]{{#1}\!\downarrow_{#2}}
\newcommandx{\gprojPN}[2][1=\aG,2=\ptp]{{#1}\!\downarrow_{#2}^{\mathsf{P\hspace{-0.5pt}N}}}
\newcommandx{\cinit}[1][1={\aQzero},usedefault=@]{{#1}}
\newcommandx{\cfinal}[1][1={q_e},usedefault=@]{{#1}}

\newcommandx{\geproj}[4][1=\aG,2=\ptp,3=\cinit,4=\cfinal,usedefault=@]{
  {#1}\downarrow_{#2}^{{#3},{#4}}
}

\newcommand*{\StrikeThruDistance}{0.15cm}%
\tikzset{strike thru arrow/.style={
    decoration={markings, mark=at position 0.5 with {
        \draw [blue, thick,-] 
            ++ (-\StrikeThruDistance,-\StrikeThruDistance) 
            -- ( \StrikeThruDistance, \StrikeThruDistance);}
    },
    postaction={decorate},
}}

\newcommandx{\ich}[1][1={\aG},usedefault=@]{{#1}^{\oplus}}
\newcommandx{\ichedges}[2][1={\aG},2={\gname},usedefault=@]{{#1}^{\oplus}({#2})}
\newcommandx{\parts}[1]{2^{#1}}
\newcommandx{\actch}{c}
\newcommandx{\soundactch}[2][1={\aG},2={\actch},usedefault=@]{{#1} \,\circledR\, {#2}}
\newcommandx{\rOnActch}[2][1={\aG},2={\actch},usedefault=@]{{#1} \setminus {#2}}
\newcommandx{\rOnActchClean}[2][1={\aG},2={\actch},usedefault=@]{{#1} \circledR {#2}}
\newcommandx{\rAllEvents}[1][1={\aG},usedefault=@]{\mathit{dom}(#1)}

\newcommand{\AV}{\mathcal{V}}
\newcommand{\aH}{H}

\newcommandx{\hgvertex}[2][1=\al,2=\gname,usedefault=@]{{#1}_{\textcolor{red}{[{#2}]}}}
\newcommand{\aE}{{\colorE E}}
\renewcommand{\ae}[1][e]{{\colorE{#1}}}

\newcommand{\al}[1][\lambda]{{\colorE{#1}}}
\newcommandx{\hyedge}[1]{\{#1\}}

\newcommandx{\rdiv}[2][1=\gcho,2=\ptp,usedefault=@]{
  \gfun{div}_{#2}(#1)
}

\newcommandx{\rrdiv}[5][1={\aG},2={\aG'},3={\AV},4={\AV'},5=\ptp,usedefault=@]{
  \gfun{div}^{#3,#4}_{#5}(#1,#2)
}
\newcommandx{\pdiv}[3][1={\apom_1},2={\apom_2},3={\apom},usedefault=@]{
  \gfun{div}_{#3}(#1,#2)
}
\newcommandx{\pfork}[3][1={\apom_1},2={\apom_2},3={\apom},usedefault=@]{
  \gfun{fork}_{#3}(#1,#2)
}

\newcommandx{\mkint}[6][3=i,4=\p,5=\msg,6=\q,usedefault=@]{
  \node[bblock,{#1}] (#2) {$\gint[#3][#4][#5][#6]$};
}

\newcommandx{\mkgraph}[3][1=.5cm]{
  \node[source,above = #1 of {#2}] (src#2) {};
  \node[sink,below  = #1 of {#3}] (sink#3) {};
  \path[line] (src#2) -- (#2);
  \path[line] (#3) -- (sink#3);
}

\newcommandx{\mkloop}[4][1=.5,2=1.5]{
  \node[ogate,above = #1 of {#3}] (entry#3) {};
  \pgfgetlastxy \xentry \yentry;
  \pgfmathtruncatemacro{\xentryrounded}{\xentry};
  \node[ogate,below  = #1 of {#4}] (exit#4) {};
  \pgfgetlastxy \xexit \yexit;
  \pgfmathtruncatemacro{\xexitrounded}{\xexit};
  \path[line] (entry#3) -- (#3);
  \path[line] (#4) -- (exit#4);
  \pgfmathsetmacro\tmpdiff{abs(\xentryrounded - \xexitrounded)}
  \path[line] (exit#4) -|  ($(exit#4)+(\tmpdiff,0)+(#2,0)$) |- (entry#3);
}

\newcommandx{\mkfork}[4][2=gatenode,3=i,4=.6,usedefault=@]{
  \mkgatebegin{#1}[{\gname[#3]}][agate][#4]{#2}
}

\newcommandx{\mkbranch}[4][2=gatenode,3=i,4=.6,usedefault=@]{
  \mkgatebegin{#1}[{\gname[#3]}][ogate][#4]{#2}
}

\newcommandx{\mkgatebegin}[5][2={},3=ogate,4=.5]{
  \coordinate (gatecord) at (0,0);
  \foreach \n [count=\i] in {#1}{
    \pgfgetlastxy \xc \yc;
    \path (\n);
    \pgfgetlastxy \xn \yn;
    \coordinate (gatecord) at ($(gatecord) + (\xn,0)$);
    \coordinate (gatecord) at ($1/\i*(gatecord)$);
    \ifdim \yn < \yc
    \node (max) at (0,\yc) {};
    \else
    \node (max) at (0,\yn) {};
    \fi
  }
  \coordinate (gatecord) at ($(gatecord) + (0,#4) + (max)$);
  \node[#3,label={below:$#2$}] (#5) at (gatecord) {};
  \pgfgetlastxy{\xgate}{\ygate};
  \pgfmathtruncatemacro{\xgateround}{\xgate};
  \StrCount{#1,}{,}[\l] % from package xxstring
  \ifnum \l < 2 {\errmessage{#1 argument should be a comma-separated list of lenght >= 2}}
  \else{
    \foreach \n in {#1}{
      \path (\n);
      \pgfgetlastxy{\xnode}{\ynode};
      \pgfmathtruncatemacro{\xnround}{\xnode};
      \pgfmathsetmacro\tmpdiff{abs(\xnround - \xgateround)}
      \ifdim \tmpdiff pt > 1 pt \path[line] (#5) -| (\n);
      \else
        \path[line] (#5) -- (\n);
      \fi
    }
  }
  \fi
}

\newcommandx{\gorthopath}[4][3=|-,4=-2pt,usedefault=@]{
  \path #1;
  \pgfgetlastxy{\xabove}{\yabove};
  \pgfmathtruncatemacro{\xaboveround}{\xabove};
  \path #2;
  \pgfgetlastxy{\xbelow}{\ybelow};
  \pgfmathtruncatemacro{\xnround}{\xbelow};
  \pgfmathsetmacro\tmpdiff{abs(\xnround - \xaboveround)}
  \ifdim \tmpdiff pt > 1 pt \path[line] #2 #3 #1;
  \else
  \path[line] #2 -- #1;
  \fi
}

\newcommandx{\mkmerge}[4][2=gatenode,3=i,4=0,usedefault=@]{\mkgateend{#1}[{\ifempty{#3}{}{\nmerge[#3]}}][ogate][#4]{#2}}

\newcommandx{\mkjoin}[4][2=gatenode,3=i,4=0,usedefault=@]{\mkgateend{#1}[{\ifempty{#3}{}{\nmerge[#3]}}][agate][#4]{#2}}

\newcommandx{\mkgateend}[5][2={},3=ogate,4=.5,usedefault=@]{
  \coordinate (gatecord) at (0,0);
  \foreach \n [count=\i] in {#1}{
    \pgfgetlastxy \xc \yc;
    \path (\n);
    \pgfgetlastxy \xn \yn;
    \coordinate (gatecord) at ($(gatecord) + (\xn,0)$);
    \coordinate (gatecord) at ($1/\i*(gatecord)$);
    \ifdim \yn > \yc
    \node (min) at (0,\yc) {};
    \else
    \node (min) at (0,\yn) {};
    \fi
  }
  \coordinate (gatecord) at ($(gatecord) - (0,#4) + (min)$);
  \node[#3,label={above:$#2$}] (#5) at (gatecord) {};
  \pgfgetlastxy{\xgate}{\ygate};
  \pgfmathtruncatemacro{\xgateround}{\xgate};
  \StrCount{#1,}{,}[\l] % from package xxstring
  \ifnum \l < 2 {\errmessage{#1 argument should be a comma-separated list of lenght >= 2}}
  \else{
    \foreach \n in {#1}{
      \path (\n);
      \pgfgetlastxy{\xnode}{\ynode};
      \pgfmathtruncatemacro{\xnround}{\xnode};
      \pgfmathsetmacro\tmpdiff{abs(\xnround - \xgateround)}
      \ifdim \tmpdiff pt > 1 pt \path[line] (\n) |- (#5);
      \else
        \path[line] (\n) -- (#5);
      \fi
    }
  }
  \fi
}

\newcommand{\gatedistancein}{3pt}
\newcommand{\gatedistanceinand}{2pt}

\usetikzlibrary{
  arrows,
  backgrounds,
  chains,
  calc,
  decorations.markings,decorations.pathreplacing,
  fadings,
  fit,
  patterns,
  petri,
  positioning,
  shadows,
  shapes,automata,shapes.callouts
}

\tikzset{
  src/.style={draw,circle,fill=white,
    minimum size=2mm,
    inner sep=0pt
  },
  sink/.style={draw,circle,double,fill=white,
    minimum size=1.5mm,
    inner sep=0pt
  },
  node/.style={draw,circle,fill=black,
    minimum size=2mm,
    inner sep=0pt
  },
  source/.style={draw,circle,fill=white,
    minimum size=3mm,
    inner sep=0pt
  },
  sink/.style={draw,circle,double,fill=white,
    minimum size=3mm,
    inner sep=0pt
  },
  block/.style = {rectangle, draw=gray, align=center, fill=orange!25, rounded corners=0.1cm,
    minimum size=5mm, inner sep=2pt},
  prenode/.style = {minimum size=9pt,inner sep=2pt, font=\Large},
  bblock/.style = {rectangle, draw=blue!50, opacity=.7, line width=.5pt, align=center, fill=white, rounded corners=0.1cm,
    minimum size=4mm, inner sep=1pt},
  prenode/.style = {minimum size=9pt,inner sep=2pt, font=\Large},
  agate/.style={draw, rectangle,
    minimum size=3mm,
    inner sep=0pt,
    fill=orange!25,
    label={[red]center:$\mid$}
  },
  ogate/.style = {
    diamond, draw, fill=orange!25,
    minimum size=4mm,
    inner sep=0pt,
    label={[red]center:$+$}
  },
  lgate/.style = {
    diamond, draw, fill=orange!25,
    minimum size=4mm,
    inner sep=0pt,
    label={[red]center:$\circlearrowleft$}
    },
  altogate/.style = {
    diamond, draw,
    minimum size=4mm,
    inner sep=0pt,
    postaction={path picture={% 
        \draw
        ([yshift=\gatedistancein]path picture bounding box.south) -- ([yshift=-\gatedistancein]path picture bounding box.north)
        ([xshift=-\gatedistancein]path picture bounding box.east) -- ([xshift=\gatedistancein]path picture bounding box.west)
        ;}}},
  altgate/.style={draw, rectangle,
    minimum size=3mm,
    inner sep=0pt,
    postaction={path picture={% 
        \draw
        ([yshift=\gatedistanceinand]path picture bounding box.south) --
        ([yshift=-\gatedistanceinand]path picture bounding box.north) ;}}},
  anygate/.style = {circle, draw, fill=white,
    minimum size=4mm,
    inner sep=0pt,
    postaction={path picture={% 
        \draw[black]
        ([xshift=-\gatedistancein,yshift=\gatedistancein]path picture bounding box.south east) --
        ([xshift=\gatedistancein,yshift=-\gatedistancein]path picture bounding box.north west)
        ([xshift=-\gatedistancein,yshift=-\gatedistancein]path picture bounding box.north east) --
        ([xshift=\gatedistancein,yshift=\gatedistancein]path picture bounding box.south west)
        ;}}
  },
  smallglobal/.style={
        node distance=1cm and 0.8cm, semithick, scale=0.8, every node/.style={transform shape}
  },
  elli/.style = {draw,densely dotted,-},
  line/.style = {draw,->, rounded corners=0.07cm,>=latex},
  nline/.style = {draw,semithick, ->},
  pline/.style = {draw,->,>=latex},
  node distance=1cm and 0.7cm,
  baseline=(current  bounding  box.center),
  local/.style={rectangle, draw, fill=\fillcolor, drop shadow,
    text centered, rounded corners, minimum height=5em
  },
  bigar/.style={
    draw,very thick, ->
  },
  process/.style={rectangle, draw=gray, fill=\fillcolor, drop shadow,
    text centered, minimum height=5em,text=gray
  },
  choreo/.style={rectangle, draw, fill=\fillcolor, drop shadow,
    text centered, rounded corners, minimum height=5em
  },
  mycfsm/.style={
        font=\footnotesize,
        initial where=above,
        ->,>=stealth,auto,
		  node distance=1.5cm,
        scale=.85,
		  every node/.style={transform shape},
        every state/.style={cnode, inner sep=1pt, transform shape},
		  every edge/.style={carrow},
        baseline=(current  bounding  box.center),
        initial text={}
  },
  machinecloud/.style={
    cloud, cloud puffs=10, cloud ignores aspect, minimum height=.1cm, minimum width=2cm, draw
  },
  fitting node/.style={
    inner sep=0pt,
    fill=none,
    draw=none,
    reset transform,
    fit={(\pgf@pathminx,\pgf@pathminy) (\pgf@pathmaxx,\pgf@pathmaxy)}
  },
  mypetri/.style={
    font=\footnotesize,
    baseline=(current  bounding  box.center)
  },
  silentrans/.style = {rectangle, draw=black, align=center, fill=black,
    minimum height=1pt,
    minimum width=15pt,
    inner sep=1.5pt
  },
  reset transform/.code={\pgftransformreset},
  tmtape/.style={draw,minimum size=1.2cm}
}

\newcommand{\gunlessop}{\mbox{\colorOp\tiny\tt unless}}

\newcommandx{\gtry}[5][1=\gname,2={\aG_1 \gchoop \cdots \gchoop \aG_n},3=\gin,4=\gout,5={j},usedefault=@]{
  \def\foo{\gtryop\ {#2} \ \gcatchop\ {#3} {\colorOp \Rightarrow} {#4} {\colorOp \bullet} {\gname[{#5}]}}
  \gnode[{#1}][{\ifempty{#1} {\foo } { \big(\foo \big) }}]
}

\newcommandx{\gtrycatch}[4][1=\gname,2={\aG},3=\gin,4={\aG'},usedefault=@]{
  \def\foo{\gtryop\ {#2} \ \gcatchop\ {#3} \gdoop\ {#4}}
  \gnode[{#1}][{\ifempty{#1} {\foo} {\big( \foo \big) }}]
}

\def\colorGuard{\color{cyan}}
\newcommand{\aguard}{{\colorGuard \phi}}
\newcommandx{\agG}[2][1={\aG},2=\aguard]{{#1} \ifempty{#2}{}{\ \gunlessop\ {#2}}}

\newcommandx{\grcho}[4][1=\gname,2={\agG},3={\agG[\aG'][\aguard']},4={\cdots},usedefault=@]{
  \def\foo{{#2} {\ \ifempty{#4}{\gchoop}{\gchoop \ {#4}\  \gchoop}\ } {#3}}
  \gnode[{#1}][\ifempty{#1}{\foo}{\big( \foo \big)}]
}

\newcommandx{\ggprefix}[3][1=\ptp,2={\aR},3={\aR'},usedefault=@]{f_{#1}} % it was \newcommandx{\common}{...}
\newcommand{\aconfigfn}{\chi}
\newcommand{\aconfig}{\ell}

\newcommand{\lstates}{\statemap}
\newcommandx{\sysconfig}[3][1=\lstates,2=\aconfigfn,3={},usedefault=@]{
  \conf{ {#1},{#2} \ifempty{#3}{}{, #3} }
}
\newcommand{\sysctxfn}[1][]{\gamma_{#1}}
\newcommandx{\sysctx}[2][1=\aQ,2={},usedefault=@]{({#1},\sysctxfn[{#2}])}

\newcommandx{\alog}[4][1=\msg,2=q,3=\gname,4=t,usedefault=@]{\big({#1},{#2},{#3},{#4}\big)}

\newcommand{\aCM}{M}\newcommand{\aM}{\aCM}
\newcommand{\aQ}{Q}
\newcommandx{\aQzero}[1][1=,usedefault=@]{
  {\ifempty{#1}{q_0}{q_{0#1}}}
}
\newcommand{\badbranches}[1][]{\beta\ifempty{#1}{}{\big({#1}\big)}}
\newcommand{\aTrs}{\tset}
\newcommandx{\guardedaction}[2][1=\al,2=\aguard,usedefault=@]{
  {#1} \ifempty{#2}{}{/} {#2}
}
\newcommandx{\atrM}[4][1=q,2=\al,3={\hat q,\hat \al, \aguard},4=q',usedefault=@]{
  {#1} \xrightarrow[{#3}]{\guardedaction[{#2}][]} {{#4}}
}
\newcommandx{\atrS}[5][
  1={\sysconfig[@][@][\badbranches]},
  2=\al,
  3=\aguard,
  4={\sysconfig[\lstates'][\aconfigfn'][\badbranches]},
  5=\sysctx,usedefault=@
]{
  {#1} \xRightarrow{\qquad} {{#4}}
}
\newcommandx{\arevtrS}[2][
  1={\sysconfig[@][@][\badbranches]},
  2={\sysconfig[\lstates'][\aconfigfn'][\badbranches']},
  usedefault=@
]{
  {#1} \rightsquigarrow {#2}
}
\newcommand{\aCS}[1][S]{\mathsf{#1}}

\newcommandx{\enables}[2][1=\aconfigfn,2=\aguard,usedefault=@]{{#1} \vdash {#2}}
\newcommandx{\gprojfn}[5][1=\aG,2=\ptp,3=\cinit,4=\cfinal,5={},usedefault=@]{
  \mathbf{proj}_{#2}({#1},{#3},{#4}\ifempty{#5}{}{,{#5}})
}

\newcommandx{\rbp}[3][1=\aG,2=\aconfigfn,3=\achan,usedefault=@]{\mathtt{RBP}_{{#1},{#2}}\ifempty{#3}{}{\big({#3}\big)}}

\newcommand{\apseudoCFSM}{\mathtt{M}}
\newcommandx{\pseudoseq}[2][1=\apseudoCFSM,2=\apseudoCFSM',usedefault=@]{{#1}  ; {#2}}
\newcommandx{\pseudoCFSM}[4][1=\aQ,2=\aQzero,3=\cfinal,4=\aTrs,usedefault=@]{(#1 \ ; #2 \ ; #3 \ ; #4)}
\newcommandx{\markt}[3][1=\hat{\al},2=\hat{q},3=\aguard,usedefault=@]{\%\big({#1} , {#2}, {#3}\big)}

\newcommandx{\borderfn}[2][1=\aconfig,2=\aloop,usedefault=@]{
  \mathsf{border}_{{#2}}\ifempty{#1}{}{\big({#1}\big)}
}

\newcommandx{\ggvisually}[4][1=5pt,2=15pt,3=5pt,4=5pt,usedefault=@]{
  \def\dist{\hspace{1.0cm}}
  \tikzset{
    mycallout/.style={
      fill=green!10, opacity=.5, overlay, align=center,
      cloud callout, cloud puffs=15, aspect=1.9, cloud ignores aspect, cloud puff arc=100, shading=ball
    }
  }
  $\begin{array}{c@{\dist}c@{\dist}c@{\dist}c@{\dist}c}
     \begin{tikzpicture}[node distance=0.9cm and 0.4cm, every node/.style={scale=.7,transform shape}]
       \mkint{}{int}[]
       \mkgraph{int}{int};
       \node[mycallout, above = .3cm of srcint, xshift=1cm, callout absolute pointer={(srcint.east)}] {source node};
       \node[mycallout, below = .3cm of sinkint, xshift=-1cm, callout absolute pointer={(sinkint.west)}] {sink node};
     \end{tikzpicture}
     &
     \begin{tikzpicture}[node distance=.9cm and 0.4cm, every node/.style={scale=.7,transform shape}]
       \node[bblock] at (0,0) (g) {$\aG$};
       \node[node, below=of g] (s1) {};
       \node[bblock, below=of s1] (gp) {$\aG'$};
       \path[line,dotted] (g) -- (s1);
       \path[line,dotted] (s1) -- (gp);
     \end{tikzpicture}
     &
     \begin{tikzpicture}[node distance=.4cm and 0.4cm, every node/.style={scale=.7,transform shape}]
       \node[bblock] at (-.7,0) (g) {$\aG$};
       \node[bblock] at (.7,0)  (gp) {$\aG'$};
       \node[node, above=of g] (f) {};
       \node[node, below=of g] (j) {};
       \node[node, above=of gp] (fp) {};
       \node[node, below=of gp] (jp) {};
       \path[line,dotted] (f) -- (g);
       \path[line,dotted] (g) -- (j);
       \path[line,dotted] (fp) -- (gp);
       \path[line,dotted] (gp) -- (jp);
       \mkfork{f,fp}[fork][][#1];
       \mkjoin{j,jp}[join][][#2];
       \mkgraph{fork}{join};
       \node[mycallout, above = .3cm of fork, xshift=1cm, callout absolute pointer={(fork.east)}] {fork gate};
       \node[mycallout, above = -.9cm of join, xshift=-1cm, callout absolute pointer={(join.west)}] {join gate};
     \end{tikzpicture}
     &
     \begin{tikzpicture}[node distance=.4cm and 0.4cm, every node/.style={scale=.7,transform shape}]
       \node[bblock] at (-.7,0) (g) {$\aG$};
       \node[bblock] at (.7,0)  (gp) {$\aG'$};
       \node[node, above=of g] (f) {};
       \node[node, below=of g] (j) {};
       \node[node, above=of gp] (fp) {};
       \node[node, below=of gp] (jp) {};
       \path[line,dotted] (f) -- (g);
       \path[line,dotted] (g) -- (j);
       \path[line,dotted] (fp) -- (gp);
       \path[line,dotted] (gp) -- (jp);
       \mkbranch{f,fp}[fork][][#3];
       \mkmerge{j,jp}[join][][#4];
       \mkgraph{fork}{join};
       \node[mycallout, above = .3cm of fork, xshift=1cm, callout absolute pointer={(fork.east)}] {branch gate};
       \node[mycallout, above = -.9cm of join, xshift=-1cm, callout absolute pointer={(join.west)}] {merge gate};
     \end{tikzpicture}
     &
     \begin{tikzpicture}[node distance=0.4cm and 0.4cm, every node/.style={scale=.7,transform shape}]
       \node[bblock] (g) {$\aG$};
       \node[node, above=.5cm of g] (f) {};
       \node[node, below=.5cm of g] (j) {};
       \path[line,dotted] (f) -- (g);
       \path[line,dotted] (g) -- (j);
       \mkloop[.5][@]{f}{j}[];
       \mkgraph[.4cm]{entryf}{exitj};
       \node[mycallout, above = .2cm of entryf, xshift=1.3cm, callout absolute pointer={(entryf.east)}] {loop entry};
       \node[mycallout, above = -.7cm of exitj, xshift=-1.3cm, callout absolute pointer={(exitj.west)}] {loop exit};
     \end{tikzpicture}
     \\
     \text{\scriptsize interaction}
     &
     \text{\scriptsize sequential}
     &
     \text{\scriptsize parallel}
     &
     \text{\scriptsize branching}
     &
     \text{\scriptsize iteration}
   \end{array}$
}

\title{
  On Composing Communicating Systems\thanks{
	 \tnxbehapi \tnxitmatters \tnxunict \tnxgncs The authors thanks the reviewers for their helpful comments and also Mariangiola Dezani for her support.
  }
  \iftr \\{Full Version} \fi
}

\def\titlerunning{On Composing Communicating Systems}
\def\authorrunning{F. Barbanera, I. Lanese, E. Tuosto}

\author{
  Franco Barbanera\institute{Dept. of Mathematics and Computer Science, University of Catania (Italy)}
  \and
  Ivan Lanese \institute{Focus Team, University of Bologna/INRIA (Italy)
  }
  \and Emilio Tuosto \institute{Gran Sasso Science Institute (Italy)}
  }

\begin{document}

% \begin{keyword}
%   Choreography, automata, communicating finite-state machines,
%   message-passing, liveness, lock-freedom, deadlocok-freedom.
% \end{keyword}

\maketitle              % typeset the title of the cont$nimac0@ribution

\begin{abstract}
  Communication is an essential element of modern software, yet
  programming and analysing communicating systems are difficult tasks.
  A reason for this difficulty is the lack of compositional mechanisms
  that preserve relevant communication properties.

  This problem has been recently addressed for the well-known model of
  \emph{communicating systems}, that is sets of components consisting
  of finite-state machines capable of exchanging messages.
  The main idea of this approach is to take two systems, select a
  participant from each of them, and derive from those participants a
  pair of coupled gateways connecting the two systems.
  More precisely, a message directed to one of the gateways is
  forwarded to the gateway in the other system, which sends it to the
  other system.
  It has been shown that, under some suitable \emph{compatibility}
  conditions between gateways, this composition mechanism preserves
  deadlock freedom for asynchronous as well as symmetric synchronous
  communications (where sender and receiver play the same part in
  determining which message to exchange).

  This paper considers the case of \emph{asymmetric synchronous
	 communications} where senders decide independently which message
  should be exchanged.
  We show here that preservation of lock freedom requires
  sequentiality of gateways, while this is not needed for preservation of
  either deadlock freedom or strong lock freedom.
\end{abstract}

\section{Introduction}\label{sec:intro}
Communication is an essential constitutive element of modern software
due to the fact that applications are increasingly developed in
distributed architectures (e.g., service-oriented architectures,
microservices, cloud, etc.).
In practice, APIs and libraries featuring different communication
mechanisms are available for practically any programming language.

Reasoning on and developing communicating systems are difficult
endeavours.
Indeed, several models have been used to study interactions between
systems (e.g., process algebras, transition systems, Petri nets,
logical frameworks, etc.).
The so-called \emph{business logic}, necessary to determine
\emph{what} has to be communicated, needs to be complemented with the
so-called \emph{application level protocol} specifying \emph{how}
information spreads across a system.
Conceptual and programming errors may occur in the realisation of
application level protocols.
For instance, it may happen that some components in a system are
prevented to communicate because all the expected partners terminated
their execution (deadlock).
Other typical errors occur when a system is not lock-free, that is
when some components cannot progress because all their partners are
perpetually involved in other communications.
These kinds of problems can arise when a system can evolve in
different ways depending on some conditions and components have
inconsistent \quo{views} of the state of the system.
If this happens, some components may reach a state no longer
\quo{compatible} with the state of their partners and therefore
communications do not happen as expected.

We illustrate these problems with some simple examples for deadlock
freedom (similar examples may be given for lock freedom).
Suppose we want to model a client-server system where clients'
requests are acknowledged either with an answer or with \quo{unknown}
from servers.
Due to its popularity, we choose CCS~\cite{MilnerR:calcs} to
introduce this scenario, so take the following agents:
\begin{align}\label{eq:v1}
  \ptp[C] = \bar{\msg[r]}. \msg[a] + \bar{\msg[r]}. \msg[u]
  \quad\qqand\quad
  \ptp[D] = \msg[r]. \bar{\msg[a]} + \msg[r]. \bar{\msg[u]}
\end{align}
where ports $\msg[r]$, $\msg[a]$, and $\msg[u]$ are respectively used
to commmunicate requests, answers, and unknowns.
(Recall that in CCS $\msg[a].\_$ and $\_+\_$ represent respectively
prefix and non-deterministic choice.)
The common interpretation of agents in \eqref{eq:v1} is that $\msg[x]$
and $\bar{\msg[x]}$ respectively represent an input and an output on
port $\msg[x]$.
It is a simple observation that the system $\ptp[C] \mid \ptp[D]$
where $\ptp[C]$ and $\ptp[D]$ run in parallel can evolve to e.g., the
deadlock state $\msg[a] \mid \bar{\msg[u]}$ where each party is
waiting for the other to progress.
The problem is that the choice of what communication should happen
after a request is taken independently by $\ptp[C]$ and $\ptp[D]$
instead of letting $\ptp[D]$ to take the decision and drive $\ptp[C]$
\quo{on the right} branch.
This is attempted in the next version:
\begin{align}\label{eq:v2}
  \ptp[C] = \bar{\msg[r]}. (\msg[a] + \msg[u])
  \quad\qqand\quad
  \ptp[D] = \msg[r]. (\bar{\msg[a]} + \bar{\msg[u]})
\end{align}
A key difference with the agents in~\eqref{eq:v1} is that the server
$\ptp[D]$ in \eqref{eq:v2} decides what to reply to the client
$\ptp[C]$, which becomes aware of the choice through the interaction
with $\ptp[D]$ after the request has been made.
%
% The agents in \eqref{eq:v2} are seemingly correct...in fact there is
% a problem: the system $(\ptp[C] \mid \ptp[S]) \mid \ptp[C]$ can
% indeed deadlock.
%
Assume now that $\ptp[D]$ acts as a proxy to another server, say
$\ptp[D]'$.
When $\ptp[D]$ cannot return an answer to the client \ptp[c] it interacts with
$\ptp[D]'$ on port $\msg[p]$.
Answers are sent directly to \ptp[c] if $\ptp[D]'$ can compute them,
otherwise $\ptp[D]'$ returns an unknown on port $\msg[u]'$ to
$\ptp[D]$ which forwards it to \ptp[c].
This is modelled by the agents
\begin{align}\label{eq:v3}
  % \ptp[C] = \bar{\msg[r]}. (\msg[a] + \msg[u])
  % \quad\qqand\quad
  \ptp[D] = \msg[r]. (\bar{\msg[a]} + \bar{\msg[p]}. \msg[u]' . \bar{\msg[u]})
  \quad\qqand\quad
  \ptp[D]' = \msg[p]. (\bar{\msg[a]} + \bar{\msg[u]'})
\end{align}
Note that this change is completely transparent to agent $\ptp[c]$,
which in fact stays as in~\eqref{eq:v2}.
It is now more difficult to ascertain if these choices may lead to a
deadlock since the decision of $\ptp[D]$ may involve also $\ptp[D]'$.
Indeed, the parallel composition of agents in \eqref{eq:v3} may
deadlock because, when $\ptp[C]$ and $\ptp[D]$ interact on port
$\msg[a]$, $\ptp[D]'$ hangs on port $\msg[p]$ and, likewise, if
$\ptp[C]$ and $\ptp[D']$ interact on port $\msg[a]$ then $\ptp[D]$
hangs on port $\msg[u]'$.

A reason for this difficulty is that it is hard to define
compositional mechanisms that preserve relevant communication
properties such as deadlock or lock freedom.
Recently, an approach to the composition of concurrent and distributed
systems has been proposed in~\cite{francoICEprev,BarbaneradH19} for
the well-known model of systems of \emph{communicating finite-state
  machines} (CFSMs)~\cite{bz83}, that is sets of finite-state automata
capable of exchanging messages.
The compositional mechanism is based on the idea that two given
systems, say $\aCS$ and $\aCS'$, are composed by transforming two
CFSMs, say $\HH$ in $\aCS$ and $\KK$ in $\aCS'$, into \quo{coupled
  forwarders}.
Basically, each message that $\HH$ receives from a machine in $\aCS$
is forwarded to $\KK$ and vice versa.
It has been shown that, under suitable \emph{compatibility} conditions
between $\HH$ and $\KK$, this composition mechanism preserves deadlock
freedom for asynchronous as well as symmetric synchronous
communications (where sender and receiver play the same part in
determining which message to exchange).
The compatibility condition identified
in~\cite{francoICEprev,BarbaneradH19} consists in exhibiting
essentially dual behaviours: gateway $\HH$ is able to receive a
message whenever gateway $\KK$ is willing to send one and vice versa.
As observed in~\cite{nostroJLAMP}, a remarkable feature of such an
approach is that it enables the composition of systems originally
designed as \emph{closed} systems.
As far as two compatible machines can be found, any two systems can be
composed by transforming as hinted above the compatible machines.

% The main result concerning the proposed technique for composition is
% that this preserves relevant communication properties such as deadlock
% freedom, lock freedom,
% etc. \eMcomm[\cite{bbo12,LangeTY15,BarbaneraLT20}?]{(see,
% e.g.,~\cite{GoudaC86,PengP92}).}

The results in~\cite{francoICEprev,BarbaneradH19} are developed in the
asynchronous semantics of CFSMs.
These results have been transferred in~\cite{BLT20b} to a setting
where CFSMs communicate synchronously much like as the communication
mechanisms considered for instance in process algebras like CCS, ACP,
etc.
This model assumes a perfect symmetry between sender and receiver in
synchronous communications.
Let us again discuss this with an example.
Consider the agents
\begin{align}\label{ex:symmetry}
  \ptp[T] = \bar{\msg[a]}.\ptp[P] + \bar{\msg[b]}.\ptp[Q]
  \qqand
  \ptp[R] = \msg[a].\ptp[P]' + \msg[b].\ptp[Q]'  
\end{align}
According to the standard semantics of CCS~\cite{MilnerR:calcs},
system $(\ptp[T] \mid \ptp[R])\setminus\Set{\msg[a],\msg[b]}$ has two
possible evolutions:
\begin{align*}
  (\ptp[T] \mid \ptp[R])\setminus\Set{\msg[a],\msg[b]} \trans \tau \ptp[P] \mid \ptp[P]'
  \qqand
  (\ptp[T] \mid \ptp[R])\setminus\Set{\msg[a],\msg[b]} \trans \tau \ptp[Q] \mid \ptp[Q]'
\end{align*}
namely, either both $\ptp[T]$ and $\ptp[R]$ opt for the \quo{leftmost}
branch (synchronising on $\msg[a]$) or they both choose the
\quo{rightmost} one (synchronising on $\msg[b]$).
(Recall that in CCS $\_ \setminus X$ is the hiding of ports in the set
$X$ and that $\tau$ represents an internal action.)
This means that, the resolution of the choice is implicit
in the communication mechanism: a branch is taken as soon as $\ptp[T]$
and $\ptp[R]$ synchronise on the corresponding port.
Intuitively, no distinction is made between sender and receiver
(formally they are indeed interchangeable); this implies that the
communication mechanism is at the very core of choice
resolution~\cite{BLT20b}.

Interestingly, for synchronous communications an alternative
interpretation is actually possible. Indeed, this perfect symmetry is not
assumed so that sender and receiver play different roles in choice
resolution while still relying on synchronous communication.
Let us explain this interpretation using again CCS.
Consider a variant of CCS where outputs must be enabled before
being fired.
One could formally specify that with the following reduction rules:
\begin{align}\label{eq:aCCS}
  \bar{\msg[a]}.\ptp[P] + \ptp[P]' \trans \tau \bar{\bar{\msg[a]}}.\ptp[P]
  \qqand
  \bar{\bar{\msg[a]}}.\ptp[P] \mid (\msg[a].\ptp[Q] + \ptp[Q]') \trans \tau \ptp[P] | \ptp[Q]
\end{align}
whereby the leftmost rule \emph{chooses} one of the possible outputs
of the sender (the chosen output is marked by the double bar in our
notation) and the rightmost rule actually synchronises sender and
receiver.
This semantics is an abstract model of \emph{asymmetric}
communications (used e.g., in~\cite{BartolettiSZ14,Padovani10}), where
silent steps taken using the left rule model some internal computation
of the sender to decide what to communicate to the partner.
In other words, now the choice is entirely resolved on \quo{one side}
while the communication is a mere interaction of complementary
actions, the output and the input.
This asymmetry, at the core of asynchronous communication,
can therefore also carry for synchronous communication.

It is worth observing that asymmetric communications abstract a rather
common programming pattern where sending components may choose the
output to execute \emph{depending} on some internal computation.
For instance, elaborating on the proxy scenario in~\eqref{eq:v3},
$\ptp[D]$ could decide to directly send unknowns to normal clients
while reserving the use of $\ptp[D]'$ only for \quo{privileged}
clients.

\paragraph{Contributions.}

This paper transfers the composition by gateway mechanism
of~\cite{francoICEprev,BarbaneradH19} to the case of asymmetric
synchronous communication of CFSMs.
%%%
%%%>>>DA INSERIRE IN EVENTUALE VERSIONE FINALE PER NON ROMPERE LA ANONYMITY<%%%
%The paper extends the investigation initiated
%in~\cite{francoICEprev,BarbaneradH19} to the asymmetric case; the study
%is a byproduct of some comments of an anonymous referee\footnote{We
%  thank the reviewer for the comments inspiring this paper.} of
%another paper.
%
The main technical results are that, in the asymmetric case, gateway composition
\begin{itemize}
\item preserves deadlock freedom (as well as a strong version of lock
  freedom) provided that systems are composable (the relation of compatibility -- one of the requirements for systems to be composable -- in the present paper is less restrictive than %more general that
  the one 
 used in \cite{BLT20b});
\item preserves lock freedom if systems are composable and gateways are \emph{sequential},
  namely each state has at most one outgoing transition.
\end{itemize}
Interestingly, preservation of deadlock freedom can be guaranteed under
milder conditions than in the %asynchronous and 
symmetric case.
In fact, sequentiality of gateways is not necessary to preserve
deadlock freedom in the asymmetric case, while it is in the symmetric one.

% without having to recur to the restrictions needed
%   in~\cite{BLT20b}, namely \quo{sequentiality} of gateways. This is
%   proved to hold for the strong lock freedom property.  For what
%   concerns lock-freedom, instead, this is not a property preserved by
%   composition in general, so one has to recur to \quo{sequentiality} of
%   gateways for this property to be guaranteed when one compose two
%   lock-free systems.

% Somehow surprisingly, fairly stricter conditions are required to
% ensure the preservation of deadlock freedom for the compositional
% mechanism in~\cite{francoICEprev,BarbaneradH19}.

% We formalise those systems in terms of \emph{communicating
%   systems}~\cite{bz83}, namely sets of finite-state automata whose
% transitions are labelled by sending and receiving actions.
%
%Traditionally these systems are viewed as \emph{closed}, 
%thus one needs full knowledge of the whole
%system in order to analyse it. In scenarios such as the Internet, the
%Cloud or serverless computing, such assumption is less and less
%realistic.

\paragraph{Structure of the paper.} \cref{sec:back} introduces systems
of (asymmetric synchronous) CFSMs, related notions and
communication properties.
Composition by gateways is introduced and discussed in
\cref{sec:binary} together with the compatibility relation.
\cref{sec:proppres} discusses the issues that prevent gateway
composition to preserve communication properties.
\cref{sec:pbc} is devoted to the preservation of communication
properties.
Conclusions, related and future work are discussed in~\cref{sec:conc}.
\iftr\else
For space limitation, full proofs are reported in~\cite{bltice22TR}.
\fi

%%% Local Variables:
%%% mode: latex
%%% TeX-master: "main"
%%% End:

\section{Background}\label{sec:back}
% A FEW MACHINES DEFINITIONS ==========
\newcommand{\cfsmA}{
  \begin{tikzpicture}[mycfsm]
	 \node[state, initial, initial where = left, initial distance = .4cm, initial text={\p}] (zero) {$0$};
	 \node[state] (one)  [above right of=zero] {$1$};
	 \node[state] (two)  [below right of=one] {$2$};
	 \path (zero) edge [bend left]  node[above] {$\tau$} (one)
	 (one) edge [bend left]  node[above] {$\aout[A][H][][m]$} (two)
	 ;
  \end{tikzpicture}
}
\newcommand{\cfsmH}{
  \begin{tikzpicture}[mycfsm]
	 \node[state, initial, initial where = above, initial distance = .4cm, initial text={$\HH$}] (zero) {$0$};
	 \node[state] (one)  [below of=zero] {$1$};
	 \path
	 (zero) edge [bend right] node[below] {$\ain[A][H][][m]$} (one)
	 edge [bend left] node[above,rotate=180] {$\ain[B][H][][n]$} (one)
	 ;
  \end{tikzpicture}
}
\newcommand{\cfsmB}{
  \begin{tikzpicture}[mycfsm]
	 \node[state, initial, initial where = left, initial distance = .4cm, initial text={\q}] (zero) {$0$};
	 \node[state] (one)  [below right of=zero] {$1$};
	 \node[state] (two)  [above right of=one] {$2$};
	 \path (zero) edge [bend right]   node[below] {$\tau$} (one)
	 (one) edge [bend right]   node[below] {$\aout[B][H][][n]$} (two)
	 ;
  \end{tikzpicture}
}
\newcommand{\cfsmK}{
  \begin{tikzpicture}[mycfsm]
	 \node[state, initial, initial where = left, initial distance = .4cm, initial text={$\KK$}] (zero) {$0$};
	 \node[state] (one) [above right of=zero, yshift = -.5cm]   {$1$};
	 \node[state] (two) [below right of=zero, yshift = .5cm]   {$2$};
	 \node[state] (three) [below right of=one, yshift = .5cm]   {$3$};
	 \path (one) edge[bend left] node[above] {$\aout[K][C][][m]$} (three)
	 (two) edge[bend right] node[below] {$\aout[K][D][][n]$} (three)
	 (zero) edge[bend left] node[above] {$\tau$} (one)
	 (zero) edge[bend right] node[below] {$\tau$} (two)
	 ;
  \end{tikzpicture}
}
\newcommand{\cfsmC}{
  \begin{tikzpicture}[mycfsm]
	 \node[state, initial, initial where = left, initial distance = .4cm, initial text={$\ptp[C]$}] (zero) {$0$};
	 \node[state] (one) [above right of=zero, yshift = -.5cm]   {$1$};
	 \node[state] (two) [right  of=one]   {$2$};
	 \node[state] (three) [below right of=two, yshift = .5cm]   {$3$};
	 \path (one) edge node[above] {$\tau$} (two)
	 (two) edge[bend left] node[above] {$\aout[C][E][][c]$} (three)
	 (zero) edge[bend left] node[above] {$\ain[K][C][][m]$}
	 (one)  edge node[below] {$\ain[E][C][][s]$} (three)
	 ;
  \end{tikzpicture}
}
\newcommand{\cfsmD}{
  \begin{tikzpicture}[mycfsm]
	 \node[state, initial, initial where = left, initial distance = .4cm, initial text={$\ptp[D]$}] (zero) {$0$};
	 \node[state] (zero) {$0$};
	 \node[state] (one) [above right of=zero, yshift = -.5cm]   {$1$};
	 \node[state] (two) [right  of=one]   {$2$};
	 \node[state] (three) [below right  of=two, yshift = .5cm]   {$3$};
	 \path (one) edge node[above] {$\tau$} (two)
	 (zero) edge[bend left] node[above] {$\ain[K][D][][n]$} (one)
	 (zero) edge node[below] {$\ain[E][D][][s]$} (three)
	 (two) edge[bend left] node[above] {$\aout[D][E][][d]$} (three)
	 ;
  \end{tikzpicture}
}
\newcommand{\cfsmE}{
  \begin{tikzpicture}[mycfsm]
	 \node[state, initial, initial where = left, initial distance = .4cm, initial text={$\ptp[E]$}] (zero) {$0$};
	 \node[state] (three)  [below right of=zero, yshift = .5cm] {$4$};
	 \node[state] (four)  [above right  of=zero, yshift = -.5cm] {$1$};
	 \node[state] (five)  [right  of=three] {$5$};
	 \node[state] (six)  [right   of=four] {$2$};
	 \node[state] (two)  [below right of=six, yshift = .5cm] {$3$};
	 \path
	 (zero) edge [bend right] node[below] {$\ain[C][E][][c]$} (three)
	 edge [bend left] node[above] {$\ain[D][E][][d]$} (four)
	 (three) edge  node[above] {$\tau$} (five)
	 (four) edge node[above] {$\tau$} (six)
	 (five) edge[bend right]  node[below] {$\aout[E][D][][s]$} (two)
	 (six) edge [bend left] node[above] {$\aout[E][C][][s]$} (two)
	 ;
  \end{tikzpicture}
}

Communicating Finite State Machines (CFSMs)~\cite{bz83} are Finite
State Automata (FSAs) where transitions are labelled by
communications.
We recall basic notions on FSAs.

% \begin{definition}[FSA]
A \emph{finite state automaton} (FSA) is a tuple
$A = \conf{\sset, q_0, \lset, \tset}$ where
\begin{itemize}
\item $\sset$ is a finite set of states (ranged over by lowercase
  italic Latin letters);
\item $q_0 \in \sset$ is the \emph{initial state};
\item $\lset$ is a finite set of labels
\item
  $\tset \,\subseteq\, \sset \times (\lset \cup \{\tau\}) \times
  \sset$ is a set of transitions.
\end{itemize}
% \end{definition}
%
Hereafter, we let $\al$ range over $\lset \cup \{\tau\}$ when it is
immaterial to specify the set of labels or it is understood.
We use the usual notation $q_1 \trans \al q_2$ for the transition
$(q_1,\al,q_2) \in \trans{}$, and $q_1 \trans{} q_2$ when there exists a
label $\al$ such that $q_1 \trans \al q_2$.
Let $\_\cat\_$ be the concatenation operation on labels and write
$p \trans \pi q$ where $\pi = \al_1 \cat \al_2 \cat \ldots \cat \al_n$
whenever
$p \trans{\al_1} p_1 \trans{\al_2} \ldots \trans{\al_n} p_n = q$.
We let $\pi,\psi,\ldots$ range over $\lset^\star$ (i.e., sequences of
labels) and define the set of \emph{reachable states in $A$ from $q$}
as
\begin{align*}
  \RS[A,q] & = \Set{p \sst \text{there is } \pi \in \lset^\star \text{ such that } q \trans \pi p}
\end{align*}
The set of \emph{reachable states in $A$} is $\RS[A] = \RS[A,q_0]$.
For succinctness, $q \trans{\al} q' \in A$ means that the transition
belongs to (the set of transitions of) $A$; likewise, $q \in A$ means
that $q$ belongs to the states of $A$.
We say that $q \trans{\al} q'$ is an \emph{outgoing}
(resp. \emph{incoming}) transition of $q$ (resp. $q'$).
Since we use FSAs to formalise communicating systems, accepting states
are disregarded (as also done in~\cite{bz83}).

We adapt the definitions in~\cite{bz83} to cast CFSMs in our context.
Let $\mathfrak{P}$ be a set of \emph{participants} (or \emph{roles},
ranged over by $\p$, $\p[B]$, etc.) and $\msgset$ a set of
\emph{messages} (ranged over by $\msg$, $\msg[n]$, etc.).
We take $\mathfrak{P}$ and $\msgset$ disjoint.
An \emph{output label} is written as $\aout$ and represents the
willingness of \p\ to send message $\msg$ to \q; likewise, an
\emph{input label} is written as $\ain$ and represents the
willingness of \q\ to receive message $\msg$ from \p.
The \emph{subjects} of an output label $\aout$ and of an input label
$\ain$ are \p\ and \q, respectively.

\begin{definition}[CFSMs]\label{def:cfsm}
  Let	 $\lact = \{\aout, \ain \mid \p[A] \neq \p[B] \in \mathfrak{P}, \msg \in \msgset\}$.
  A \emph{communicating finite-state machine} (CFSM) is an FSA $\aCM$
  with labels $\lact \cup \Set{\tau}$ such that, for any transition
  $p \trans \al q$ of $\aCM$,
  \begin{itemize}
  \item if $\al$ is an output label then $p \neq q$ and $p$ has
	 exactly one incoming transition, and such transition is labelled
	 by $\tau$;
  \item if $\al = \tau$ then $p \neq q$ and $q$ has exactly one
	 outgoing transition, and such transition is labelled by an output
	 label.
  \end{itemize}
  A state of $\aCM$ is
  \begin{itemize}
  \item \emph{terminal}, if it has no outgoing transition; we define
	 $\tstates \aCM = \Set{p \in \aCM \sst p \text{ is terminal in }
		\aCM}$
  \item \emph{sending}, if it is not terminal and all its outgoing
	 transitions have output labels
  \item \emph{receiving}, if it is not terminal and all its outgoing
	 transitions have input labels
  \item \emph{mixed}, if it has a silent outgoing transition and an
	 outgoing transition with an input label.
  \end{itemize}
  A CFSM is \emph{\p-local} if all its non $\tau$ transitions have subject \p.
\end{definition}
Unlike in~\cite{BLT20b}, the transitions of our CFSMs can also be
labelled by the silent action $\tau$.
\cref{def:cfsm} can be looked at as the CFSM %communicating automata
counterpart of the $\tau C$ contracts described
in~\cite{BartolettiCZ15}.
Imposing the no-mixed state condition on our CFSM, turns them into the
communicating automata counterpart of the processes (contracts) called
\quo{session behaviours}\footnote{Actually different variations of
  this name are used in the listed references.} in
e.g.,~\cite{HennessyB12,BdL13,BartolettiSZ14}.
These processes are in turn the process counterpart of (binary) session types~\cite{honda.vasconcelos.kubo:language-primitives}.
As we shall see below (and also shown in~\cite{BarbaneradH19}
and~\cite{BLT20b}), the absence of mixed states is necessary in
order to get the preservation of properties by composition.
As a matter of fact, we could drop the conditions related to
$\tau$-transitions in case a transition like
$\cauttr[p][{\aout[X][Y][@][z]}][q][][]$ is the only outgoing
transition from $p$, namely when no actual choice of output actions is
possible in $p$.
We however prefer to avoid this distinction because
($i$) our uniform treatment of transitions allows us to
immediately adapt definitions in a more abstract setting
and ($ii$)
  % for instance, the behaviour of participants could be described by
  % means of formal languages rather than by automata;
secondly, uniformity allows us to simplify some technicalities.
Said that, all proofs in the present paper could easily be adapted to
the above mentioned alternative definition of CFSM.

% It is worth remarking that the use of $\tau$-transitions in the above
% definition is not equivalent to that of buffers of size 1 in the model
% of asynchronous CFSMs, neither to their particular use in the CFSMs
% model \eMcomm[quello a cui accennava Emilio]{XXX}, since we are
% modeling (asymmetric) synchronous interactions, which are
% \emph{blocking}.

\begin{definition}[Communicating systems]\label{def:cs}
  A \emph{(communicating) system over $\ptpset$} is a map
  $\aCS = (\aCM_{\p})_{\p \in \ptpset}$ assigning an $\p$-local CFSM
  $\aCM_{\p}$ to each participant $\p \in \ptpset$ where
  $\ptpset \subseteq \mathfrak{P}$ is finite and any participant
  occurring in a transition of $\aCM_{\p}$ is in $\ptpset$.
\end{definition}
Note that \cref{def:cs} requires that any input or output label does
refer to participants belonging to the system itself.
In other words, \cref{def:cs} restricts to \emph{closed} systems.

We define the synchronous semantics of communicating systems as an FSA
(differently from the asynchronous case, where the set of states can
be infinite).
Hereafter, $\dom f$ denotes the domain of a function $f$ and
$\upd f x y$ denotes the update of $f$ in a point $x \in \dom f$ with
the value $y$.

\begin{definition}[Asymmetric synchronisations]\label{def:syncSem}
  Let $\aCS$ be a communicating system.
  A \emph{configuration of } $\aCS$ is a map
  $\aConf = (q_{\p})_{\p \in \dom \aCS}$ assigning a \emph{local state}
  $q_{\p}\in \aCS(\p)$ to each $\p \in \dom \aCS$.

  The \emph{asymmetric synchronisations} of $\aCS$ is the FSA
  $\ssem{\aCS} = \conf{\sset, \aConf_0, \lint\cup\Set{\tau},\tset}$
  where
  \begin{itemize}
  \item $\sset$ is the set of synchronous configurations of $\aCS$, as
	 defined above;
  \item 
	 $\aConf_0=(q_{0\p})_{\p \in \dom \aCS} \in
	 \sset$ is the \emph{initial} configuration where, for  each $\p \in \dom \aCS$, 
	 $q_{0\p}$ is the initial state of $\aCS(\p)$; 
  \item
	 $\lint \ =\ \{\gint[] \mid \p \neq \q \in \mathfrak{P} \text{ and
	 } \msg \in \msgset\}$ is a set of interaction labels;
	 % (ranged over by $\aint$, $\aint[\beta]$, etc.);
  \item 
	 $\aConf \trans{\gint[]} \upd \aConf \p {q, \q \mapsto q'} \in \ssem{\aCS}$
	  if $\aConf(\p) \trans{\aout} q \in \aCS(\p)$ and $\aConf(\q)
	 \trans{\ain} q' \in \aCS(\q)$;
   \item 
	 $\aConf \trans{\tau} \upd \aConf \p {q} \in \ssem{\aCS}$
	  if $\aConf(\p) \trans{\tau} q \in \aCS(\p)$;
  \end{itemize}
  Configuration $\aConf$ \emph{enables \p\ in $\aCS$} if $\aConf(\p)$
  has at least an outgoing transition.
  %
  % We say
  % that \emph{$\aConf$ enables $q \trans{\aout} q' \in \aCS(\p)$}
  % (resp.\ $q \trans{\ain[b][a]} q', q \trans{\tau} q' \in \aCS(\p)$)
  % when $\aConf(\p) = q$.
\end{definition}
As expected, an interaction $\gint[]$ occurs when $\p$ performs an
output $\aout$ (which has been previously chosen) and $\q$ the
corresponding input $\ain$.
\begin{example}\label{ex:sem}Let us consider the communicating system
  $\aCS= (\aCM_{\ptp[X]})_{\ptp[X] \in
	 \Set{\ptp[K],\ptp[C],\ptp[D],\ptp[E]}}$, where
  
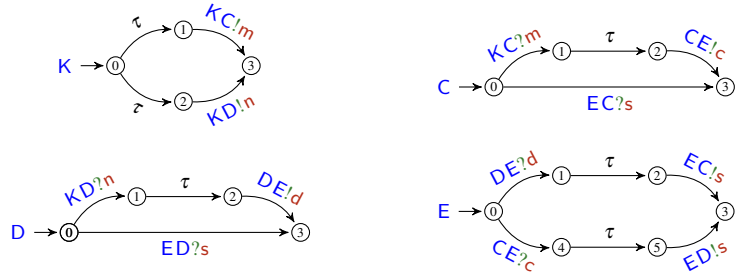
\[
	 \begin{array}{c@{\hspace{1.5cm}}c}
		\cfsmK
		&
		\cfsmC
		\\
		\cfsmD
		&
		\cfsmE
	 \end{array}
  \]
  %One of the possible
  A sequence of %configuration
  transitions of
  $\ssem{\aCS}$ out of $\aConf_0$ is, according to \cref{def:syncSem},
\[\begin{array}{cccccr}
		\aConf_0 = (0_{\KK},0_{\ptp[C]},0_{\ptp[D]}, 0_{\ptp[E]})
		& \trans \tau & (1_{\KK},0_{\ptp[C]},0_{\ptp[D]}, 0_{\ptp[E]})
		& \trans{\gint[][K][m][C]} & (3_{\KK},1_{\ptp[C]},0_{\ptp[D]},0_{\ptp[E]})
		\\
		& \trans \tau & (3_{\KK},2_{\ptp[C]},0_{\ptp[D]},0_{\ptp[E]})
		& \trans{\gint[][C][c][E]} & (3_{\KK},3_{\ptp[C]},0_{\ptp[D]},4_{\ptp[E]})
		\\
		& \trans \tau & (3_{\KK},3_{\ptp[C]},0_{\ptp[D]},5_{\ptp[E]})
		& \trans{\gint[][E][s][D]} & (3_{\KK},3_{\ptp[C]}, 3_{\ptp[D]},3_{\ptp[E]})
		& \hspace{3.25cm}\mbox{\finex}
	 \end{array}\]
\end{example}
The symmetric synchronisation in~\cite{BLT20b} for systems without
$\tau$-transitions can be readily obtained from the above definition
by disregarding the clause for the $\tau$-transitions.

In the following, $\ptpof[\tau] = \emptyset$ and
$\ptpof = \ptpof[\aout] = \ptpof[\ain] = \Set{\p,\q}$ and, for a
sequence $\pi = \al_1 \cat \cdots \cat \al_n$, we let
$\ptpof[\pi] = \cup_{1 \leq i \leq n} \ptpof[\al_i]$.

As discussed in \cref{sec:intro}, we shall study the preservation of
communication properties under composition.
We shall consider the following relevant properties: deadlock freedom, lock
freedom and strong lock freedom.
The definitions below adapt the ones in~\cite{cf05} to a synchronous
setting (as done also in~\cite{LangeTY15,gt18,BLT20b}).
\begin{definition}[Communication properties]\label{def:props}
  Let $\aCS$ be a communicating system on $\ptpset$.
  We say that a participant $\p \in \ptpset$ is \emph{involved} in a
  run $\aConf \trans{\al_1} \aConf_1 \ldots \trans{\al_n} \aConf_n$ of
  $\aCS$ if there is $1 \leq i \leq n$ such that either
  $\p \in \ptpof[\al_i]$ or $\al_i = \tau$,
  $\aConf_i(\p) \trans \tau q$ in $\aCS(\p)$, and
  $\aConf_{i+1} = \upd{\aConf_i} \p q$.
  \begin{description}
  \item[Deadlock freedom]
	 A configuration $\aConf \in \RS[\ssem \aCS]$ is a \emph{deadlock}
	 if
	 \begin{itemize}
	 \item $\aConf$ has no outgoing transitions in $\ssem{\aCS}$ and
	 \item there exists $\p \in \ptpset$ such that $\aConf(\p)$ enables \p\ in $\aCS$.
	 \end{itemize}
	 A system is \emph{deadlock-free} if none of its configurations is a deadlock.
    \item[Lock freedom]
	 Let $\p\in\ptpset$.
	 A configuration $\aConf \in \RS[\ssem
	 \aCS]$ is a \emph{lock} for $\p$ if
	 \begin{itemize}
	 \item $\aConf(\p)$ has outgoing transitions; and
	 \item \p\ is not involved in any run from $\aConf$.
	 \end{itemize}
	 A system is \emph{lock-free} if none of its configurations is a
	 lock for any of its participants.
  \item[Strong lock freedom] System $\aCS$ is \emph{strongly lock-free
		for $\p \in \ptpset$} if for each $\aConf \in \RS[\ssem \aCS]$
	 enabling \p\ in $\aCS$ then \p\ is involved in all maximal
	 sequences from $\aConf$.
	 \\
	 A system is \emph{strongly lock free} if it is strongly lock free
	 for each of its participants.
	 \end{description}
  \end{definition}
  
  \begin{proposition}
  \begin{enumerate}
  \item
  Lock-freedom implies deadlock-freedom;
  \item
  Strong lock freedom implies lock freedom.
  \end{enumerate}
  \end{proposition}
  
  \begin{example}
	 Let us consider the system $\aCS$ of \cref{ex:sem}. The only other
	 maximal transition sequence in $\ssem \aCS$ out of $\aConf_0$,
	 besides the one described in \cref{ex:sem}, is {\small
	 $$\begin{array}{ccccc}
		\aConf_0 = (0_{\KK},0_{\ptp[C]},0_{\ptp[D]}, 0_{\ptp[E]})
		& \trans \tau & (2_{\KK},0_{\ptp[C]},0_{\ptp[D]}, 0_{\ptp[E]})
		& \trans{\gint[][K][n][D]} & (3_{\KK},0_{\ptp[C]},1_{\ptp[D]},0_{\ptp[E]})
		\\
		& \trans \tau & (3_{\KK},0_{\ptp[C]},2_{\ptp[D]},0_{\ptp[E]})
		& \trans{\gint[][D][d][E]} & (3_{\KK},0_{\ptp[C]},3_{\ptp[D]},1_{\ptp[E]})
		\\
		& \trans \tau & (3_{\KK},0_{\ptp[C]},3_{\ptp[D]},2_{\ptp[E]})
		& \trans{\gint[][E][s][C]} & (3_{\KK},3_{\ptp[C]}, 3_{\ptp[D]},3_{\ptp[E]})
	 \end{array}$$}
These two sequences are both maximal and contain all the elements of $\RS[\ssem \aCS]$.
By the above observations it is possible to check $\aCS$ to be strongly lock free.
\finex
\end{example}
%%% Local Variables:
%%% mode: latex
%%% TeX-master: "main"
%%% End:

\section{Composition via Gateways}\label{sec:binary}
We now discuss the composition of systems of CFSMs via gateways (cf.
~\cite{francoICEprev,BarbaneradH19}) and study its properties under
asymmetric synchronisation.
The main idea is that two systems of CFSMs, say $\aCS_1$ and $\aCS_2$,
can be composed by transforming one participant in each of them into
gateways connected to each other.

\subsection{Building gateways}
Hereafter, $\HH$ and $\KK$ denote the selected participant in $\aCS_1$
and $\aCS_2$ respectively selected for the composition.
The gateways for $\HH$ and $\KK$ are connected to each other and act
as forwarders: each message sent to the gateway for $\HH$ by a
participant from the original system $\aCS_1$ is now forwarded to the
gateway for $\KK$, that in turn forwards it to the same participant to
which $\KK$ sent it in the original system $\aCS_2$.
The dual will happen to messages that the gateway for $\KK$ receives
from $\aCS_2$. A main advantage of this approach is that no extension
of the CFSM model is needed to transform systems of CFSMs, which are
normally closed systems, into open systems that can be
composed. Another advantage is that the composition is fully
transparent to all participants different from $\HH$ and $\KK$.

We will now define composition via gateways on systems of CFSMs,
following the intuition above.

\begin{definition}[Gateway]\label{def:gateway}
  Given a $\HH$-local CFSM $\aCM$ and a participant $\KK$, the
  \emph{gateway of $\aCM$ towards $\KK$} is the CFSM
  $\gateway \aCM \KK$ obtained by replacing in $\aCM$
  \begin{itemize}
    \item each pair of consecutive transitions $p \trans{\tau} q
  \trans{\aout[H][A]} r$ with
  \begin{align}\label{eq:gtwout}
	 p \trans{\ain[K][H]} p' \trans{\tau} q
	 \trans{\aout[H][A]} r
	 \qqand[for some fresh state $p'$]
  \end{align}
\item each transition $p \trans{\ain[A][H]} r$ with
  \begin{align}\label{eq:gtwin}
	 p \trans{\ain[A][H]} p' \trans{\tau} p'' \trans{\aout[H][K]} r
	 \qqand[for some fresh states $p'$ and $p''$]
  \end{align}
\end{itemize}
  We shall call \emph{external} the states like $p$ and $r$ and
  \emph{internal} the states like $p'$, $p''$ and $q$.
\end{definition}
Note that gateways execute \quo{segments} of the form
described in~\eqref{eq:gtwout} and~\eqref{eq:gtwin} in the above
definition.
Also, by very construction, we have the following
\begin{fact}
  Given a $\HH$-local CFSM $\aCM$ and a participant $\KK$,
  each state of $\gateway \aCM \KK$ %is such that each of its states
  has at most one incoming or outgoing
  $\tau$ transition.
\end{fact}

We compose systems with disjoint participants taking all the
participants of the original systems but $\HH$ and $\KK$, whereas
$\HH$ and $\KK$ are replaced by their respective gateways.

Given two functions $f$ and $g$ such that
$\dom f \cap \dom g = \emptyset$, we let $f + g$ denote the function
behaving as function $f$ on $\dom f$ and as function $g$ on $\dom g$.
\begin{definition}[System composition]
Let $\aCS_1$ and $\aCS_2$ be two systems with disjoint domains.
  The \emph{composition of $\aCS_1$ and $\aCS_2$ via
	 $\HH \in \dom{\aCS_1}$ and $\KK \in \dom{\aCS_2}$} is defined as
  \[\aCS_1 \connect \HH \KK \aCS_2 = \upd{\aCS_1}{\HH}{\gateway{\aCS_1(\HH)} \KK}
	 + \upd{\aCS_2} \KK {\gateway{\aCS_2(\KK)} \HH}
	 % \p \mapsto \begin{cases}
  %   \aCS_1(\p), &  \text{if } \p \in \dom{\aCS_1} \setminus \Set{\HH, \KK}
  %   \\
  %   \aCS_2(\p), &  \text{if } \p \in \dom{\aCS_2} \setminus \Set{\HH, \KK}
  %   \\
  %   \gateway{\aCS_1(\HH)} \KK, & \text{if } \p = \HH
  %   \\
  %   \gateway{\aCS_2(\KK)} \HH, & \text{if } \p = \KK
  % \end{cases}
  \]
(Note that $\dom{\aCS_1 \connect \HH \KK \aCS_2} = \dom{\aCS_1} \cup \dom{\aCS_2}$.)
\end{definition}
We remark that, by the above approach for composition, we do not
actually need to formalise the notion of \emph{open} system.
In fact any closed system can be looked at as open by choosing two
suitable participants in the \quo{to-be-connected} systems and
transforming them into two forwarders.
Also, note that the notion of composition above is structural:
corresponding notions of behaviourals composition have been studied
in~\cite{nostroJLAMP} and in~\cite{survey} for multiparty session
types.

\begin{example}\label{ex:gc}
  Let us take the following two communicating systems.
  \[
	 \aCS_1 = \dboxed
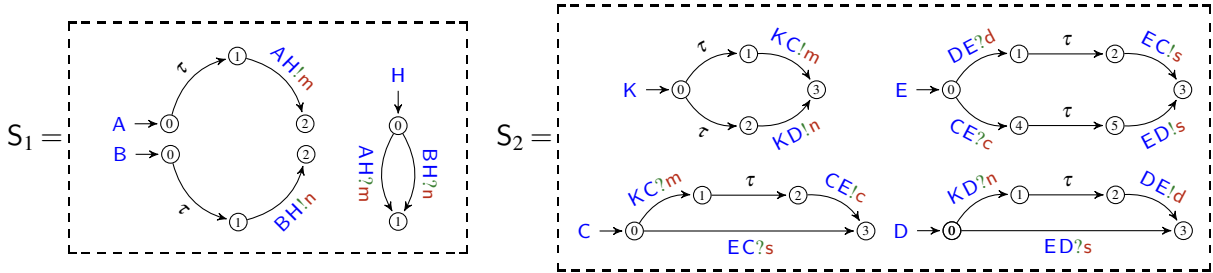
{
		\begin{array}{l@{\hspace{.1cm}}l}
		  \begin{array}{c}
			 \cfsmA
			 \\
			 \cfsmB
		  \end{array}
		  &
		  \cfsmH
		\end{array}
	 }
	 \quad
	 \aCS_2 = \dboxed{\arraycolsep=1pt\def\arraystretch{3}
		\begin{array}{cc}
		  \cfsmK
		  &
			 \cfsmE
		  \\[-1em]
		  \cfsmC
		  &
		  	 \cfsmD
		\end{array}
	 }
  \]
  The system $\aCS_1 \connect{\HH}{\KK} \aCS_2$ is
  \[
	 \csys{\arraycolsep=1.2pt
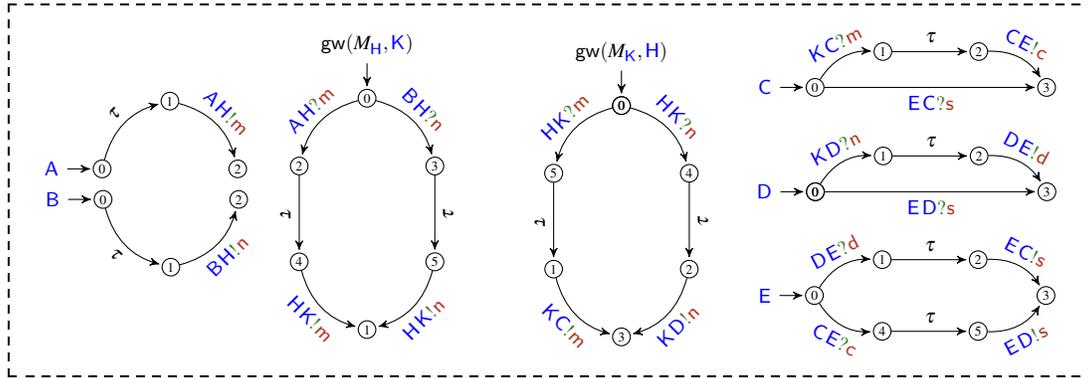
		\begin{array}{ll}
		  \begin{array}{c}
		  \cfsmA
		  \\
		  \cfsmB
		  \end{array}
		  &
			 \begin{tikzpicture}[mycfsm]
				\node[state, initial, initial where = above, initial distance = .4cm, initial text={$\gateway {\aCM_\HH} \KK$}] (zero) {$0$};
				\node[state] (three)  [below left of=zero] {$2$};
				\node[state] (four)  [below right  of=zero] {$3$};
				\node[state] (five)  [below  of=three] {$4$};
				\node[state] (six)  [below   of=four] {$5$};
				\node[state] (one)  [below  right of=five] {$1$};
				\path (zero) edge [bend right] node[above] {$\ain[A][H][][m]$}
				(three) edge [bend left] node[above] {$\ain[B][H][][n]$} (four)
				(three) edge  node[below] {$\tau$} (five)
				(four) edge node[below] {$\tau$} (six)
				(five) edge[bend right]  node[below] {$\aout[H][K][][m]$} (one)
				(six) edge[bend left]  node[below] {$\aout[H][K][][n]$} (one)
				;
			 \end{tikzpicture}
		\end{array}
		\qquad
		\begin{array}{l@{\quad}l}
		  \begin{tikzpicture}[mycfsm]
			 \node[state, initial, initial where = above, initial distance = .4cm, initial text={$\gateway {\aCM_\KK} \HH$}] (zero) {$0$};
			 \node[state] (zero) {$0$};
			 \node[state] (four) [below right of=zero]   {$4$};
			 \node[state] (five) [below left of=zero]   {$5$};
			 \node[state] (one) [below of=five]   {$1$};
			 \node[state] (two) [below  of=four]   {$2$};
			 \node[state] (three) [below right of=one]   {$3$};
			 \path (one) edge[bend right] node[below] {$\aout[K][C][][m]$} (three)
			 (two) edge[bend left] node[below] {$\aout[K][D][][n]$} (three)
			 (five) edge node[below] {$\tau$} (one)
			 (four) edge node[below] {$\tau$} (two)
			 (zero) edge[bend right] node[above] {$\ain[H][K][][m]$} (five)
			 (zero) edge[bend left] node[above] {$\ain[H][K][][n]$} (four)
			 ;
		  \end{tikzpicture}
		  &
			 \begin{array}{c}
				\cfsmC
				\\
				\cfsmD
				\\
				\cfsmE
			 \end{array}
		\end{array}
	 }
  \]
  Note that the CFSMs \p, \q, $\ptp[C]$, $\ptp[D]$, and $\ptp[E]$
  remain unchanged.
  \finex
\end{example}
  
\iftr
  Given a configuration of the composition of systems $\aCS_1$ and
  $\aCS_2$ we can retrieve the configurations of the two subsystems by
  taking only the states of participants in $\aCS_1$ and $\aCS_2$
  while avoiding, for the gateways, to take the fresh states
  introduced by the gateway construction.
  Indeed, we shall prove in \cref{prop:comp} that for each
  $\aConf \in \RS[\ssem{\aCS}]$, we have
  $\cproj \aConf \HH \in \RS[\ssem{\aCS_1}]$ and
  $\cproj \aConf \KK \in \RS[\ssem{\aCS_2}]$.
  \begin{definition}[Configuration projections]\label{def:cproj}
	 Let $\aConf$ be a configuration of a composed system
	 $\aCS_1 \connect \HH \KK \aCS_2$ and
	 $\aCM = \gateway{\aCS_1(\HH)} \KK$.
	 The \emph{left-projection of $\aConf$ on $\aCS_1$} is the map
	 $\cproj \aConf \HH$ defined on $\dom{\aCS_1}$ by
	 \[
		\cproj \aConf \HH : \p \mapsto
		\begin{cases}
		  q, & \text{if } \aConf(\p) \not\in \aCS_1(\p)
		  \text{ and there is } q \trans{\ain[k][h]} \aConf(\p) \in \aCM
		  \text{ or } q \trans {\ain[k][h]} \trans \tau \aConf(\p) \in \aCM
		  \\
		  r, & \text{if } \aConf(\p) \not\in \aCS_1(\p)
		  \text{ and there is } \aConf(\p) \trans{\aout[h][k]} r \in \aCM
		  \text{ or } \aConf(\p) \trans \tau \trans{\aout[h][k]} r \in \aCM
		  \\
		   \aConf(\p), & \text{otherwise } %\aConf(\p) \in \aCS_1(\p)
		\end{cases}
	 \]
	 % \eMcomm[versione ``originale'']{}
	 % \[
	 % 	 \cproj \aConf \HH : \p \mapsto
	 % 	 \begin{cases}
	 % 	  \aConf(\p), & \text{if } \aConf(\p) \text{ is
	 % 	 external} %\text{ is not fresh}
	 % 		\\
	 % 		q, & \text{if } \p = \HH,\ \aConf(\HH) \text{ is internal and}
	 % 		\\ & \quad \text{either } q \trans{\ain[K][h]} \aConf(\HH) \in \aCM
	 % 		\text{ or } q \trans{\ain[k][h]} q' \trans{\tau} \aConf(\HH) \in \aCM
	 % 		\\
	 % 		q, & \text{if } \p = \HH,\ \aConf(\HH) \text{ internal, and}
	 % 		\\ & \quad \text{either } \aConf(\HH) \trans{\aout[h][k]} q \in \aCM
	 % 		\text{ or } \aConf(\HH) \trans{\tau} q' \trans{\aout[h][k]} q \in \aCM
	 % \end{cases}
	 % \]
	 The definition of \emph{right-projection} $\cproj \aConf \KK$ is analogous.
  \end{definition}
  \begin{proposition}
	 Left- and right-projections are well-defined.
  \end{proposition}
  \begin{proof}
	 It is enough to show that $\cproj \cdot \HH$ and $\cproj \cdot
	 \KK$ uniquely assign a state to fresh states because on non-fresh
	 stats both functions act as the identify map.
	 This follows since, by construction, each state introduced by our
	 gateway construction has unique successor and predecessor.
  \end{proof}
  Intuitively, only $\HH$ is aware of an input from $\KK$ when $\HH$
  is in the internal state reached after an input from $\KK$; hence to
  have a coherent configuration we take the state of $\HH$ before the
  input.
  If instead $\HH$ is in an internal state corresponding to an output
  to $\KK$, then other participants in $\aCS_1$ know that the message
  has been sent; hence to have a coherent configuration we take the
  state of $\HH$ after the send.
  (A similar intuition applies to $\cproj \aConf \KK$.)
  \begin{example}\label{ex:confproj}
	 Let
	 $\aConf =
	 (2_{\p},0_{\q},1_{\HH},5_{\KK},0_{\ptp[C]},0_{\ptp[D]},0_{\ptp[e]})$;
	 and $\aCS = \aCS_1 \connect{\HH}{\KK} \aCS_2$ be the system of
	 \cref{ex:gc}.
	 Then, $\aConf \in \RS[\ssem{\aCS}]$, namely
	 $\aConf$ is reachable in $\aCS$.
	 In fact
	 \begin{eqnarray*}
		\aConf_0 = (0_{\p}, 0_{\q}, 0_{\HH},0_{\KK},0_{\ptp[C]},0_{\ptp[D]}, 0_{\ptp[E]})
		& \trans \tau & (1_{\p}, 0_{\q}, 0_{\HH},0_{\KK},0_{\ptp[C]},0_{\ptp[D]}, 0_{\ptp[E]})
		\\
		& \trans{\gint[][A][m][H]} & (2_{\p}, 0_{\q}, 2_{\HH},0_{\KK},0_{\ptp[C]},0_{\ptp[D]},0_{\ptp[E]})
		\\
		& \trans \tau & (2_{\p}, 0_{\q}, 4_{\HH},0_{\KK},0_{\ptp[C]},0_{\ptp[D]},0_{\ptp[E]})
		\\
		& \trans{\gint[][H][m][K]} & (2_{\p}, 0_{\q}, 1_{\HH},5_{\KK},0_{\ptp[C]},0_{\ptp[D]},0_{\ptp[E]})
	 \end{eqnarray*}
  	 The projections of $\aConf$ on, respectively, $\aCS_1$ and $\aCS_2$
	 are $\cproj \aConf \HH = (2_{\p},0_{\q}, 1_{\HH})$ and
	 $\cproj \aConf \KK = (0_{\KK},0_{\ptp[C]},0_{\ptp[D]},0_{\ptp[E]})$.
	 \finex
  \end{example}
\fi
\subsection{Compatibility}

A few simple auxiliary notions are useful.
Let
$\lio = \Set{\ain[][][][m], \aout[][][][m] \sst \msg[m]\in\msgset}$
and define the functions $\mathsf{io} : \lact \to \lio$ and
$\Dual{(\cdot)} : \lio \to \lio$ as follows
\[
  \begin{array}[c]{c@{\qquad}c@{\qquad}c}
    \mathsf{io}(\ain) = \ain[][][][m]
    &
           \mathsf{io}(\aout) = \aout[][][][m]
    \qqand
    \Dual{\ain[][][][m]} = \aout[][][][m]
    &
      \Dual{\aout[][][][m]} = \ain[][][][m]
  \end{array}
\]
These functions extend to CFSMs
in the obvious way: given a CFSM  $\aCM = \conf{\sset,q_0,\lact,\tset}$ we define
$\mathsf{io}(\aCM) = \conf{\sset,q_0,\lio,\tset'}$ where
$\tset' = \Set{ q \trans[{}]{\mathsf{io}(\al)} q' \sst q
  \trans[{}]{\al} q' \in \aCM, \al\in\lact} \cup \Set{q
  \trans[{}]{\tau} q' \sst q \trans[{}]{\tau} q' \in \aCM}$ and
likewise for $\Dual{\aCM}$.

Informally, two CFSMs $\aCM_1$ and $\aCM_2$ are \emph{compatible} if
each output of $\aCM_1$ has a corresponding input in $\aCM_2$ and vice
versa once the identities of communicating partners are blurred away.
\begin{definition}[Compatibility]\label{def:compatibility}
  Let  $\aCM$ and $\aCM'$ be two FSAs on $\lio$.
  An \emph{io-correspondence} is a relation $R$ between states of
  $\aCM$ and those of $\aCM'$ such that whenever $(q, q') \in R$:
  \begin{itemize}
  \item $q \in \tstates{\aCM}$ if, and only if,
	 $q' \in \tstates{\aCM'}$ (cf. \cref{def:cfsm})
  \item if $q \trans{\aout[][][]} r \in \aCM$ then there is
	 $q' \trans{\ain[][][]} r' \in \aCM'$ such that $(r, r') \in R$
  \item if $q' \trans{\aout[][][]} r' \in \aCM'$ then there is
	 $q \trans{\ain[][][]} r \in \aCM$ such that $(r, r') \in R$
  \item if $q \trans \tau r \in \aCM$ then  $(r, q') \in R$
  \item if $q' \trans \tau r' \in \aCM'$ then  $(q, r') \in R$
  \end{itemize}
  Two CFSMs $\aCM$ and $\aCM'$ are \emph{compatible} (in symbols
  $\aCM \comp \aCM'$) if there is an io-correspondence relating the
  initial states of $\mathsf{io}(\aCM_1)$ and $\mathsf{io}(\aCM')$.
\end{definition}
\begin{example}[Compatibility]\label{ex:compatibility}
  The machines $\HH$ and $\KK$ of \cref{ex:gc} are compatible.
  For a more complex example, consider the following CFSMs
  
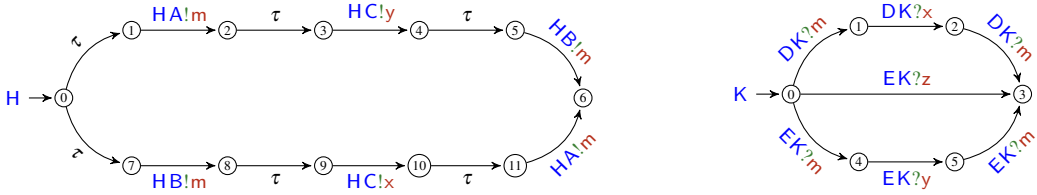
\[
	 \begin{tikzpicture}[mycfsm]
		\node[state, initial, initial where = left, initial distance = .4cm, initial text={$\HH$}] (zero) {$0$};
      \node[state] (one) [above right of=zero]   {$1$};
      \node[state] (two) [right  of=one]   {$2$};
      \node[state] (three) [right of=two]   {$3$};
      \node[state] (four) [right  of=three]   {$4$};
      \node[state] (five) [right  of=four]   {$5$};
      \node[state] (six) [below right of=five]   {$6$};
      \node[state] (seven) [below right of=zero]   {$7$};
      \node[state] (eight) [right  of=seven]   {$8$};
      \node[state] (nine) [right of=eight]   {$9$};
		\node[state] (ten) [right of=nine]   {$10$};
		\node[state] (eleven) [right of=ten]   {$11$};
      \path
      (zero) edge [bend left] node[above] {$\tau$} (one)
      (zero) edge   [bend right] node[below] {$\tau$} (seven)
      (one) edge node[above] {$\aout[H][A]$} (two)
      (two) edge    node[above] {$\tau$} (three)
      (three) edge   node[above] {$\aout[H][C][][y]$} (four)
      (four) edge    node[above] {$\tau$} (five) %
      (five) edge  [bend left]  node[above] {$\aout[H][B]$} (six)
      (seven) edge   node[below] {$\aout[H][B]$} (eight)
      (eight) edge  node[below] {$\tau$} (nine)
      (nine) edge  node[below] {$\aout[H][C][][x]$} (ten)
      (ten) edge  node[below] {$\tau$} (eleven)
      (eleven) edge   [bend right] node[below] {$\aout[H][A]$} (six)
      ;
	 \end{tikzpicture}
	 \qquad\qquad
	 \begin{tikzpicture}[mycfsm]
		\node[state, initial, initial where = left, initial distance = .4cm, initial text={$\KK$}] (zero) {$0$};
      \node[state] (one) [above right of=zero]   {$1$};
      \node[state] (two) [right  of=one]   {$2$};
      \node[state] (three) [below right of=two]   {$3$};
      \node[state] (four) [below right of=zero]   {$4$};
      \node[state] (five) [right  of=four]   {$5$};
      \path
      (zero) edge [bend left] node[above] {$\ain[D][K]$} (one)
      (one) edge node[above] {$\ain[D][K][][x]$} (two)
      (two) edge  [bend left]  node[above] {$\ain[D][K]$} (three)
      (zero) edge   [bend right] node[below] {$\ain[E][K]$} (four)
      (four) edge  node[below] {$\ain[E][K][][y]$} (five)
      (five) edge  [bend right]  node[below] {$\ain[E][K]$} (three)
		(zero) edge node[above] {$\ain[E][K][][z]$} (three)
      ;
	 \end{tikzpicture}
  \]
The above $\HH$ and $\KK$ are compatible. Apart for $\tau$ actions preceding them,
$\HH$ can only perform output actions, whereas $\KK$ can only perform input actions.
By disregarding the names of the receivers in the actions of  $\HH$, and of the senders in those of
$\KK$, any output action after its corresponding $\tau$ can find a matching input in $\KK$.
The vice versa does not hold, since none of the possible output actions that can occur
after a $\tau$ from $0$ (i.e. the outputs from $1$ and $7$ in $\HH$) can actually 
match  the input action $\ain[E][K][][z]$ from $0$ in $\KK$. Such a possibility is in fact allowed
by our definition of compatibility.
\finex
\end{example}
\cref{def:compatibility} transfers the notion of compatibility given
in~\cite{nostroJLAMP} for processes in multiparty sessions.
Also, \cref{def:compatibility} differs from the notions of
compatibility in~\cite{BLT20b} and
in~\cite{francoICEprev,BarbaneradH19} which are defined as
bisimulations and do not involve $\tau$-transitions.

\begin{definition}\label{inoutdet}
  An \p-local CFSM $\aCM$ is:
  \begin{enumerate}
  \item
  \label{inoutdet-in}
   {\em ?-deterministic} if $p \trans{\ain[x][a]} q$ and
	 $p \trans{\ain[y][a]} r \in \aCM$ implies $q = r$;
  \item 
  \label{inoutdet-out}
  {\em !-deterministic} if $p \trans{\tau}\, \trans{\aout[a][x]} q$ and
	 $p  \trans{\tau} \, \trans{\aout[a][y]} r \in \aCM$ implies $q = r$;
  \item {\em ?!-deterministic} if it is both ?-deterministic and
	 !-deterministic;
%  \item {\em mixed-deterministic} if $\msg \neq \msg[n]$ for all
%	 $p \trans{\ain[x][a]} q$ and
%	 $p \trans{\aout[a][y][][n]} r \in \aCM$. >>> NOT NEEDED CONDITION
\end{enumerate}
\end{definition}

A non-terminal state $q \in \aCM$ is \emph{asymmetric sending}
(resp. \emph{receiving}) if all its outgoing transitions have $\tau$
(resp. receiving) labels; $q$ is a \emph{asymmetric mixed} state if it
is neither asymmetric sending nor receiving.

\begin{example}
Machines $\HH$ and $\KK$ in \cref{ex:compatibility} are, respectively
non !-deterministic and non ?-deterministic. In particular, conditions
   (\ref{inoutdet-out}) and  (\ref{inoutdet-in}) of  \cref{inoutdet}  fail for, respectively,
   state $0$ of $\HH$ and state $0$ of $\KK$.
\end{example}

We require a stronger condition then compatibity for two systems to be
composable.
\begin{definition}[$(\HH,\KK)$-composability]
  Let $\aCS_1$ and $\aCS_2$ be two systems with disjoint domains.  We
  say that $\aCS_1$ and $\aCS_2$ are \emph{$(\HH,\KK)$-composable} if
  $\HH \in \dom{\aCS_1}$ and $\KK \in \dom{\aCS_2}$ are two compatible
  ?!-deterministic machines with no asymmetric mixed states.
\end{definition}

%%% Local Variables:
%%% mode: latex
%%% TeX-master: "main"
%%% End:

\section{Composition Related Issues}\label{sec:proppres}
It is known that under symmetric synchronisation composition spoils
deadlock-freedom; this is shown by the example below, borrowed
from~\cite{BLT20b}.
\begin{example}[Symmetric synchronisation spoils deadlock-freedom preservation]\label{ex:sync}%\label{rem:syncpresfails} 
  Take the systems
  \[
  \aCS_1 = \dboxed
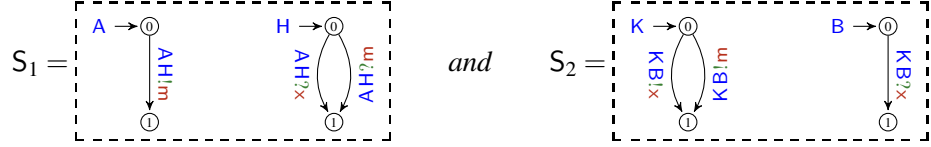
{
	 \begin{tikzpicture}[mycfsm]
		\node[state, initial, initial where = left, initial distance = .4cm, initial text={\p}] (zero) {$0$};
      \node[state] (one) [below of=zero] {$1$};
      \path (zero) edge node[above] {$\aout[A][H][][m]$} (one)
      ;
	 \end{tikzpicture}
	 \qquad\quad
	 \begin{tikzpicture}[mycfsm]
		\node[state, initial, initial where = left, initial distance = .4cm, initial text={$\HH$}] (zero) {$0$};
      \node[state] (one) [below of=zero] {$1$};
      \path
      (zero) edge[bend left] node[below] {$\ain[A][H][][m]$} (one)
      (zero) edge[bend right] node[below] {$\ain[A][H][][x]$} (one)
      ;
	 \end{tikzpicture}
  }
  \qqand
  \aCS_2 =
  \dboxed{
    \begin{tikzpicture}[mycfsm]
		\node[state, initial, initial where = left, initial distance = .4cm, initial text={$\KK$}] (zero) {$0$};
      \node[state] (one) [below of=zero]   {$1$};
      \path
      (zero) edge[bend left] node[below] {$\aout[K][B][][m]$} (one)
      (zero) edge[bend right] node[below] {$\aout[K][B][][x]$} (one)
      ;
	 \end{tikzpicture}
	 \qquad\quad
	 \begin{tikzpicture}[mycfsm]
		\node[state, initial, initial where = left, initial distance = .4cm, initial text={$\q$}] (zero) {$0$};
      \node[state] (one) [below of=zero] {$1$};
      \path (zero) edge node[above] {$\ain[K][B][][x]$} (one)
      ;
  \end{tikzpicture}
  }
  \]
  Clearly, $\aCS_1$ and $\aCS_2$ are $(\HH,\KK)$-composable and
  deadlock-free, yet their composition
  $\aCS = \aCS_1 \connect \HH \KK \aCS_2$ has a deadlock. In fact,
  when the gateway for $\KK$ receives $\msg$, it tries to
  synchronise with participant \q\ on message $\msg$ while \q\ is
  waiting only for $\msg[x]$.
  For $\aCS_2$ in isolation, this is not a deadlock, since \q\ and
  $\KK$ synchronise on $\msg[x]$ under the symmetric semantics.
  \finex
\end{example}
Notice that the counterexample of \cref{ex:sync} does not apply in an
asynchronous setting.
Indeed, $\aCS_2$ could deadlock due to the fact that $\KK$ could send
$\msg$ without synchronising with $\q$.
Likewise, the counterexample of \cref{ex:sync} does not apply in our
asymmetric setting.
Even if communication is still synchronous, the $\tau$-transitions
introduced to resolve internal choices (i.e., those prefixing outputs)
allow $\aCS_2$ to reach a deadlock configuration by choosing the
$\tau$-transition leading to the output $\aout[k][b]$.

\iftr
  Let $\aCM$ be an $\HH$-local CFSM and
  $\KK \in \ptpset \setminus \Set{\HH}$ be a participant not occurring
  in $\aCM$.
  Function $\nof[]{}$ maps the states of $\gateway \aCM \KK$ to the
  states of $\aCM$ as follows:
  \begin{align*}
	 \nof q =
	 &
		\begin{cases}
		  p, & \text{if } p \trans \al q \in \gateway \aCM \KK \text{ and } \al \text{ input label with }\KK \not\in \ptpof[\al]
		  \\
		  p, & \text{if } p \trans \al p' \trans \tau q \in \gateway \aCM \KK \text{ and } \al \text{ input label with }\KK \in \ptpof[\al]
		  \\
		  r, & \text{if } q \trans \al r \in \gateway \aCM \KK \text{ and } \al \text{ output label with }\KK \not\in \ptpof[\al]
		  \\
		  r, & \text{if } q \trans \tau q' \trans \al r \in \gateway \aCM \KK \text{ and } \al \text{ output label with }\KK \in \ptpof[\al]
		  \\
		  q,  & \text{if } q \text{ is a state of } \aCM
		\end{cases}
  \end{align*}
  Note that the first four clause imply that $q$ is a fresh state
  of $\gateway \aCM \KK$.
  
%   \hsl[old def]

%     Let $\gateway{\aCM_\HH}\KK$ be a gateway extracted from an
%   $\HH$-local CFSM.
%   %
%   Function $\nof[]{}$ maps the states of $\gateway{\aCM_\HH}\KK$ to the
%   states of $\aCM_\HH$ as follows:
%   \[  \nof{q} =
%   \begin{cases}
%          q & \text{ if } q \text{ is external }
%          \\
%          q' & \text{ if } q \text{ is internal and } q' \trans{\ain[A][H][][m]} q \in \gateway{\aCM_\HH}\KK \text{ for some } \p,\msg
%          \\
%          q' & \text{ if } q \text{ is internal and } q' \trans{\ain[A][H][][m]} q''\trans{\tau}q \in \gateway{\aCM_\HH}\KK \text{ for some } \p,\msg
%          \\
%          q' & \text{ if } q \text{ is internal and }  q \trans{\aout[H][A][][m]} q' \in \gateway{\aCM_\HH}\KK \text{ for some } \p,\msg
%           \\
%          q' & \text{ if } q \text{ is internal and }  q  \trans{\tau} q'' \trans{\aout[H][A][][m]} q' \in \gateway{\aCM_\HH}\KK \text{ for some } \p,\msg
%   \end{cases}\]
%   where $\p\neq\KK$.

% \hsl[end old def]

\begin{lemma}\label{lem:nof}%\label{cref:lem}
  Function $\nof[]{}$ is well-defined.
\end{lemma}
\begin{proof}
  Let $\aCM$ be an $\HH$-local CFSM and
  $\KK \in \ptpset \setminus \Set{\HH}$.
  We have to check only internal states as the restriction of $\nof[]{}$
  to the states of $\aCM$ is the identity by definition.
  If $q$ is an internal state of $\gateway \aCM \KK$, by definition of
  $\gateway \aCM \KK$, there is a unique $q'$ such that either
  $q' \trans{\ain[A][H][][m]} q \in \gateway \aCM \KK$ or
  $q \trans{\aout[H][A][][m]} q' \in \gateway \aCM \KK$.
\end{proof}

\begin{example}
  For $\aCS_1 \connect{\HH}{\KK} \aCS_2$ of \cref{ex:gc},
  $\nof[\HH,\KK][\aCS_1(\HH)] 0 = 0$ and
  $\nof[\KK,\HH][{\aCS_2(\KK)}] 2 = 0$.
  \finex
\end{example}

Function $\nof[]{}$ is similar to configuration projection when
considering CFSMs in isolation, with a main difference: when e.g.,
$\gateway \aCM \KK$ in state $q$ receives from a partner in its own
system going to some fresh state $p'$ with a $\tau$-transition to
$p''$, $\nof[]{}$ maps both $p'$ and $p''$ to $p$ (unlike configuration
projection $\cproj \_ \HH$ which maps $q$ to the successor of $p''$).
This represents the fact that the other system, and $\KK$ in
particular, are oblivious of the transition.
In fact, function $\nof[]{}$ is designed to establish a correspondence
with the other system as shown by the next proposition.

\begin{proposition}\label{prop:comp}
  Let $\aCS = \aCS_1 \connect \HH \KK \aCS_2$ be the composition of
  two $(\HH,\KK)$-composable systems $\aCS_1$ and $\aCS_2$.
  If $\aConf \in \RS[\ssem{\aCS}]$ then
  $\cproj \aConf \HH \in \RS[\ssem{\aCS_1}]$,
  $\cproj \aConf \KK \in \RS[\ssem{\aCS_2}]$, and
  $\nof[\HH,\KK][\aCS_1(\HH)]{\aConf(\HH)} \comp
  \nof[\KK,\HH][\aCS_2(\KK)]{\aConf(\KK)}$.
\end{proposition}
\begin{proof}
  Let $\aConf_0$ be the initial configuration of $\aCS$ and
  \begin{align}\label{eq:run}
	 \aConf_0 \trans{\al_1} \aConf_1 \cdots \aConf_{n-1} \trans{\al_n} \aConf_n = \aConf
  \end{align}
  a run reaching $\aConf$ from $\aConf_0$.
  We proceed by induction on $n$.

  If $n = 0$ the thesis is immediate by observing that
  \begin{itemize}
  \item $\cproj \aConf \HH \in \RS[\ssem{\aCS_1}]$ and
	 $\cproj \aConf \KK \in \RS[\ssem{\aCS_2}]$ because left- and
	 right-projections are the initial configurations of $\aCS_1$ and
	 $\aCS_2$ respectively, and
  \item $\nof[\HH,\KK][\aCS_1(\HH)]{\aConf(\HH)}$ and
	 $\nof[\KK,\HH][\aCS_2(\KK)]{\aConf(\KK)}$ are the initial states
	 of $\aCS_1(\HH)$ and $\aCS_2(\KK)$ respectively which are
	 compatible by hypothesis.
  \end{itemize}

  Let $n > 0$ and assume that the statement holds for all
  configurations reachable from $\aConf_0$ in less the $n$
  transitions.
  We have that either none of $\HH$ and $\KK$ are involved in
  $\al_n$ or that at least one of them is.
  In the former case, without loss of generality, assume that the
  interacting participants, say \p\ and \q\, are both in $\aCS_1$.
  Then, by construction (cf. \cref{def:cproj}),
  $\cproj \aConf \KK = \cproj{\aConf_{n-1}} \KK$ and by inductive
  hypothesis $\cproj{\aConf_{n-1}} \KK \in \RS[\ssem{\aCS_2}]$.
  Moreover, $\cproj \aConf \HH$ equals $\cproj{\aConf_{n-1}} \HH$ but
  for the local states of \p\ and \q; hence
  $\cproj{\aConf_{n-1}} \HH \in \RS[\ssem{\aCS_1}]$ by the semantics
  of communicating systems (cf. \cref{def:syncSem}).
  Also,
  \[
	 \nof[\HH,\KK][\aCS_1(\HH)]{\aConf(\HH)}=
	 (\cproj \aConf \HH)(\HH) = (\cproj{\aConf_{n-1}} \HH)(\HH) \comp
	 (\cproj{\aConf_{n-1}} \KK)(\KK) =(\cproj \aConf \KK)(\KK)
	 = \nof[\KK,\HH][\aCS_2(\KK)]{\aConf(\KK)}
  \]
  because the equalities above hold by the definition of asymmetric
  synchronisation (cf. \cref{def:syncSem}) and the compatibility
  relation holds by the inductive hypothesis.
  So, let us assume that at least one between $\HH$ and $\KK$ is
  involved in the last transition reaching $\aConf$ and proceed by
  case analysis on $\al_n$.
  \begin{description}
  \item[\fbox{$\al_n = \tau$}] We consider only the case where $\HH$ is involved since
	 the case where $\KK$ is involved is symmetric.
	 By our gateway construction (cf. \cref{def:gateway}) only one of
	 the following two cases are possible for the transitions in
	 $\gateway{\aCS_1(\HH)} \KK$:
	 \begin{align}
		\label{eq:fromK}
		p \trans{\ain[K][H]} \aConf_{n-1}(\HH) \trans{\tau} q \trans{\aout[H][A]} r
		\\
		\label{eq:fromA}
		q \trans{\ain[A][H]} \aConf_{n-1}(\HH)  \trans{\tau} p \trans{\aout[H][K]} r
	 \end{align}
	 for some participant \p\ of $\aCS_1$ and states $p$, $q$, $r$ of
	 $\gateway{\aCS_1(\HH)} \KK$.
	 % with $p$ fresh in~\eqref{eq:fromA} (i.e., $p$ not in
	 % $\aCS_1(\HH)$).
	 %
	 Cases~\eqref{eq:fromK} and~\eqref{eq:fromA}
	 respectively correspond to have the transitions
	 \[
		p \trans{\tau} q \trans{\aout[H][A]} r
		\qquad\qand\qquad
		q \trans{\ain[A][H]} r
	 \]
	 in the machine $\aCS_1(\HH)$.
	  
	 In case \eqref{eq:fromK} the last transition in the
	 run~\eqref{eq:run} must therefore be preceded by a transition, say
	 the $i$-th one, where the machines $\gateway{\aCS_1(\HH)} \KK$ and
	 $\gateway{\aCS_2(\KK)} \HH$ have exchanged message $\msg$.
	 Hence, we have that $\gateway{\aCS_2(\KK)} \HH$ has a transition
	 $\aConf_i(\KK) \trans{\aout[k][h]} \aConf_{i+1}(\KK)$ with
	 $\aConf_{n-1}(\HH) \comp \aConf_{i+1}(\KK)$ since (by inductive
	 hypothesis)
	 $\aConf_i(\HH) = p = \nof[\HH,\KK][\aCS_1(\HH)]{\aConf_i(\HH)}
	 \comp \nof[\KK,\HH][\aCS_2(\KK)]{\aConf_i(\KK)}$ and both gateways
	 are $?!$-deterministic.
	 We now observe that either $\aConf(\KK) = \aConf_{i+1}(\KK)$ or
	 it is a fresh state that $\gateway{\aCS_2(\KK)} \HH$ reaches after
	 having received a message from a participant of $\aCS_2$ (possibly
	 followed by a $\tau$ transition) between the $i$-th and the last
	 interactions in~\eqref{eq:run}.
	 We have that
	 $\cproj \aConf \HH = \upd{\cproj {\aConf_{n-1}} \HH}{\HH}{p}$,
	 hence $\cproj \aConf \HH \in \RS[\ssem{\aCS_1}]$ by inductive
	 hypothesis; also, $\cproj \aConf \KK \in \RS[\ssem{\aCS_2}]$ by
	 the inductive hypothesis since
	 $\cproj \aConf \KK = \cproj {\aConf_{n-1}} \KK$ because $\KK$ is
	 not involved in $\al_n$.
	 Finally, in both cases the last part of the thesis immediately
	 follows by observing that
	 $\nof[\KK,\HH][\aCS_2(\KK)]{\aConf(\KK)} =
	 \nof[\KK,\HH][\aCS_2(\KK)]{\aConf_i(\KK)}$ by definition and
	 $\nof[\HH,\KK][\aCS_1(\HH)]{\aConf(\HH)} = p$.
		
	 In case~\eqref{eq:fromA}, firstly note that by construction
	 $\p \neq \HH$ (cf. \cref{def:syncSem}).
	 Then $\gateway{\aCS_1(\HH)} \KK$ has the transition
	 $\aConf_{n-1}(\HH) \trans{\ain[a][h]} \aConf(\HH)$ which
	 corresponds to an input transition
	 $\aConf_{n-1}(\HH) \trans{\ain[a][h]} r$ where
	 $r = \cproj \aConf \HH (\HH) =
	 \nof[\HH,\KK][\aCS_1(\HH)]{\aConf_{n-1}(\HH)} =
	 \nof[\HH,\KK][\aCS_1(\HH)]{\aConf(\HH)}$ in $\aCS_1(\HH)$.
	 Hence, the thesis follows since
	 $\cproj \aConf \KK = \cproj {\aConf_{n-1}} \KK$ because $\KK$ is
	 not involved in $\al_n$ and
	 $\nof[\HH,\KK][\aCS_1(\HH)]{\aConf(\HH)} \comp
	 \nof[\KK,\HH][\aCS_2(\KK)]{\aConf(\KK)} =
	 \nof[\KK,\HH][\aCS_2(\KK)]{\aConf_{n-1}(\KK)}$ by inductive
	 hypothesis.

  \item[\fbox{$\al_n = \gint[][h][m][x]$}] By construction, either
	 $\p[x] = \KK$ or $\p[x] \neq \HH$ is a participant of $\aCS_1$.
	 In the former case, $\gateway{\aCS_1(\HH)} \KK$ has the transition
	 $\aConf_{n-1}(\HH) \trans{\aout[h][k]} \aConf(\HH)$ which
	 corresponds to an input transition while $\gateway{\aCS_2(\KK)}
	 \HH$ has the transition
	 $\aConf_{n-1}(\KK) \trans{\ain[h][k]} \aConf(\KK)$ which
	 corresponds to a sequence of transitions
	 $p \trans \tau \trans{\aout[k][b]} \aConf(\KK)$ in $\aCS_2(\KK)$
	 for some participant $\q \neq \KK$ in $\aCS_2$.
	 Then the thesis follows by the fact that
	 $\cproj \aConf \HH = \upd{\cproj {\aConf_{n-1}} \HH} \HH
	 {\aConf(\HH)}$ and
	 $\cproj \aConf \KK = \upd{\cproj {\aConf_{n-1}} \KK} \KK p$ by
	 construction (cf. \cref{def:cproj}) and that, by inductive
	 hypothesis,
	 $\nof[\HH,\KK][\aCS_1(\HH)]{\aConf_{n-1}(\HH)} = q \comp
	 \aConf(\KK) = \nof[\KK,\HH][\aCS_2(\KK)]{\aConf_{n-1}(\KK)}$ and
	 therefore, by ?!-determinism and the compatibility relation
	 $\aConf(\HH) \comp \aConf(\KK)$.

	 Suppose now that $\p[x]$ is a participant of $\aCS_1$; note that
	 by construction $\p[x] \neq \HH$ (cf. \cref{def:syncSem}).
	 Since $\cproj \aConf \KK = \cproj {\aConf_{n-1}} \KK$, the inductive
	 hypothesis immediately entails that $\cproj \aConf \KK \in \RS[\ssem{\aCS_2}]$.

	 We first show the reachability of left- and right-projections.
	 The  transition
	 $\aConf_{n-1}(\HH) \trans{\aout[h][x]} \aConf(\HH)$ is in
	 $\gateway{\aCS_1(\HH)} \KK$ by construction and it corresponds to
	 a pair of transitions
	 $p \trans \tau \aConf_{n-1}(\HH) \trans{\aout[h][x]} \aConf(\HH)$
	 in $\aCS_1(\HH)$.
	 We have that $\cproj {\aConf_{n-1}} \HH \in \RS[\ssem{\aCS_1}]$
	 (by inductive hypothesis) and since
	 $\cproj {\aConf_{n-1}} \HH(\HH) = p$ (by \cref{def:cproj}) we have
	 $\cproj {\aConf_{n-1}} \HH \trans \tau \trans{\gint[][h][m][x]} \cproj \aConf \HH$ (by \cref{def:syncSem}).

	 We now show the compatibility condition.
	 The last transition in the run~\eqref{eq:run} must be preceded by
	 a transition, say the $i$-th one, where the machines
	 $\gateway{\aCS_1(\HH)} \KK$ and $\gateway{\aCS_2(\KK)} \HH$ have
	 exchanged message $\msg$.
	 Hence, we have that $\gateway{\aCS_2(\KK)} \HH$ has a transition
	 $\aConf_i(\KK) \trans{\aout[k][h]} \aConf_{i+1}(\KK)$ with
	 \begin{align}\label{eq:nof}
		  \nof[\HH,\KK][\aCS_1(\HH)]{\aConf_{n-1}(\HH)}  = \aConf(\HH)
		  = \nof[\HH,\KK][\aCS_1(\HH)]{\aConf(\HH)}
	 \end{align}
	 which hold by definition of $\nof[{}][{}]{\_}$.
	 We now observe that either $\aConf(\KK) = \aConf_{i+1}(\KK)$ or
	 it is a fresh state that $\gateway{\aCS_2(\KK)} \HH$ reaches after
	 having received a message from a participant of $\aCS_2$ (possibly
	 followed by a $\tau$ transition) between the $i$-th and the last
	 interactions in~\eqref{eq:run}.
	 In both cases the
	 $\nof[\KK,\HH][\aCS_2(\KK)]{\aConf(\KK)} =
	 \nof[\KK,\HH][\aCS_2(\KK)]{\aConf_{n-1}(\KK)}$ and, by inductive
	 hypothesis,
	 \[
		\nof[\KK,\HH][\aCS_2(\KK)]{\aConf_{n-1}(\KK)} \comp
		\nof[\HH,\KK][\aCS_1(\HH)]{\aConf_{n-1}(\HH)}
	 \]
	 hence, by equalities~\eqref{eq:nof},
	 $\nof[\HH,\KK][\aCS_1(\HH)]{\aConf(\HH)} \comp
	 \nof[\KK,\HH][\aCS_2(\KK)]{\aConf(\KK)} $.
  \item[\fbox{$\al_n = \gint[][x][m][h]$}] The case $\p[x] = \KK$ is
	 symmetric to the previous case with $\al_n = \gint[][h][m][k]$ and
	 therefore omitted.
	 So, assume that $\p[x]$ is a participant of $\aCS_1$; note that by
	 construction $\p[x] \neq \HH$ (cf. \cref{def:syncSem}).
	 Then $\gateway{\aCS_1(\HH)} \KK$ has the transition
	 $\aConf_{n-1}(\HH) \trans{\ain[x][h]} \aConf(\HH)$ which
	 corresponds to an input transition
	 $\aConf_{n-1}(\HH) \trans{\ain[x][h]} r$.
	 % $ =\nof[\HH,\KK][\aCS_1(\HH)]{\aConf_{n-1}(\HH)}=	 \nof[\HH,\KK][\aCS_1(\HH)]{\aConf(\HH)}$ in $\aCS_1(\HH)$.
	 %
	 We have that $\cproj {\aConf_{n-1}} \HH \in RS[\ssem{\aCS_1}]$ by
	 inductive hypothesis, and so is
	 $\cproj {\aConf_{n-1}} \HH \trans{\gint[][x][m][h]} \cproj \aConf
	 \HH$ since, by definition of left-projection,
	 $\cproj \aConf \HH(\HH) = r$.
	 The reachability of the right-projection of $\aConf$ immediately
	 follows by inductive hypothesis, since
	 $\cproj \aConf \KK = \cproj {\aConf_{n-1}} \KK$.

	 We now show the compatibility condition.
	 By definition, we have that
	 $\nof[\HH,\KK][\aCS_1(\HH)]{\aConf(\HH)} =
	 \nof[\HH,\KK][\aCS_1(\HH)]{\aConf_{n-1}(\HH)}$.
	 Moreover, we necessarily have that $\aConf(\KK)=
	 \aConf_{n-1}(\KK)$.
	 Hence by inductive hypothesis,
	 $$\nof[\HH,\KK][\aCS_1(\HH)]{\aConf(\HH)} = \nof[\HH,\KK][\aCS_1(\HH)]{\aConf_{n-1}(\HH)} \comp 
	 \nof[\KK,\HH][\aCS_2(\KK)]{\aConf_{n-1}(\KK)}= \nof[\KK,\HH][\aCS_2(\KK)]{\aConf(\KK)}$$
%	 
%	 $\nof[\HH,\KK][\aCS_1(\HH)]{\aConf(\HH)} \comp \nof[\KK,\HH][\aCS_2(\KK)]{\aConf(\KK)}
%	 = \nof[\KK,\HH][\aCS_2(\KK)]{\aConf_{n-1}(\KK)}$.
  \end{description}
  The cases $\al_n = \gint[][x][m][k]$ and $\al_n = \gint[][k][m][x]$
  are similar to the last two cases above.
\end{proof}

\fi

Now, one may think that analogously to what happens
in~\cite{francoICEprev,BarbaneradH19}, if two systems are
$(\HH,\KK)$-composable and deadlock-free then their composition is
deadlock-free too.

%
%It is immediate to check that $\gateway{M}{\HH}$ is sequential if $M$ is so.
%Moreover, trivially, a sequential $M$ is also $?!$-deterministic
%and with no mixed state (and hence mixed-deterministic).

In \cref{sec:pbc} we shall prove that in our setting lock-freedom is
preserved by composition, without any further condition beside
$(\HH,\KK)$-composability.
Before doing that, we give examples showing the necessity of our conditions for  
deadlock freedom preservation.

Let us begin with compatibility.
Properties cannot be preserved under composition without
compatibility, as shown in the next example.
\begin{example}[Lack of compatibility spoils deadlock freedom preservation]
  The systems
  \[
	 \aCS_1 =\dboxed{
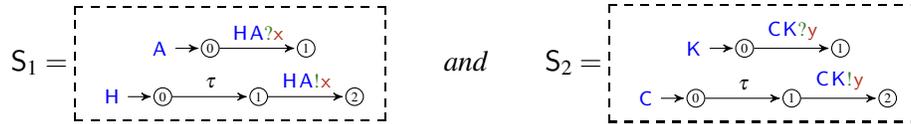
		\begin{array}{c}
		  \begin{tikzpicture}[mycfsm]
			 \node[state, initial, initial where = left, initial distance = .4cm, initial text={$\p$}] (zero) {$0$};
			 \node[state] (one) [right of=zero]   {$1$};
			 \path (zero) edge node[above] {$\ain[H][A][][x]$} (one)
			 ;
		  \end{tikzpicture}
		  \\
		  \begin{tikzpicture}[mycfsm]
			 \node[state, initial, initial where = left, initial distance = .4cm, initial text={$\HH$}] (zero) {$0$};
			 \node[state] (one) [right of=zero]   {$1$};
			 \node[state] (two) [right of=one]   {$2$};
			 \path
			 (zero) edge node[above] {$\tau$} (one)
			 (one) edge node[above] {$\aout[H][A][][x]$} (two)
			 ;
		  \end{tikzpicture}
		\end{array}
	 }
  \qquad\text{and}\qquad
  \aCS_2 =\dboxed{
	 \begin{array}[c]{c}
		\begin{tikzpicture}[mycfsm]
		  \node[state, initial, initial where = left, initial distance = .4cm, initial text={$\KK$}] (zero) {$0$};
		  \node[state] (one) [right of=zero]   {$1$};
		  \path (zero) edge node[above] {$\ain[C][K][][y]$} (one)
		  ;
		\end{tikzpicture}
		\\
		\begin{tikzpicture}[mycfsm]
		  \node[state, initial, initial where = left, initial distance = .4cm, initial text={$\ptp[C]$}] (zero) {$0$};
		  \node[state] (one) [right of=zero]   {$1$};
		  \node[state] (two) [right of=one]   {$2$};
		  \path
		  (zero) edge node[above] {$\tau$} (one)
		  (one) edge node[above] {$\aout[C][K][][y]$} (two)
		  ;
		\end{tikzpicture}
	 \end{array}
  }
  \]
  are trivially deadlock free.
  However, $\HH$ and $\KK$ are not compatible, since there is no
  corresponding input in $\KK$ for the output from $\HH$.
  The composition of $\aCS_1$ and $\aCS_2$  via $\HH$ and $\KK$ yields
  \[
	 \csys[c@{\hspace{1cm}}c]
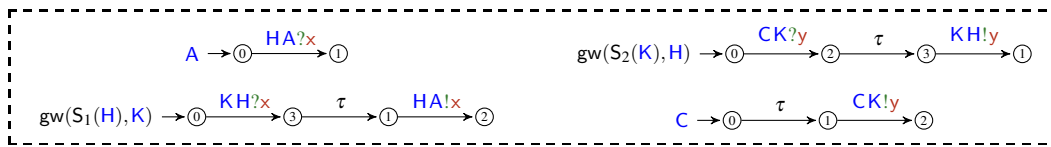
{
		\begin{tikzpicture}[mycfsm]
		  \node[state, initial, initial where = left, initial distance = .4cm, initial text={$\p$}] (zero) {$0$};
		  \node[state] (one) [right of=zero]   {$1$};
		  \path(zero) edge node[above] {$\ain[H][A][][x]$} (one)
		  ;
		\end{tikzpicture}
		&
		\begin{tikzpicture}[mycfsm]
		  \node[state, initial, initial where = left, initial distance = .4cm, initial text={$\gateway {\aCS_2(\KK)} \HH$}] (zero) {$0$};
		  \node[state] (two) [right of=zero]   {$2$};
		  \node[state] (three) [right of=two]   {$3$};
		  \node[state] (one) [right of=three]   {$1$};
		  \path
		  (zero) edge node[above] {$\ain[C][K][][y]$} (two)
		  (two) edge node[above] {$\tau$} (three)
		  (three) edge node[above] {$\aout[K][H][][y]$} (one)
		  ;
		\end{tikzpicture}
		\\[1em]
		\begin{tikzpicture}[mycfsm]
		  \node[state, initial, initial where = left, initial distance = .4cm, initial text={$\gateway {\aCS_1(\HH)} \KK$}] (zero) {$0$};
		  \node[state] (three) [right of=zero]   {$3$};
		  \node[state] (one) [right of=three]   {$1$};
		  \node[state] (two) [right of=one]   {$2$};
		  \path
		  (zero) edge node[above] {$\ain[K][H][][x]$} (three)
		  (three) edge node[above] {$\tau$} (one)
		  (one) edge node[above] {$\aout[H][A][][x]$} (two)
		  ;
		\end{tikzpicture}
		&
		\begin{tikzpicture}[mycfsm]
		  \node[state, initial, initial where = left, initial distance = .4cm, initial text={$\ptp[C]$}] (zero) {$0$};
		  \node[state] (one) [right of=zero]   {$1$};
		  \node[state] (two) [right of=one]   {$2$};
		  \path
		  (zero) edge node[above] {$\tau$} (one)
		  (one) edge node[above] {$\aout[C][K][][y]$} (two)
		  ;
		\end{tikzpicture}
	 }
  \]
  Starting from the initial configuration of
  $\aCS_1 \connect \HH \KK \aCS_2$, the following transitions are
  possible in $\ssem{\aCS_1 \connect \HH \KK \aCS_2}$
  \[
	 (0_{\p},0_{\HH},0_{\KK},0_{\ptp[C]}) 
	 \trans{\tau} (0_{\p},0_{\HH},0_{\KK},1_{\ptp[C]})
	 \trans{\gint[][C][y][K]} (0_{\p},0_{\HH},2_{\KK},2_{\ptp[C]})
	 \trans{\tau} (0_{\p},0_{\HH},3_{\KK},2_{\ptp[C]}) \not\trans{}
  \]
  where $(0_{\p},0_{\HH},3_{\KK},2_{\ptp[C]})$ is a deadlock
  configuration for $\ssem{\aCS_1 \connect \HH \KK \aCS_2}$ since
  $\KK$ wishes to send $\msg[y]$ to $\HH$, which is instead waiting
  for message $\msg[x]$.
  \finex
\end{example}

The following example, casting in our setting an example given
in~\cite{BarbaneradH19} for the asynchronous semantics, illustrates
that asymmetric mixed states must be avoided to preserve properties
under composition.
\begin{example}[Asymmetric mixed-states spoil deadlock freedom preservation]
  % Let us consider the following CFSMs containing mixed-states:
  % $0_\HH$ for $\aCM_\HH$ and $0_\KK$ for $\aCM_\KK$.
  % Besides, let us assume 
  % $\HH\in\ptpof[\aCS_1]$, $\KK\in\ptpof[\aCS_2]$,
  % for certain $\aCS_1$ and $\aCS_2$.
  Let $\aCS_1$ and $\aCS_2$ be 
  $$
  \aCS_1 =\dboxed{
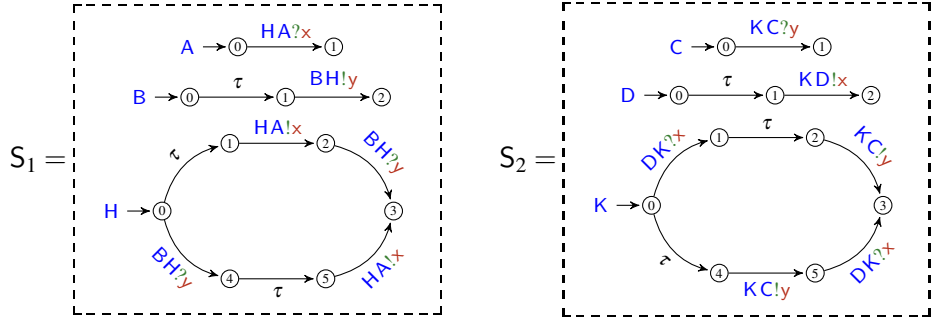
		\begin{array}{c}
		  \begin{tikzpicture}[mycfsm]
			 \node[state, initial, initial where = left, initial distance = .4cm, initial text={$\p$}] (zero) {$0$};
			 \node[state] (one) [right of=zero]   {$1$};
			 \path (zero) edge node[above] {$\ain[H][A][][x]$} (one)
			 ;
		  \end{tikzpicture}
		  \\
		  \begin{tikzpicture}[mycfsm]
			 \node[state, initial, initial where = left, initial distance = .4cm, initial text={$\q$}] (zero) {$0$};
			 \node[state] (one) [right of=zero]   {$1$};
			 \node[state] (two) [right of=one]   {$2$};
			 \path
			 (zero) edge node[above] {$\tau$} (one)
			 (one) edge node[above] {$\aout[B][H][][y]$} (two)
			 ;
		  \end{tikzpicture}
		  \\
		  \begin{tikzpicture}[mycfsm]
		\node[state, initial, initial where = left, initial distance = .4cm, initial text={$\HH$}] (zero) {$0$};
		\node[state] (one) [above right of=zero]   {$1$};
		\node[state] (two) [right of=one]   {$2$};
	     \node[state] (three) [below right of=two]   {$3$};
	     \node[state] (four) [below right of=zero]   {$4$};
	     \node[state] (five) [right of=four]   {$5$};
		\path
		(zero) edge [bend left] node[above] {$\tau$} (one)
		(one) edge node[above] {$\aout[H][A][][x]$} (two)
		(zero) edge [bend right] node[below] {$\ain[B][H][][y]$} (four)
		(two) edge [bend left]   node[above] {$\ain[B][H][][y]$} (three)
		(four) edge  node[below] {$\tau$} (five)
		(five) edge [bend right]  node[below] {$\aout[H][A][][x]$} (three)
		;
	 \end{tikzpicture}
		\end{array}
	 }
	 \qquad
	 \aCS_2 =\dboxed{
	 \begin{array}[c]{c}
		\begin{tikzpicture}[mycfsm]
		  \node[state, initial, initial where = left, initial distance = .4cm, initial text={$\ptp[C]$}] (zero) {$0$};
		  \node[state] (one) [right of=zero]   {$1$};
		  \path (zero) edge node[above] {$\ain[K][C][][y]$} (one)
		  ;
		\end{tikzpicture}
		\\
		\begin{tikzpicture}[mycfsm]
		  \node[state, initial, initial where = left, initial distance = .4cm, initial text={$\ptp[D]$}] (zero) {$0$};
		  \node[state] (one) [right of=zero]   {$1$};
		  \node[state] (two) [right of=one]   {$2$};
		  \path
		  (zero) edge node[above] {$\tau$} (one)
		  (one) edge node[above] {$\aout[K][D][][x]$} (two)
		  ;
		\end{tikzpicture}
		\\
		\begin{tikzpicture}[mycfsm]
		  \node[state, initial, initial where = left, initial distance = .4cm, initial text={$\KK$}] (zero) {$0$};
		\node[state] (one) [above right of=zero]   {$1$};
		\node[state] (two) [right of=one]   {$2$};
	     \node[state] (three) [below right of=two]   {$3$};
	     \node[state] (four) [below right of=zero]   {$4$};
	     \node[state] (five) [right of=four]   {$5$};
		\path
		(zero) edge [bend left] node[above] {$\ain[D][K][][x]$} (one)
		(one) edge node[above] {$\tau$} (two)
		(zero) edge [bend right] node[below] {$\tau$} (four)
		(two) edge [bend left]   node[above] {$\aout[K][C][][y]$} (three)
		(four) edge  node[below] {$\aout[K][C][][y]$} (five)
		(five) edge [bend right]  node[below] {$\ain[D][K][][x]$} (three)
		;
		\end{tikzpicture}
	 \end{array}
  }
	 $$
	 Noticeably, the initial states are asymmetric mixed and $\HH$
	 and $\KK$ are compatible.
	 The gateways are
    
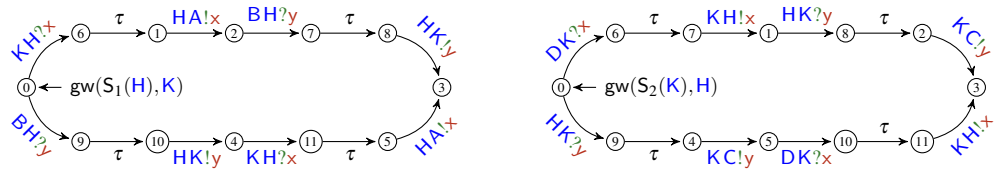
\begin{align*}
	 \begin{tikzpicture}[mycfsm, node distance = 1.2cm]
		\node[state, initial, initial where = right, initial distance = .4cm, initial text={$\gateway {\aCS_1(\HH)} \KK$}] (zero) {$0$};
	     \node[state] (six) [above right of=zero]   {$6$};
	     \node[state] (nine) [below right of=zero]   {$9$};
	     \node[state] (ten) [right of=nine]   {$10$};
		\node[state] (one) [right of=six]   {$1$};
		\node[state] (two) [right of=one]   {$2$};
		\node[state] (seven) [right of=two]   {$7$};
		\node[state] (eight) [right of=seven]   {$8$};
	     \node[state] (three) [below right of=eight]   {$3$};
	     \node[state] (four) [right of=ten]   {$4$};
	     \node[state] (eleven) [right of=four]   {$11$};
	     \node[state] (five) [right of=eleven]   {$5$};
		\path
		(zero) edge [bend left] node[above] {$\ain[K][H][][x]$} (six)
		(one) edge node[above] {$\aout[H][A][][x]$} (two)
		(zero) edge [bend right] node[below] {$\ain[B][H][][y]$} (nine)
		(two) edge   node[above] {$\ain[B][H][][y]$} (seven)
		(eleven) edge  node[below] {$\tau$} (five)
		(four) edge  node[below] {$\ain[K][H][][x]$} (eleven)
		(five) edge [bend right]  node[below] {$\aout[H][A][][x]$} (three)
		(six) edge   node[above] {$\tau$} (one)
		(seven) edge   node[above] {$\tau$} (eight)
		(nine) edge   node[below] {$\tau$} (ten)			
		(eight) edge [bend left]  node[above] {$\aout[H][K][][y]$} (three)			
		(ten) edge  node[below] {$\aout[H][K][][y]$} (four)
		;
	 \end{tikzpicture}
		\qquad
	 \begin{tikzpicture}[mycfsm, node distance = 1.2cm]
		\node[state, initial, initial where = right, initial distance = .4cm, initial text={$\gateway {\aCS_2(\KK)} \HH$}] (zero) {$0$};
		\node[state] (six) [above right of=zero]   {$6$};
		\node[state] (nine) [below right of=zero]   {$9$};
		\node[state] (seven) [right of=six]   {$7$};
		\node[state] (one) [ right of=seven]   {$1$};
		\node[state] (eight) [right of=one]   {$8$};
		\node[state] (two) [right of=eight]   {$2$};
	     \node[state] (three) [below right of=two]   {$3$};
	     \node[state] (four) [ right of=nine]   {$4$};
	     \node[state] (five) [right of=four]   {$5$};
	     \node[state] (ten) [right of=five]   {$10$};
	     \node[state] (eleven) [right of=ten]   {$11$};
		\path
		(zero) edge [bend left] node[above] {$\ain[D][K][][x]$} (six)
		(zero) edge [bend right] node[below] {$\ain[H][K][][y]$} (nine)
		(six) edge node[above] {$\tau$} (seven)
		(seven) edge node[above] {$\aout[K][H][][x]$} (one)
		(one) edge node[above] {$\ain[H][K][][y]$} (eight)
		(eight) edge node[above] {$\tau$} (two)
		(nine) edge  node[below] {$\tau$} (four)
		(two) edge [bend left]   node[above] {$\aout[K][C][][y]$} (three)
		(four) edge  node[below] {$\aout[K][C][][y]$} (five)
		(five) edge  node[below] {$\ain[D][K][][x]$} (ten)
		(ten) edge node[above] {$\tau$} (eleven)
		(eleven) edge [bend right] node[below] {$\aout[K][H][][x]$} (three)
		;
	 \end{tikzpicture}
	 \end{align*}
  The composed system $\aCS_1 \connect \HH \KK \aCS_2$ deadlocks when
  $\gateway {\aCS_1(\HH)} \KK$ receives from \q\ while
  $\gateway {\aCS_2(\KK)} \HH$ receives from \p[d] since both gateways
  reach an output state (respectively states $10$ and $7$).
  \finex
\end{example}

The following example shows that%, as asymmetric mixed states,
!?-nondeterminism is problematic too.
%
% Let us first take two deadlock free systems.
\begin{example}\label{ex:dlfree}
  Consider the two deadlock-free systems
  \begin{align*}
	 \aCS_1 =
	 \csys[c@{\hspace{.5cm}}c]{
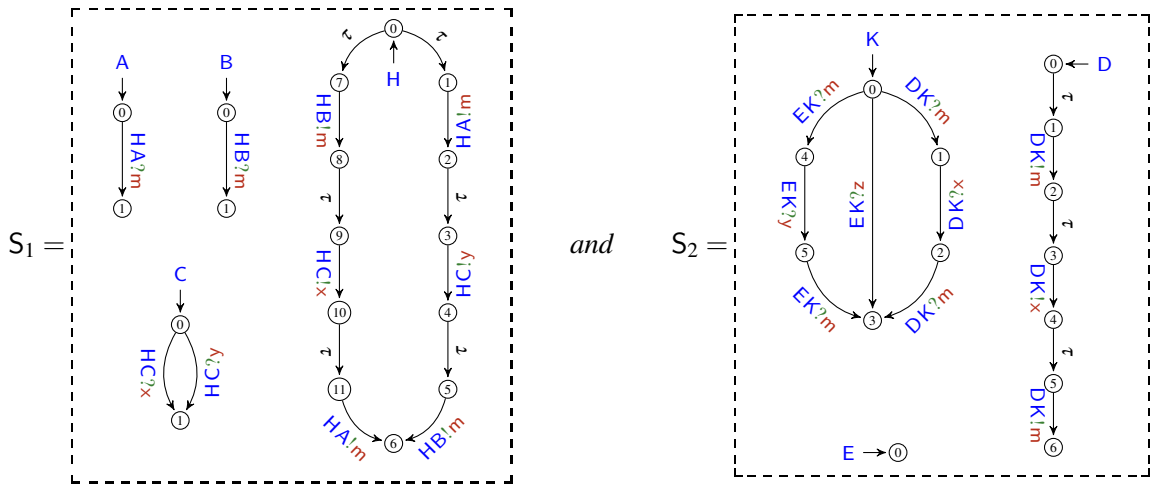
	 \begin{array}[c]{c}
		\begin{tikzpicture}[mycfsm]
		  \node[state, initial, initial where = above, initial distance = .4cm, initial text={\p}] (zero) {$0$};
		  \node[state] (one) [below of=zero]   {$1$};
		  \path (zero) edge node[above] {$\ain[H][A][][m]$} (one)
		  ;
		\end{tikzpicture}
		\qquad
		\begin{tikzpicture}[mycfsm]
		  \node[state, initial, initial where = above, initial distance = .4cm, initial text={\q}] (zero) {$0$};
		  \node[state] (one) [below of=zero]   {$1$};
		  \path (zero) edge node[above] {$\ain[H][B][][m]$} (one)
		  ;
		\end{tikzpicture}
		\\[4em]
		\begin{tikzpicture}[mycfsm]
		  \node[state, initial, initial where = above, initial distance = .4cm, initial text={$\ptp[c]$}] (zero) {$0$};
		  \node[state] (one) [below of=zero]   {$1$};
		  \path
		  (zero) edge   [bend right] node[below] {$\ain[H][C][][x]$} (one)
		  (zero) edge   [bend left] node[below] {$\ain[H][C][][y]$} (one)
		  ;
		\end{tikzpicture}
	 \end{array}
	 &
		\begin{tikzpicture}[mycfsm, node distance = 1.2cm]
		  \node[state, initial, initial where = below, initial distance = .4cm, initial text={$\HH$}] (zero) {$0$};
		  \node[state] (one) [below right of=zero]   {$1$};
		  \node[state] (two) [below  of=one]   {$2$};
		  \node[state] (three) [below of=two]   {$3$};
		  \node[state] (four) [below  of=three]   {$4$};
		  \node[state] (five) [below  of=four]   {$5$};
		  \node[state] (six) [below left of=five]   {$6$};
		  \node[state] (seven) [below left of=zero]   {$7$};
		  \node[state] (eight) [below  of=seven]   {$8$};
		  \node[state] (nine) [below of=eight]   {$9$};
		  \node[state] (ten) [below of=nine]   {$10$};
		  \node[state] (eleven) [below of=ten]   {$11$};
		  \path
		  (zero) edge [bend left] node[above] {$\tau$} (one)
		  (zero) edge   [bend right] node[above] {$\tau$} (seven)
		  (one) edge node[below] {$\aout[H][A][][m]$} (two)
		  (two) edge    node[below] {$\tau$} (three)
		  (three) edge   node[below] {$\aout[H][C][][y]$} (four)
		  (four) edge    node[below] {$\tau$} (five)
		  (five) edge  [bend left]  node[below] {$\aout[H][B][][m]$} (six)
		  (seven) edge   node[below] {$\aout[H][B][][m]$} (eight)
		  (eight) edge  node[below] {$\tau$} (nine)
		  (nine) edge  node[below] {$\aout[H][C][][x]$} (ten)
		  (ten) edge  node[below] {$\tau$} (eleven)
		  (eleven) edge   [bend right] node[below] {$\aout[H][A][][m]$} (six)
		  ;
		\end{tikzpicture}
		}
		\qqand
		\aCS_2 =
		\csys[c@{\hspace{.5cm}}c]{
		\begin{array}[c]{c}
		\begin{tikzpicture}[mycfsm]
		  \node[state, initial, initial where = above, initial distance = .4cm, initial text={$\KK$}] (zero) {$0$};
		  \node[state] (one) [below right of=zero]   {$1$};
		  \node[state] (two) [below  of=one]   {$2$};
		  \node[state] (three) [below left of=two]   {$3$};
		  \node[state] (four) [below left of=zero]   {$4$};
		  \node[state] (five) [below  of=four]   {$5$};
		  \path
		  (zero) edge [bend left] node[above] {$\ain[D][K][][m]$} (one)
		  (one) edge node[below] {$\ain[D][K][][x]$} (two)
		  (two) edge  [bend left]  node[below] {$\ain[D][K][][m]$} (three)
		  (zero) edge   [bend right] node[above] {$\ain[E][K][][m]$} (four)
		  (four) edge  node[below] {$\ain[E][K][][y]$} (five)
		  (five) edge  [bend right]  node[below] {$\ain[E][K][][m]$} (three)
		  (zero) edge node[above] {$\ain[E][K][][z]$} (three)
		  ;
		\end{tikzpicture}
		  \\[8em]
		  \begin{tikzpicture}[mycfsm]
			 \node[state, initial, initial where = left, initial distance = .4cm, initial text={$\ptp[E]$}] (zero) {$0$};
		  \end{tikzpicture}
		\end{array}
	 &
      \begin{tikzpicture}[mycfsm, node distance = 1cm]
		  \node[state, initial, initial where = right, initial distance = .4cm, initial text={$\ptp[D]$}] (zero) {$0$};
		  \node[state] (one) [below  of=zero]   {$1$};
		  \node[state] (two) [below  of=one]   {$2$};
		  \node[state] (three) [below  of=two]   {$3$};
		  \node[state] (four) [below  of=three]   {$4$};
		  \node[state] (five) [below  of=four]   {$5$};
		  \node[state] (six) [below  of=five]   {$6$};
		  \path
		  (zero) edge node[above] {$\tau$} (one)
		  (one) edge node[below] {$\aout[D][K][][m]$} (two)
		  (two) edge  node[above] {$\tau$} (three)
		  (three) edge    node[below] {$\aout[D][K][][x]$} (four)
		  (four) edge    node[above] {$\tau$} (five)
		  (five) edge  node[below] {$\aout[D][K][][m]$} (six)
		  ;
		\end{tikzpicture}
		}
  \end{align*}
%  are deadlock free.
 % 
  % In $\aCS_1$ and they are
  % $\tau$ transitions. After both of them, the two following sequences
  % of interactions do proceed without any possible branch, and both end
  % into the terminal configuration
  % $(6_{\ptp[H]}, 1_{\ptp[A]}, 1_{\ptp[B]}, 1_{\ptp[C]})$.
  %
  The deadlock freedom of $\aCS_1$ follows from the fact that, from
  its initial configuration $(0_{\p}, 0_{\q}, 0_{\ptp[C]}, 0_{\HH})$,
  $\aCS_1$ can only branch over the two $\tau$-transitions of $\HH$
  reaching either of the following configurations
  $(0_{\p}, 0_{\q}, 0_{\ptp[C]}, 7_{\HH})$ or
  $(0_{\p}, 0_{\q}, 0_{\ptp[C]}, 1_{\HH})$.
  From the former (resp. latter) configuration $\aCS_1$ can only reach
  configurations where $\HH$ synchronises with $\ptp[c]$ and then with
  \p\ (resp. \q).
  In either case $\aCS_1$ reaches the terminal configuration
  $(1_{\ptp[A]}, 1_{\ptp[B]}, 1_{\ptp[C]}, 6_{\ptp[H]})$.
  Let us take a look at $\aCS_2$.
  The only possible transitions of $\aCS_2$ can
  involve $\KK$ and $\ptp[d]$ only since $\ptp[e]$ cannot synchronise being it
  terminated.
  We therefore have that
  \begin{align*}
	 (0_{\ptp[K]}, 0_{\ptp[D]}, 0_{\ptp[E]})
	 &
		\trans \tau
		(0_{\ptp[K]}, 1_{\ptp[D]}, 0_{\ptp[E]})
		\trans{\gint[][d][m][k]}
		(1_{\ptp[K]}, 2_{\ptp[D]}, 0_{\ptp[E]})
		\trans \tau
		(1_{\ptp[K]}, 3_{\ptp[D]}, 0_{\ptp[E]})
		\trans{\gint[][d][x][k]}
		(2_{\ptp[K]}, 4_{\ptp[D]}, 0_{\ptp[E]})
	 \\ &
			\trans \tau
			(2_{\ptp[K]}, 5_{\ptp[D]}, 0_{\ptp[E]})
			\trans{\gint[][d][m][k]}
	 (3_{\ptp[K]}, 6_{\ptp[D]}, 0_{\ptp[E]})
  \end{align*}
  is the only possible execution from the initial configuration
  $(0_{\ptp[K]}, 0_{\ptp[D]}, 0_{\ptp[E]})$ of $\aCS_2$, leading to
  the terminal configuration
  $(3_{\ptp[K]}, 6_{\ptp[D]}, 0_{\ptp[E]})$.
%
  % By the structure of the system, the input actions $\ain[E][K][][m]$
  % and $\ain[E][K][][z]$ can never be performed by $\KK$, since state
  % $0$ of $\ptp[E]$ has no outgoing transition at all. So all the
  % states in the terminal configuration
  % $(3_{\ptp[K]}, 6_{\ptp[D]}, 0_{\ptp[E]})$ are terminal for their
  % respective machines, and hence also $\aCS_2$ is deadlock free.
  %
  \finex
\end{example}
The next example shows that the compositions of the systems $\aCS_1$
and $\aCS_2$ in \cref{ex:dlfree} can deadlock.
\begin{example}[$?!$-determinism is necessary]\label{rem:iodetnecessary}
  The CFSMs $\HH$ and $\KK$ in \cref{ex:dlfree} are compatible as
  seen in~\cref{ex:compatibility}.
  Hence, we can build the composed system
  $\aCS_1 \connect \HH \KK \aCS_2$ through the gateways
  
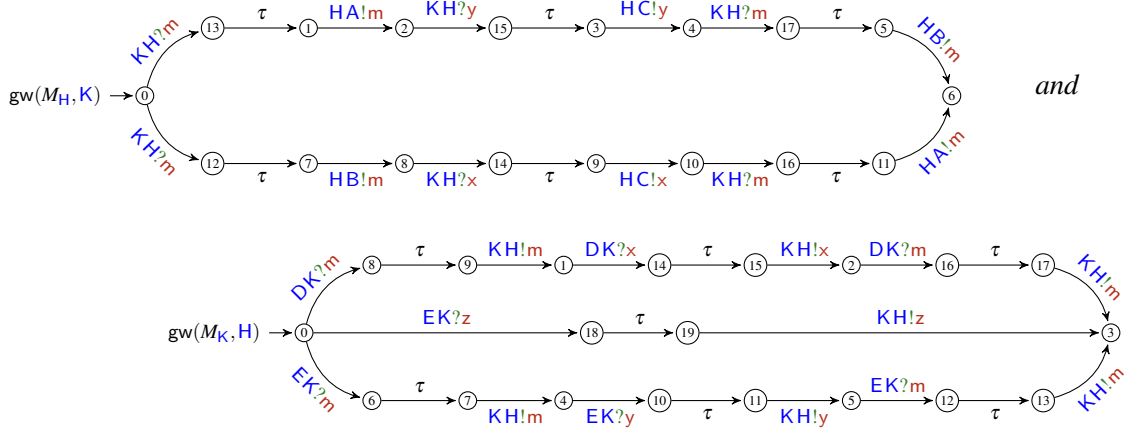
\begin{align*}
	 \begin{tikzpicture}[mycfsm]
		\node[state, initial, initial where = left, initial distance = .4cm, initial text={$\gateway{\aCM_\HH}{\KK}$}] (zero) {$0$};
		\node[state] (twelve) [below right of=zero]   {$12$};
		\node[state] (thirteen) [above right of=zero]   {$13$};
		\node[state] (one) [right of=thirteen]   {$1$};
		\node[state] (two) [right  of=one]   {$2$};
		\node[state] (fifteen) [right  of=two]   {$15$};
		\node[state] (three) [right of=fifteen]   {$3$};
		\node[state] (four) [right  of=three]   {$4$};
		\node[state] (seventeen) [right  of=four]   {$17$};
		\node[state] (five) [right  of=seventeen]   {$5$};
		\node[state] (six) [below right of=five]   {$6$};
		\node[state] (seven) [right of=twelve]   {$7$};
		\node[state] (eight) [right  of=seven]   {$8$};
		\node[state] (fourteen) [right  of=eight]   {$14$};
		\node[state] (nine) [right of=fourteen]   {$9$};
		\node[state] (ten) [right of=nine]   {$10$};
		\node[state] (sixteen) [right of=ten]   {$16$};
		\node[state] (eleven) [right of=sixteen]   {$11$};
		\path
		(zero) edge [bend right] node[below] {$\ain[K][H][][m]$} (twelve)
		(zero) edge   [bend left] node[above] {$\ain[K][H][][m]$} (thirteen)
		(twelve) edge  node[below] {$\tau$} (seven)
		(thirteen) edge  node[above] {$\tau$} (one)
		(one) edge node[above] {$\aout[H][A][][m]$} (two)
		(two) edge    node[above] {$\ain[K][H][][y]$} (fifteen)
		(fifteen) edge  node[above] {$\tau$} (three)
		(three) edge   node[above] {$\aout[H][C][][y]$} (four)
		(four) edge    node[above] {$\ain[K][H][][m]$} (seventeen)
		(seventeen) edge    node[above] {$\tau$} (five)
		(five) edge  [bend left]  node[above] {$\aout[H][B][][m]$} (six)
		(seven) edge   node[below] {$\aout[H][B][][m]$} (eight)
		(eight) edge  node[below] {$\ain[K][H][][x]$} (fourteen)
		(fourteen) edge  node[below] {$\tau$} (nine)
		(nine) edge  node[below] {$\aout[H][C][][x]$} (ten)
		(ten) edge  node[below] {$\ain[K][H][][m]$} (sixteen)
		(sixteen) edge  node[below] {$\tau$} (eleven)
		(eleven) edge   [bend right] node[below] {$\aout[H][A][][m]$} (six)
		;
	 \end{tikzpicture}
	 \qqand
	 \\[1em]
	 \begin{tikzpicture}[mycfsm]
		\node[state, initial, initial where = left, initial distance = .4cm, initial text={$\gateway{\aCM_\KK}{\HH}$}] (zero) {$0$};
		\node[state] (eight) [above right of=zero]   {$8$};
		\node[state] (nine) [right  of=eight]   {$9$};
		\node[state] (one) [right  of=nine]   {$1$};
		\node[state] (fourteen) [right  of=one]   {$14$};
		\node[state] (fifteen) [right  of=fourteen]   {$15$};
		\node[state] (two) [right  of=fifteen]   {$2$};
		\node[state] (sixteen) [right  of=two]   {$16$};
		\node[state] (seventeen) [right  of=sixteen]   {$17$};
		\node[state] (six) [below right of=zero]   {$6$};
		\node[state] (seven) [right  of=six]   {$7$};
		\node[state] (four) [right  of=seven]   {$4$};
		\node[state] (ten) [right  of=four]   {$10$};
		\node[state] (eleven) [right  of=ten]   {$11$};
		\node[state] (five) [right  of=eleven]   {$5$};
		\node[state] (twelve) [right  of=five]   {$12$};
		\node[state] (thirteen) [right  of=twelve]   {$13$};
		\node[state] (three) [above right of=thirteen]   {$3$};
		\node[state] (18) [right of=zero,xshift=3cm]   {$18$};
		\node[state] (19) [right of=18]   {$19$};
		\path
		(zero) edge [bend left] node[above] {$\ain[D][K][][m]$} (eight)
		(nine) edge node[above] {$\aout[K][H][][m]$} (one)
		(seventeen) edge  [bend left]  node[above] {$\aout[K][H][][m]$} (three)
		(zero) edge   [bend right] node[below] {$\ain[E][K][][m]$} (six)
		(seven) edge  node[below] {$\aout[K][H][][m]$} (four)
		(five) edge   node[above] {$\ain[E][K][][m]$} (twelve)
		(four) edge  node[below] {$\ain[E][K][][y]$} (ten)
		(six) edge  node[above] {$\tau$} (seven)
		(eight) edge  node[above] {$\tau$} (nine)
		(ten) edge  node[below] {$\tau$} (eleven)
		(one) edge  node[above] {$\ain[D][K][][x]$} (fourteen)
		(twelve) edge  node[below] {$\tau$} (thirteen)
		(thirteen) edge  [bend right]  node[below] {$\aout[K][H][][m]$} (three)
		(eleven) edge  node[below] {$\aout[K][H][][y]$} (five)
		(fourteen) edge  node[above] {$\tau$} (fifteen)
		(fifteen) edge  node[above] {$\aout[K][H][][x]$} (two)
		(two) edge node[above] {$\ain[D][K][][m]$} (sixteen)
		(sixteen) edge  node[above] {$\tau$} (seventeen)
		(zero) edge node[above]{$\ain[e][k][][z]$} (18)
		(18) edge node[above]{$\tau$} (19)
		(19) edge node[above]{$\aout[k][h][][z]$} (three)
		;
	 \end{tikzpicture}
  \end{align*}
  Now, from the initial configuration
  $\aConf_0 = (0_{\p}, 0_{\q}, 0_{\ptp[c]},
  0_{\gateway{\aCM_\HH}{\KK}}, 0_{\gateway{\aCM_\KK}{\HH}},
  0_{\ptp[D]}, 0_{\ptp[e]})$ of $\aCS_1 \connect \HH \KK \aCS_2$ we have the following run
  \begin{align}
	 \aConf_0
	 \trans \tau \nonumber
	 &
		(0_{\p}, 0_{\q}, 0_{\ptp[c]}, 0_{\gateway{\aCM_\HH}{\KK}}, 0_{\gateway{\aCM_\KK}{\HH}}, 1_{\ptp[D]}, 0_{\ptp[e]})
	 \\ \nonumber
	 \trans{\gint[][D][m][K]} \trans \tau &
		(0_{\p}, 0_{\q}, 0_{\ptp[c]}, 0_{\gateway{\aCM_\HH}{\KK}}, 9_{\gateway{\aCM_\KK}{\HH}}, 2_{\ptp[D]}, 0_{\ptp[e]})
	 \\\label{eq:gatewaysmove}
	 \trans{\gint[][k][m][h]} \trans \tau &
		(0_{\p}, 0_{\q}, 0_{\ptp[c]}, 13_{\gateway{\aCM_\HH}{\KK}}, 1_{\gateway{\aCM_\KK}{\HH}}, 3_{\ptp[D]}, 0_{\ptp[e]})
	 \\\nonumber
	 \trans{\gint[][d][x][k]} &
		(0_{\p}, 0_{\q}, 0_{\ptp[c]}, 13_{\gateway{\aCM_\HH}{\KK}}, 14_{\gateway{\aCM_\KK}{\HH}}, 4_{\ptp[D]}, 0_{\ptp[e]})
	 \\\nonumber
	 \trans \tau \trans \tau \trans \tau &
		(0_{\p}, 0_{\q}, 0_{\ptp[c]}, 1_{\gateway{\aCM_\HH}{\KK}}, 15_{\gateway{\aCM_\KK}{\HH}}, 5_{\ptp[D]}, 0_{\ptp[e]})
	 \\\label{eq:dl}
	 \trans{\gint[][h][m][a]} &
		(1_{\p}, 0_{\q}, 0_{\ptp[c]}, 2_{\gateway{\aCM_\HH}{\KK}}, 15_{\gateway{\aCM_\KK}{\HH}}, 5_{\ptp[D]}, 0_{\ptp[e]})
  \end{align}
  where the $\tau$-transition of $\ptp[d]$ enables the synchronisation
  of $\gateway{\aCM_\KK}{\HH}$ and $\ptp[d]$ with label
  $\gint[][D][m][K]$ that leads the gateway in state $9$ after its
  $\tau$-transition from state $8$.
  Now, the two gateways can communicate and exchange message
  $\msg[m]$.
  Due to $?!$-nondeterminism of $\aCS_1$, from state $0$
  $\gateway{\aCM_\HH}{\KK}$ can move either to state $12$ or to state
  $13$.
  Fatally, transition \eqref{eq:gatewaysmove} leads to a
  deadlock: after $\gateway{\aCM_\KK}{\HH}$ and $\ptp[d]$ synchronise
  to exchange message $\msg[x]$ the system goes into a configuration
  from where $\gateway{\aCM_\HH}{\KK}$ forwards $\msg[m]$ to \p\ and
  reaches the last configuration \eqref{eq:dl}.
  This is a deadlock for $\aCS_1 \connect \HH \KK \aCS_2$, since no
  CFSM can do a $\tau$-transitions, the only enabled output action is
  from $\gateway{\aCM_\KK}{\HH}$ which tries to send message $\msg[x]$
  to $\gateway{\aCM_\HH}{\KK}$; however, $\gateway{\aCM_\HH}{\KK}$ can
  only receive message $\msg[y]$ from $\KK$ and hence these actions
  cannot synchronise.
  \finex
\end{example}

%%% Local Variables:
%%% mode: latex
%%% TeX-master: "main"
%%% End:

\section{Preserving Properties by Composition}\label{sec:pbc}
Composition via gateways does not ensure the preservation of
communication properties.
We provide below sufficient conditions for this to happen.
%
% \begin{definition}[Sequential CFSM]
% 	 A CFSM is {\em sequential} if  each  of its states has at most one outgoing transition. 
% \end{definition}
% It is immediate to check that $\gateway{M}{\HH}$ is sequential if $M$ is so.
% Moreover, trivially, a sequential $M$ is also $?!$-deterministic
% and with no mixed state (and hence mixed-deterministic).
%
Recall that $(\HH, \KK)$-composability requires absence of asymmetric
mixed states and ?!-determinism.
\begin{restatable}[Deadlock freedom preservation]{theorem}{seqgatthm}\label{th:dfSeqGat}
  Let $\aCS_1$ and $\aCS_2$ be two $(\HH,\KK)$-composable and
  deadlock-free systems.
  Then the composed system $\aCS_1 \connect \HH \KK \aCS_2$ is
  deadlock-free.
\end{restatable}
\iftr
  %\seqgatthm*

  Given an $\HH$-local CFSM $\aCM$ and a participant
  $\KK \in \pset \setminus \{\HH\}$, call \emph{connecting} a fresh
  asymmetric sending state of $\gateway \aCM \KK$ whose next outgoing
  transition does not have $\KK$ as receiver.
  \begin{proof} 
	 We show that if the composed system $\aCS_1 \connect \HH \KK
	 \aCS_2$ reaches a deadlock configuration $\aConf$ then at least
	 one of $\cproj \aConf \HH$ and $\cproj \aConf \KK$ is a deadlock.
	 Without loss of generality, we assume that the deadlock is the
	 left-projection; the case where the deadlock is the
	 right-projection is similar.
	 
	 First, we show that if a participant $\p$ from $\aCS_1$ has an
	 enabled transition in $\aConf$ then some participant in $\aCS_1$
	 has a transition enabled in $\cproj \aConf \HH$.
	 Note that $\cproj \aConf \HH$ is reachable in $\aCS_1$ by
	 \cref{prop:comp}.
	 
	 If $\p \neq \HH$ then any transition of $\p$ enabled in $\aConf$
	 is also enabled in $\cproj \aConf \HH$ since
	 $\cproj \aConf \HH(\p) = \aConf(\p)$ by \cref{def:cproj}.
	 If $\p=\HH$, then either of the following cases occurs
	 \begin{itemize}
	 \item $\HH$ has enabled an input

		Assume that the input is from $\KK$.
		Then by construction $\cproj \aConf \HH(\HH)$ has a
		$\tau$-transition enabled in $\aCS_1(\HH)$.

		If the input of $\HH$ is from a participant \p\ of $\aCS_1$ then
		by construction $\cproj \aConf \HH(\HH)$ has an input transition
		enabled in $\aCS_1(\HH)$.
	 \item $\HH$ has enabled an output.
		
		Assume that the receiver of such output is $\KK$.
		Then
		$\nof[\HH,\KK][\aCS_1(\HH)]{\aConf(\HH)} \comp
		\nof[\KK,\HH][\aCS_2(\KK)]{\aConf(\KK)}$ by \cref{prop:comp}.
		By definition of $\nof[]{}$ and of gateway,
		$\nof[\HH,\KK][\aCS_1(\HH)]{\aConf(\HH)}$ has an input
		transition enabled from a participant in $\aCS_1$, hence
		$\nof[\KK,\HH][\aCS_2(\KK)]{\aConf(\KK)}$ has a corresponding
		output transition enable towards a participant in $\aCS_2$ by
		compatibility.
		By construction (\cref{def:cproj}),
		$\cproj \aConf \KK(\KK) =
		\nof[\KK,\HH][\aCS_2(\KK)]{\aConf(\KK)}$, hence there is a
		participant, in particular $\KK$, willing to take a transition.

		If the receiver of the output from $\HH$ is a participant of
		$\aCS_1$ then, by definition of gateway and configuration
		projection, we get that in $\HH$ is willing to perform an output
		from $\cproj \aConf \HH(\HH)$ in $\aCS_1$.
	 \item $\HH$ can perform a $\tau$-transition.

		Then, there is a sequence of transitions of the form
		$\aConf(\HH) \trans \tau \trans{\aout[h][x]}$ in $\aCS(\HH)$
		with $\p[x] = \KK$ or $\p[x] \neq \HH$ participant of $\aCS_1$.
		If the former case we can reason as in the previous case when
		$\HH$ outputs to $\KK$.
		Otherwise, $\cproj \aConf \HH(\HH)$ has an enabled
		$\tau$-transition in $\aCS_1(\HH)$ by \cref{def:cproj}.
	 \end{itemize}
	 If $\aConf$ is a deadlock, by definition of deadlock freedom
	 (cf. \cref{def:props}), $\aConf \not\trans{}$ but there are
	 participants in $\aCS$ with enabled transitions in $\aConf$.
	 Under the assumption that $\cproj \aConf \KK$ is deadlock-free,
	 such participants must belong to $\aCS_1$.
	 By the cases shown above, $\cproj \aConf \HH$ enables some
	 participants in $\aCS_1$; therefore, there is
	 $\cproj \aConf \HH \trans \al$ because $\aCS_1$ is deadlock free
	 by hypothesis.
	 It must be that $\HH$ is involved in all transitions from
	 $\cproj \aConf \HH$ of $\aCS_1$ otherwise $\aConf \trans \al$
	 since for all $\p[x] \in \ptpof[\al]$
	 $\aConf(\p[x]) = \cproj \aConf \HH (\p[x])$ (by \cref{def:cproj})
	 contrary to our assumption that $\aConf$ is a deadlock.
	 We proceed by case analysis on $\al$.

	 \noindent
	 \fbox{$\al = \gint[][h][m][x]$} If $\p[x] \neq \KK$ then
	 $\cproj \aConf \HH (\p[x]) = \aConf(\p[x]) \trans{\ain[h][a]}$;
	 also, $\HH$ would be in a connecting state and, by
	 \cref{def:cproj},
	 $\cproj \aConf \HH (\HH) = \aConf(\HH) \trans{\aout[h][a]}$.
	 Hence $\aConf \trans \al$ contrary to the hypothesis
	 that $\aConf$ is a deadlock configuration.

	 If $\p[x] = \KK$ then we have again a contradiction since
	 $\aConf \trans{\al}$ by \cref{prop:comp}.
	 
	 \noindent
	 \fbox{$\al = \gint[][x][m][h]$} If $\p[x] \neq \KK$ then
	 $\cproj \aConf \HH (\p[x]) = \aConf(\p[x]) \trans{\aout[a][h]}$
	 and $\cproj \aConf \HH (\p[x]) \trans{\ain[a][h]} = \aConf(\HH)$
	 Hence $\aConf \trans \al$ contrary to the hypothesis that $\aConf$
	 is a deadlock configuration.

	 If $\p[x] = \KK$ then we have again a contradiction since
	 $\aConf \trans{\al}$ by \cref{prop:comp}.

	 \noindent
	 \fbox{$\al = \tau$} If $\aConf(\HH) = p''$ is fresh then
	 $\gateway{\aCS_1(\HH)} \KK$ must have a sequence of transitions
	 \begin{align} \label{eq:toK}
		p \trans{\ain[a][h]} p' \trans \tau p'' \trans{\aout[h][k]} r
	 \end{align}
	 with $\p$ participant of $\aCS_1$ and $p'$ fresh.
	 (Note that it cannot be $\aConf(\HH) = p'$ otherwise
	 $\aConf(\HH) \trans \tau$ contradicting $\aConf \not\trans{}$.)
	 By \cref{def:cproj}, $\cproj \aConf \HH(\HH) = r$.
	 Hence, by \cref{prop:comp} we have $s \trans{\gint[][h][m][k]}$
	 contradicting $\aConf \not\trans{}$.

	 If $\aConf(\HH)$ is not fresh then $\HH$ has a $\tau$ transition
	 enabled at $\aConf$ (because $\aConf(\HH) = \cproj \aConf
	 \HH(\HH)$ by \cref{def:cproj}) again contradicting
	 $\aConf \not\trans{}$.
  \end{proof}
  \else
\begin{proofsketch}
  The proof relies on the fact that the reachable configurations of
  $\aCS_1 \connect \HH \KK \aCS_2$ can be projected on reachable
  configurations of $\aCS_1$ and $\aCS_2$.
  This implies that a deadlock in $\aCS_1 \connect \HH \KK \aCS_2$
  corresponds to a deadlock in $\aCS_1$ or in $\aCS_2$.
  See the appendix for the detailed proof.
\end{proofsketch}
  \fi

\begin{example}
  We can infer deadlock-freedom of the system
  $\aCS = \aCS_1 \connect{\HH}{\KK} \aCS_2$ of \cref{ex:gc} by the
  result above, since $\aCS_1$ and $\aCS_1$ are $(\HH,\KK)$-composable
  and deadlock-free.
\end{example}

Somehow surprisingly, in the symmetric case preservation of deadlock
freedom requires stricter conditions on gateways than in the
asymmetric case.
In fact, in the asymmetric case, deadlock freedom preservation
requires only absence of asymmetric mixed states and $?!$-determinism
while the symmetric case requires the (stronger) condition of
\emph{sequentiality}.
\begin{definition}[Sequential CFSM]
  A CFSM is \emph{sequential} if each of its states has at most one
  outgoing transition.
  A participant \p\ of a system $\aCS$ is sequential if $\aCS(\p)$ is
  so.
\end{definition}
As we will see (cf. \cref{thmt@@lfpresthm}), sequentiality is necessary
to preserve lock-freedom also in the asymmetric case.
We note that sequentiality implies absence of asymmetric mixed states
and $?!$-determinism, while the converse does not hold.

As mentioned before, the property of lock freedom is not preserved in
general by composition, as shown by the following example.
\begin{example}[Composability does not preserve lock-freedom]\label{rem:lfnotpres} 
  Take the communicating systems
  \begin{align*}
	 \aCS_1 = \dboxed{
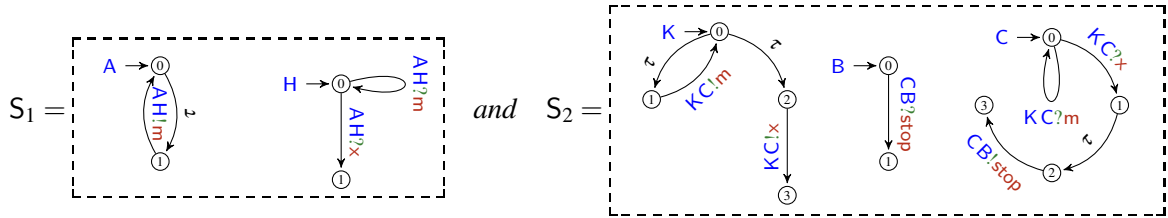
	 \begin{array}{c}
		\begin{tikzpicture}[mycfsm]
		  \node[state, initial, initial where = left, initial distance = .4cm, initial text={$\p$}] (zero) {$0$};
		  \node[state] (one) [below of=zero]   {$1$};
		  \path
		  (zero) edge[bend left] node[below] {$\tau$} (one)
		  (one) edge[bend left]  node[above] {$\aout[A][H][][m]$} (zero)
		  ;
		\end{tikzpicture}
	 \end{array}
	 \qquad
	 \begin{tikzpicture}[mycfsm]
		\node[state, initial, initial where = left, initial distance = .4cm, initial text={$\HH$}] (zero) {$0$};
      \node[state] (one) [below of=zero]   {$1$};
      \path
      (zero) edge [loop right,looseness=40] node[above] {$\ain[A][H][][m]$} (zero)
      (zero) edge node[above] {$\ain[A][H][][x]$} (one)
      ;
	 \end{tikzpicture}
	 }
	 \qand
	 \aCS_2 = \dboxed{
	 \begin{array}{ccc}
		\begin{tikzpicture}[mycfsm]
		  \node[state, initial, initial where = left, initial distance = .4cm, initial text={$\KK$}] (zero) {$0$};
		  \node[state] (one) [below left of=zero]   {$1$};
		  \node[state] (two) [below right of=zero]   {$2$};
		  \node[state] (three) [below  of=two]   {$3$};
		  \path
		  (zero)  edge[bend right] node[above, near end] {$\tau$} (one)
		  (zero) edge[bend left] node[above] {$\tau$} (two)
		  (two) edge node[above] {$\aout[K][C][][x]$} (three)
		  (one) edge[bend right] node[below] {$\aout[K][C][][m]$} (zero)
		  ;
		\end{tikzpicture}
		&
		\begin{tikzpicture}[mycfsm]
		  \node[state, initial, initial where = left, initial distance = .4cm, initial text={$\q$}] (zero) {$0$};
		  \node[state] (one) [below of=zero]   {$1$};
		  \path (zero) edge node[above] {$\ain[C][b][][stop]$} (one)
		  ;
		\end{tikzpicture}
		&
		\begin{tikzpicture}[mycfsm]
		  \node[state, initial, initial where = left, initial distance = .4cm, initial text={$\p[c]$}] (zero) {$0$};
		  \node[state] (one) [below right of=zero]   {$1$};
		  \node[state] (two) [below left of=one]   {$2$};
		  \node[state] (three) [above left of=two]   {$3$};
		  \path
		  (zero) edge[loop below,looseness=40] node[below] {$\ain[K][C][][m]$} (zero)
		  (zero) edge [bend left] node[above] {$\ain[K][C][][x]$} (one)
		  (one) edge [bend left]  node[above] {$\tau$} (two)
		  (two) edge [bend left]  node[below] {$\aout[C][b][][stop]$} (three)
		  ;
		\end{tikzpicture}
	 \end{array}
	 }
  \end{align*}
  Note that both $\aCS_1$ and $\aCS_2$ are lock-free and that $\HH$
  and $\KK$ are compatible.
  The gateways are
  
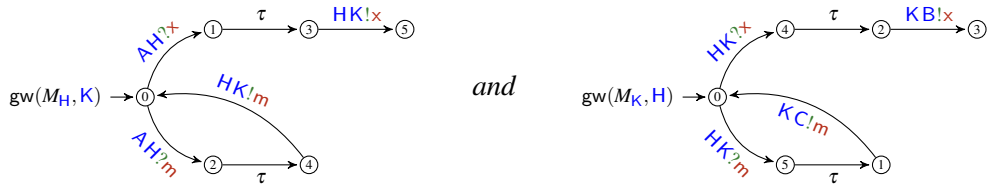
\begin{align*}
	 \begin{tikzpicture}[mycfsm]
		\node[state, initial, initial where = left, initial distance = .4cm, initial text={$\gateway {\aCM_\HH} \KK$}] (zero) {$0$};
      \node[state] (one) [above right of=zero]   {$1$};
      \node[state] (two) [below right of=zero]   {$2$};
      \node[state] (three) [right  of=one]   {$3$};
      \node[state] (four) [right  of=two]   {$4$};
      \node[state] (five) [right  of=three]   {$5$};
      \path
      (zero) edge[bend right] node[below] {$\ain[A][H][][m]$} (two)
      (zero) edge[bend left] node[above] {$\ain[a][H][][x]$} (one)
      (one) edge node[above] {$\tau$} (three)
      (two) edge node[below] {$\tau$} (four)
      (three) edge node[above] {$\aout[H][K][][x]$} (five)
      (four) edge[bend right] node[above] {$\aout[H][K][][m]$} (zero)
      ;
	 \end{tikzpicture}
	 \qqand
	 \begin{tikzpicture}[mycfsm]
		\node[state, initial, initial where = left, initial distance = .4cm, initial text={$\gateway {\aCM_\KK} \HH$}] (zero) {$0$};
      \node[state] (five) [below right of=zero]   {$5$};
      \node[state] (four) [above right of=zero]   {$4$};
      \node[state] (one) [right  of=five]   {$1$};
      \node[state] (two) [right of=four]   {$2$};
      \node[state] (three) [right  of=two]   {$3$};
      \path
      (zero)  edge[bend right] node[below] {$\ain[H][K][][m]$} (five)
      (zero) edge[bend left] node[above] {$\ain[H][K][][x]$} (four)
      (five)  edge node[below] {$\tau$} (one)
      (four) edge node[above] {$\tau$} (two)
      (two) edge node[above] {$\aout[K][b][][x]$} (three)
      (one) edge[bend right] node[below] {$\aout[K][C][][m]$} (zero)
      ;
	 \end{tikzpicture}
  \end{align*}
  Hence, the composed system $\aCS_1 \connect \HH \KK \aCS_2$ is non
  lock-free because e.g. the configuration
  \[
	 \aConf = (0_{\p},0_{\gateway {\aCM_{\HH}}{ \KK}}, 0_{\gateway {\aCM_{\KK}}{ \HH}}, 0_{\q}, 0_{\ptp[c]})
  \]
  is a lock for $\q$, since the only outgoing transition from $0_{\q}$
  could be fired only in case the transition $\aout[c][b][][stop]$ is
  enabled.
  However, this is impossible since $\gateway {\aCM_{\HH}}{ \KK}$
  forwards only message $\msg[m]$; hence, the run
  $\aConf \trans \tau \trans{\gint[][a][m][h]} (0_{\p},2_{\gateway
	 {\aCM_{\HH}}{\KK}}, 0_{\gateway{\aCM_{\KK}}{\HH}}, 0_{\q},
  0_{\ptp[c]}) \trans \tau \trans{\gint[][a][m][h]} \aConf \cdots$
  (which does not involve $\q$) is perpetually executed.
  \finex
\end{example}

% \eMcomm[fix]{The main problem for the failure of lock-freedom preservation by
% composition, exemplified above is that there is a configuration
% $\hat \aConf=(3_\gateway {\aCM_{\KK}}{ \HH},1_{\ptp[C]},3_{\ptp[D]})
% \in\RS[\ssem{\aCS_2}]$ such that for no configuration
% $\aConf\in\RS[\ssem{\aCS_1 \connect \HH \KK \aCS_2}]$ we have that
% $\cproj \aConf \KK = \hat \aConf$.
% }
% 	 
We show that the problem of \cref{rem:lfnotpres} cannot happen in case
we restrict to sequential gateways, as done for deadlock freedom in
the symmetric case (cf.~\cite{BLT20b}).
As usual, $\prestrict f X$ denotes the restriction of a function $f$
on a subset $X$ of its domain.
% 
% Given a configuration $\aConf$ with domain $\ptpset$ and
% $\ptpset' \subseteq \ptpset$ we write $\prestrict{\ptpset'}{s}$ is
% defined as the restriction of $\aConf$ to the subdomain $\ptpset'$.

\begin{restatable}[Lock-freedom preservation]{theorem}{lfpresthm}
  Let $\aCS_1$ and $\aCS_2$ be two $(\HH,\KK)$-composable and
  lock-free systems with $\HH$ and $\KK$ sequential.
  Then the composed system $\aCS_1 \connect \HH \KK \aCS_2$ is
  lock-free.
\end{restatable}
\iftr
  % The following lemma allows us to \quo{postpone} a $\tau$
  % transition leading to an output until the receiver is receiver
  % is ready to synchronise.
  % \begin{lemma}\label{lem:hold}
  % 	 If
  % 	 $\aConf_0 \trans{\al_1} \cdots \trans{\al_n} \aConf_n
  % 	 \trans{\gint[]} \aConf$ be a run of from the initial configuration
  % 	 $\aConf_0$ of a communicating system $\aCS$ then there are a run
  % 	 $\aConf_0 \trans{\al_1'} \cdots \trans{\al_n'} \aConf_n
  % 	 \trans{\gint[]} \aConf$ and an index $1 \leq i \leq n$ such that
  % 	 $\al_j' = \al_j$ for all $1 \leq j < i$, \p\ is not involved in
  % 	 any transition between $i+1$ and $n-1$, and
  % 	 $\aConf_{n-1} \trans{\al'_n} \aConf_n$ is a $\tau$-transition involving
  % 	 \p.
  % \end{lemma}
  % \begin{proofsketch}
  % 	 Let 
  % 	 $\aConf_i \trans{\al_i} \aConf_{i+1}$ be the last $\tau$-transition
  % 	 involving \p.
  % 	 %
  % 	 Necessarily, $\aCS(\p)$ has a transition
  % 	 $\aConf_{i+1}(\p) = \aConf_n(\p) \trans\aout \aConf(\p)$ and \p\
  % 	 is not involved in any transition between $i+1$ and $n$.
  % 	 %
  % 	 Let $q = \aConf_i(\p)$.
  % 	 %
  % 	 Then
  % 	 $\aConf_0 \trans{\al_1} \cdots \trans{\al_{i-1}} \aConf_i \trans{\al_{i+1}} \upd{\aConf_{i+1}} \p q
  % 	 \cdots \upd{\aConf_{n-1}} \p q \trans \tau \aConf_n
  % 	 \trans{\gint[]} \aConf$ is a run of $\aCS$.
  % \end{proofsketch}
  \begin{lemma}\label{lem:reconstruction}
	 Let $\aCS = \aCS_1 \connect \HH \KK \aCS_2$ where $\aCS_1$ and
	 $\aCS_2$ are two $(\HH,\KK)$-composable systems with $\HH$ and
	 $\KK$ sequential.
	 Given $\aConf \in \RS[\ssem \aCS]$, if
	 \begin{itemize}
	 \item[(1)] either $\cproj{\aConf} \HH \trans \al \aConf'$
		in $\aCS_1$ and
		$\al = \tau \implies \cproj \aConf \HH(\HH) = \aConf'(\HH)$
	 \item[(2)] or $\cproj \aConf \HH \trans \tau \hat \aConf$
		involving $\HH$ in $\aCS_1(\HH)$ and $\hat \aConf$ reaching a
		configuration $\hat \aConf'$ such that
		$\hat \aConf' \trans \al \aConf'$ in $\aCS_1$ with
		$\al = \gint[][h][m][a]$
	 \end{itemize}
	 then there is a run
	 $\aConf \trans{\al_1} \cdots \trans{\al_n} \hat\aConf$ in $\aCS$
	 such that $\al_n = \al$ and $\cproj{\hat \aConf} \HH = \aConf'$.

	 The same holds for the right-projection of $\aCS$.
  \end{lemma}
  \begin{proof}
	 We give the proof for each case.

	 \noindent
	 \fbox{Case (1)} By case analysis on $\al$ noticing that the case
	 $\al = \gint[][h][m][a]$ is not possible since $\cproj \aConf \HH$
	 cannot enable output transitions from $\HH$ by construction
	 (cf. \cref{def:cproj}).
	 \begin{description}
	 \item[$\al = \tau$.] Then the $\tau$-transition is executed by
		$\p \neq \HH$ in $\aCS_1$; hence there is a transition
		$\cproj \aConf \HH(\p) \trans \tau q$ in $\aCS_1(\p)$.
		Observing that $\aConf(\p) = \cproj \aConf \HH(\p)$ by
		\cref{def:cproj} we have the thesis since
		$\aConf \trans \tau \upd \aConf \p q$.
	 \item[{$\al = \gint[][a][m][h]$}.] We have
		\[
		  \cproj \aConf \HH(\HH) \trans{\ain[a][h]} p
		  \text{ in } \aCS_1(\HH)
		  \qqand
		  \cproj \aConf \HH(\p) \trans{\aout[a][h]} q
		  \text{ in } \aCS_1(\p)
		\]
		moreover, $\aConf(\HH) = \cproj \aConf \HH(\HH)$ and
		$\aConf(\p) = \cproj \aConf \HH(\p)$.
		Hence, 
		\[
		  \aConf \trans{\gint[][a][m][h]} \aConf''
		  \qqand[with]
		  \aConf''(\p[x]) =
		  \begin{cases}
			 p, & \text{if } \p[x] = \HH
			 \\
			 q, & \text{if } \p[x] = \p
			 \\
			 \aConf(\p[x]), & \text{otherwise}		
		  \end{cases}
		\]
		by \cref{def:syncSem}.
		Therefore $\cproj{\aConf''}\HH = \aConf'$.
	 \end{description}

	 \noindent
	 \fbox{Case (2)} Note that $\aCS_1(\HH)$ necessarily contains the
	 transitions
	 $\cproj \aConf \HH(\HH) \trans \tau \hat \aConf'(\HH)
	 \trans{\aout[h][a]} r$.
	 Then, by construction, in $\gateway{\aCS_1(\HH)}{\KK}$ we have
	 \begin{align}\label{eq:taus}
		\cproj \aConf \HH(\HH) \trans{\ain[k][h]} p \trans \tau \hat \aConf'(\HH) \trans{\aout[h][a]} r
	 \end{align}
	 note that $p$ and $\hat \aConf'(\HH)$ are fresh, that these are
	 the only transitions from $\cproj \aConf \HH(\HH)$ to $r$ by
	 sequentiality, and that $\aConf(\HH) \in \{\cproj \aConf \HH(\HH),p, \hat \aConf'(\HH)$ by construction
	 (cf. \cref{def:cproj}).
	 If $\aConf(\HH) = \cproj \aConf \HH(\HH)$ then, by
	 \cref{prop:comp}, we have
	 $\aConf \trans{\gint[][k][m][h]} \trans \tau \aConf''$ with
	 $\aConf''(\HH) = \hat \aConf'(\HH)$.
	 Hence
	 $\cproj{\aConf''} \HH = \upd{\cproj \aConf \HH} \HH r$ and therefore
	 $\aConf'' \trans{\gint[][h][m][a]} \aConf''[\HH \mapsto r, \p \mapsto
	 \aConf'(\p)] = \aConf'''$ which yields the thesis noticing that
	 $\cproj{\aConf'''} \HH = \aConf'$.
	 If $\aConf(\HH) = p$ then $\aConf \trans \tau \aConf''$ with
	 $\aConf''(\HH) = \hat \aConf'(\HH)$; hence the thesis follows as in
	 the previous case.
	 Finally, if $\aConf(\HH) = \hat\aConf'(\HH)$ then
	 $\aConf \trans{\gint[][h][m][a]} \aConf''$ with
	 $\aConf'' = \aConf[\HH \mapsto r, \p \mapsto \aConf'(\p)]$ and therefore
	 $\cproj{\aConf''}\HH = \aConf'$.
  \end{proof}
  \lfpresthm*
  \begin{proof}
	 By contradiction, let us assume
	 $\aCS = \aCS_1 \connect \HH \KK \aCS_2$ not to be lock-free.
	 Then there is a configuration $\aConf \in \RS[\ssem \aCS]$ and a
	 participant $\p[x]$ not involved in any run from $\aConf$.
	 Without any loss of generality, we can assume $\p[x] \in \aCS_1$.
	 We have $\cproj \aConf \HH \in \RS[\ssem{\aCS_1}]$ by
	 \cref{prop:comp} and, by lock-freedom of $\aCS_1$,
	 $\cproj \aConf \HH$ cannot be a lock of $\aCS_1$ for $\p[x]$.
	 Hence, there exists a run
	 $\cproj \aConf \HH \trans{\al_0} \aConf_0 \cdots \aConf_{n-1} \trans{\al_n}
	 \aConf_n$ of $\aCS_1$ with $\p[x]$ involved in $\al_n$.
	 We show that this induces a run from $\aConf$ in $\aCS$ involving
	 $\p[x]$ by induction on $n$.

	 %
	 % Otherwise
	 % $p \trans{\ain[k][h]} p' \trans \tau \aConf(\HH)
	 % \trans{\aout[h][a]} r$ is in $\gateway{\aCS_1(\HH)} \KK$.
	 % %
	 % If $\p[x] = \HH$ then
	 % $\aConf \trans{\gint[][k][m][h]} \trans \tau
	 % \trans{\gint[][h][m][a]}$ by \cref{prop:comp} and
	 % $\aConf(\p) = \cproj \aConf \HH(\p)$, contrary to our assumption
	 % that $\aConf$ is a lock for $\p[x]= \HH$; so, let us assume
	 % $\p[x] \neq \HH$.
	 % %
	 % Then $\al \neq \tau$ otherwise $\aConf$ enables a $\tau$-transition
	 % of $\p[x]$ by \cref{def:cproj}, violating our assumption that no
	 % runs from $\aConf$ can involve $\p[x]$; for the same reason
	 % $\HH \in \ptpof[\al]$ otherwise the partners of $\p[x]$ in $\al$
	 % would have a complementary transition of $\aConf(\p[x])$ enabled
	 % in $\aConf$, and therefore $\aConf \trans \al$.
	 % %
	 % So, necessarily $\al = \gint[][x][m][h]$; in fact, the case
	 % $\al = \gint[][h][m][x]$ is impossible because
	 % $\cproj \aConf \HH(\HH)$ cannot have output transitions by
	 % construction (cf. \cref{def:cproj}).
	 % %
	 % Then $\aConf \trans{\gint[][x][m][h]}$ because
	 % $\cproj \aConf \HH(\p[x]) = \aConf(\p[x])$ and
	 % $\cproj \aConf \HH(\HH) = \aConf(\HH)$ by \cref{def:cproj}.
	 
	 % All the cases yield a contradiction, hence $\aConf$ cannot be a
	 % lock for $\p[x]$.
	 \begin{itemize}
	 \item If $n=0$, by \cref{lem:reconstruction} there is a run
		$\aConf \trans{\psi} \trans \al \aConf'$ such that
		$\cproj{\aConf'} \HH = \aConf_0$ with $\al_0 = \al$.
	 \item If $n>0$, we assume that the statement holds for all runs with
		less than $n$ transitions.
		If $\p[x]$ is involved in $\al_i$ with $0 \leq i < n$ then
		the thesis follows by inductive hypothesis.
		Let us therefore assume that $\p[x]$ is involved in $\al_n$ only.
		By repeated application of \cref{lem:reconstruction}, there is a
		run
		$\aConf \trans{\psi_1 \cdot \al_1'} \aConf_1' \cdots
		\trans{\psi_{n-1} \cdot \al_{n-1}'} \aConf_{n-1}' \trans{\psi_n
		  \cdot \al_n'} \aConf_n'$ in $\aCS$ such that $\al_i' = \al_i$
		and $\cproj{\aConf_i'} \HH = \aConf_i$ for each $1 \leq i \leq n$.
	 \end{itemize}
	 In both cases $\aConf$ reaches a configuration with a run
	 involving \p[x], which contradicts our assumption.
  \end{proof}
  \else
  \begin{proofsketch}
  The proof goes as the one of \cref{th:dfSeqGat} noticing that
  we have to reconstruct \quo{backward} the sequence of interactions.
  This exploits sequentiality and lock-freedom of $\aCS_1$ and
  $\aCS_2$ in order to guarantee the reconstruction when we
  \quo{cross} the two composed systems through the gateways.
\end{proofsketch}
\fi

We turn now our attention to strong lock-freedom.
In this case, as for deadlock freedom, $(\HH,\KK)$-composability
suffices for preservation by composition; we shall see that this is
not the case for lock freedom preservation.

\begin{restatable}[Strong lock freedom preservation]{theorem}{slfpresthm}\label{th:slfpres}
  Let $\aCS_1$ and $\aCS_2$ be two $(\HH,\KK)$-composable and strongly
  lock free systems.
  Then the composed system $\aCS_1 \connect \HH \KK \aCS_2$ is
  strongly lock free.
\end{restatable}
\iftr
  %\slfpresthm*
  \begin{proof}
	 By contradiction, let us assume $\aCS_1 \connect \HH \KK \aCS_2$
	 not to be strongly lock free.
	 This means that there are a reachable configuration $\aConf$, a
	 participant $\p[x]$, and a maximal run $\psi$ of
	 $\aCS_1 \connect \HH \KK \aCS_2$ such that $\aConf \trans{\psi}$
	 and \p[x] is not involved in any of those transitions.
	 By the first part of the proof of \cref{th:dfSeqGat}, there exist
	 two run $\cproj \aConf \HH \trans{\psi_1}$ and
	 $\cproj \aConf \KK \trans{\psi_2}$ of $\aCS_1$ and $\aCS_2$
	 respectively such that \p[x] is involve neither in $\psi_1$ nor
	 in $\psi_2$.
	 \begin{itemize}
	 \item In case $\psi$ is infinite, we get that either $\psi_1$ or
		$\psi_2$ is infinite, and hence maximal.
	 \item In case $\psi$ is finite it is possible to use the second
		part of the proof of \cref{th:dfSeqGat} to show that either
		$\psi_1$ or $\psi_2$ is maximal,
	 \end{itemize}
	 In both cases we get a contradiction of the hypothesis that
	 $\aCS_1$ and $\aCS_2$ are strong lock free.
  \end{proof}
  \else
  \begin{proofsketch}
  The proof is similar to the one of \cref{thmt@@lfpresthm} but for
  the use of strong lock freedom of $\aCS_1$ and $\aCS_2$ instead of
  their deadlock freedom.
\end{proofsketch}
\fi

%%% Local Variables:
%%% mode: latex
%%% TeX-master: "main"
%%% End:

\section{Conclusions and Future Work}\label{sec:conc}
We introduce an asymmetric synchronous semantics of communicating
systems which breaks the symmetry between senders and receivers.
In fact, our semantics decouples communication from choice resolution
as in standard semantics of communicating systems (and other models).
We then adapted the gateway composition mechanism defined
in~\cite{francoICEprev,BarbaneradH19} to our asymmetric semantics and
gave conditions for the preservation of some communication properties
under this notion of composition.

In \emph{contract automata}~\cite{bdft16,btp20} transitions express
\quo{requests} and \quo{offers} among participants.
The composition mechanism is based on \quo{trimming} a product of
contract automata according to relevant \emph{agreement} properties.
This yields controllers that preserve deadlocks.
Contract automata do not consider asymmetric synchronous semantics.
Our composition mechanism does not introduce orchestrators which,
under some conditions, can be avoided also for contract
automata~\cite{bdft16,btp20}.

Modular approaches to the development of concurrent systems can be
exploited even for systems designed using formalisms intrinsically
dealing with {\em closed} systems.
Indeed, given two systems, any two components -- one per system --
exhibiting {\em compatible} behaviours can be replaced by two coupled
forwarders (gateways) connecting the systems, as investigated
initially in~\cite{francoICEprev,BarbaneradH19} for an asynchronous
interaction model.
The investigation on the composition-by-gateways technique was shifted
in \cite{BLT20b} towards synchronous symmetric interactions.
We pushed a step forward such an investigation,
by considering {\em asymmetric} synchronous interactions.
Interestingly, deadlock freedom preservation in the synchronous
asymmetric case we consider does not require sequentiality of
gateways, like in the asynchronous case, and differently from the
synchronous symmetric case.
Notably, sequentiality is needed here for lock-freedom preservation,
but not for strong-lock freedom preservation.

While the path of investigation above is quite homogeneous, the
different analyses present some methodological differences.
For instance, \cite{BLT20b} considers also another form of
composition, where one single gateway (interacting with both the
composed systems) is used.
On the other side, \cite{BLT20b} focused only on deadlocks,
disregarding other properties we consider.
A first item of future research consist in filling the bits missing
due to the mismatches above.

%% RESULT In order to get a
%% complete view it is natural to prosecute it by taking into account
%% \begin{itemize}
%% \item properties like lock freedom for symmetric interactions:
%% \item the use of ``semi-direct''  composition as done

%% \end{itemize} 

A more challenging direction for future work is looking for refined
composition mechanisms in order to get preservation of relevant
properties under weaker conditions.

%Dual asymmetric approach (in intro?)

  % 
  %
%%% Local Variables:
%%% mode: latex
%%% TeX-master: "main"
%%% End:

%
\bibliographystyle{eptcs}
\bibliography{bib}

\end{document}

%%% Local Variables: 
%%% mode: latex
%%% TeX-master: t
%%% End: 